\documentclass[11pt,twoside,reqno,centertags]{amsart}
\copyrightinfo{}{}

\usepackage{amsmath,amsthm,amsfonts,amssymb}
\usepackage[mathscr]{euscript}
\usepackage{graphicx} 
\usepackage{epsfig}

 \usepackage{float}
 \floatstyle{ruled}
 \newfloat{myalgorithm}{htbp}{aux}[section]
 \floatname{myalgorithm}{Algorithm}

\newtheorem{thm}{\bf Theorem}[section] 
\newtheorem{rem}{\bf Remark}[section]

\newcommand{\R}{{\mathbb R}}

\newcommand{\llambda}{b}

\newcommand{\spa}{\operatorname{span}}

\newcommand{\diag}{\operatorname{diag}}
\newcommand{\diam}{\operatorname{diam}}
\newcommand{\vol}{\operatorname{vol}}

\newcommand{\supp}{\operatorname{supp}}
\newcommand{\tG}{{\widetilde{G}}}

\newcommand{\dH}{\;{\rm d}{\mathcal{H}}^{d-1}}  
\newcommand{\dL}{\;{\rm d}{\mathcal{L}}^{d}}
  \newcommand{\ds}{\dH}
\newcommand{\bigchi}{\ensuremath{\mathrm{\mathcal{X}}}}
\newcommand{\charfcn}[1]{\bigchi_{#1}}  
\newcommand{\Domain}{\Omega}
\newcommand{\Sm}{S^m}
\newcommand{\Smhom}{S^m_0}
\newcommand{\SmD}{S^m_D}
\newcommand{\Sml}{S^m_+}
\newcommand{\Smhoml}{S^m_{0,+}}
\newcommand{\SmDl}{S^m_{D,+}}
\newcommand{\Vh}{\underline{V}(\Gamma^m)}

\newcommand{\Wh}{W(\Gamma^m)}

\newcommand{\sigmaO}{o}
\newcommand{\nabs}{\nabla_{\!s}}

\newcommand{\id}{\rm id}
\newcommand{\ddt}{\frac{\rm d}{{\rm d}t}}

\newcommand{\Nbulk}{\vec{N}_{\Gamma,\Omega}^T}
 
\newcommand{\Mbulk}{M_{\Gamma,\Omega}}
\newcommand{\errorXx}{\|\vec{X} - \vec{x}\|_{L^\infty}}
\newcommand{\tinyplus}{\mbox{\tiny $+$}}
\newcommand{\errorUupl}{\|U - I^h\,u\|_{L^\infty,\tinyplus}}
\newcommand{\errorUu}{\|U - I^h\,u\|_{L^\infty}}

\def\conduct{\mathcal{K}}

\def\TTime{\overline{T}}
\def\epsilon{\varepsilon}

\def\vL{L\kern-0.08cm\char39}

\begin{document}
 
\title{ Finite-Element Approximation of \\ One-Sided Stefan Problems with \\
Anisotropic, Approximately Crystalline, Gibbs--Thomson Law }
\date{}
\maketitle     
 
\vspace{ -1\baselineskip}

{\small
\begin{center}
 {\sc John W. Barrett} \\
Department of Mathematics, 
Imperial College London, London, SW7 2AZ, UK \\[10pt]
 {\sc Harald Garcke} \\
Fakult{\"a}t f{\"u}r Mathematik, Universit{\"a}t Regensburg, 
93040 Regensburg, Germany \\[10pt]
 {\sc Robert N\"urnberg} \\
Department of Mathematics, 
Imperial College London, London, SW7 2AZ, UK \\[10pt]
\end{center}
}

\numberwithin{equation}{section}
\allowdisplaybreaks
 
 \smallskip

 \begin{quote}
\footnotesize
{\bf Abstract.}  
We present a finite-element approximation for the one-sided Stefan problem and
the one-sided Mullins--Sekerka problem, respectively. The problems feature a
fully anisotropic Gibbs--Thomson law, as well as kinetic undercooling.
Our approximation, which couples a parametric approximation of the moving
boundary with a finite-element approximation of the bulk quantities, 
can be shown to satisfy a stability bound, and it enjoys very
good mesh properties, which means that no mesh smoothing is necessary in
practice.
In our numerical computations we concentrate on the simulation of
snow crystal growth. On choosing realistic physical parameters, we are able to 
produce several distinctive types of snow crystal morphologies. In
particular, facet breaking in approximately crystalline evolutions can
be observed.
\end{quote}

\section{Introduction}

Pattern formation during crystal growth is one of the most fascinating
areas in physics and materials science. Furthermore, crystallisation
is a fundamental phase transition, and a good understanding is crucial for many
applications. In this paper we will concentrate on a mathematical
model based on the one-sided Stefan and Mullins--Sekerka problems, for which 
we will introduce a new numerical method of approximation. The numerical
solutions presented here are tailored for the description of snow crystal
growth. However, we note that with minor modifications our approach can be
used for other crystal growth scenarios (see \cite{dendritic}), which
in particular have applications in engineering as, for example, 
in the foundry industry.

The basic mathematical model for crystal growth involves diffusion
equations in the bulk phases together with complex conditions at the moving
boundary, which separates the phases. Depending on the application, either heat
diffusion or the diffusion of a solidifying species has to be
considered. If a pure, e.g.\ metallic, substance solidifies, then the basic
diffusion equation is the heat equation for the temperature  (see
\cite{Gurtin93,dendritic}), whereas for snow crystal growth the diffusion of
water molecules in the air is the main diffusion mechanism (see
\cite{Libbrecht05}). In the case that a binary metallic substance
solidifies, then models involving both heat and species diffusion
simultaneously, and which are coupled through the interface conditions, 
are considered, see e.g.\ \cite{Davis01}. 

At the moving boundary a conservation law
either for the energy or for the matter has to hold. In the case of heat 
diffusion,
one has to take into account the release of latent heat through the well-known
Stefan condition, which relates the velocity of
the interface to the temperature gradients at the interface, the latter being
proportional to the energy flux; see \cite{Gurtin93,Davis01,dendritic}. 
For snow crystal growth the continuity equation at the
interface relates its velocity to the particle flux
at the interface, which is given in terms of the gradient of the water
molecule density. In conclusion, mathematically very similar conditions
arise in both models.

Beside  the above-discussed continuity equation, another condition has
to be specified at the interface. In the case that heat diffusion is
the main driving force in the bulk, thermodynamical considerations lead
to the Gibbs--Thomson law with kinetic undercooling at the interface;
see \cite{Gurtin93,Davis01,dendritic}. This law relates the
undercooling (or superheating) at the interface to the 
curvature and the velocity of the interface. In the case of snow
crystal growth one has to consider a modified Hertz--Knudsen formula, 
which relates the supersaturation of the water molecules at the
interface to the curvature and velocity of the interface; see e.g.\ 
equations (1) and (23) in \cite{Libbrecht05}. 
The physics at the interface depends
on the local orientation of the crystal lattice in space, and hence
the parameters in the interface conditions discussed above are
anisotropic. In particular, the corresponding surface energy density
leads, through variational calculus, 
to an anisotropic version of curvature, 
which then appears in the moving boundary condition; see
\cite{Giga}. In addition, kinetic coefficients in the moving boundary
condition will also, in general, be anisotropic.

In the numerical experiments in Section~\ref{sec:6},
we focus on snow crystal growth, where the unknown will be
a properly scaled number density of the water molecules. 
However, straightforward
modifications, e.g.\ choosing different anisotropies, 
allow our approach to apply in the context of other crystal
growth phenomena. In addition, we note that our approach can be used for
many other moving boundary problems; see e.g.\ \cite{dendritic}.

In earlier work, the present authors introduced a new methodology
to approximate curvature-driven curve and surface evolution; see
\cite{triplej,triplejMC,gflows3d}. The method has the important feature that
mesh properties remain good during the evolution. In fact, for curves
semidiscrete versions of the
approach lead to polygonal approximations, where the vertices are equally
spaced throughout the evolution. This property is important, 
as most other approaches typically
lead to meshes which deteriorate during the evolution and often the
computation cannot be continued. The approach was first proposed
for isotropic geometric evolution equations, but later the method was
generalized to anisotropic situations, \cite{triplejANI,ani3d}, 
and to situations where an interface geometry was coupled to bulk fields, 
\cite{dendritic}. In most cases it was even possible to show stability
bounds. In \cite{dendritic} the two-sided Stefan and
Mullins--Sekerka problems, as a model for dendritic solidification, were
numerically studied. The physical parameters, such as the heat
conductivity, had to be chosen the same in both phases, whereas in this
paper we focus on the situation where diffusion can be restricted to
the liquid or gas phase, respectively. 
Hence, we need to study
a one-sided Stefan or Mullins--Sekerka problem. 
This has a significant impact on the
numerical analysis, and it necessitates novel computational techniques; see
e.g.\ Section~\ref{sec:41} below.
We remark that  
an anisotropic version of the one-sided Mullins--Sekerka problem is
relevant for snow crystal growth; see 
\cite{Libbrecht05} and \cite{jcg}. 
This, and the fact that the anisotropy in snow crystal
growth is so strong that nearly faceted shapes occur, makes this
application a
perfect situation in order to test whether our approach is suitable for
one-sided models for solidification. 

Before discussing our numerical approach and several phenomena, which
we wish to simulate, we formulate the anisotropic
one-sided Stefan and
Mullins--Sekerka problem with the Gibbs--Thomson law and kinetic
undercooling in detail. 
Let $\Domain\subset\mathbb{R}^d$ be a given domain, where $d=2$ or $d=3$. 
We now seek a time-dependent interface $(\Gamma(t))_{t\in[0,\overline{T}]}$,
$\Gamma(t)\subset\hspace*{-2pt}\subset\Domain$, 
which for all $t\in[0,\overline{T}]$ separates
$\Domain$ into a domain $\Omega_+(t)$, occupied by the liquid/gas,
and a domain $\Omega_{-}(t):=\Domain\setminus\overline\Omega_+(t)$, 
which is occupied by the
solid phase. See Figure~\ref{fig:sketch} for an illustration.

\begin{figure}[ht]
\def\myheight{0.1\textheight}
\center
    \includegraphics[scale=.3]{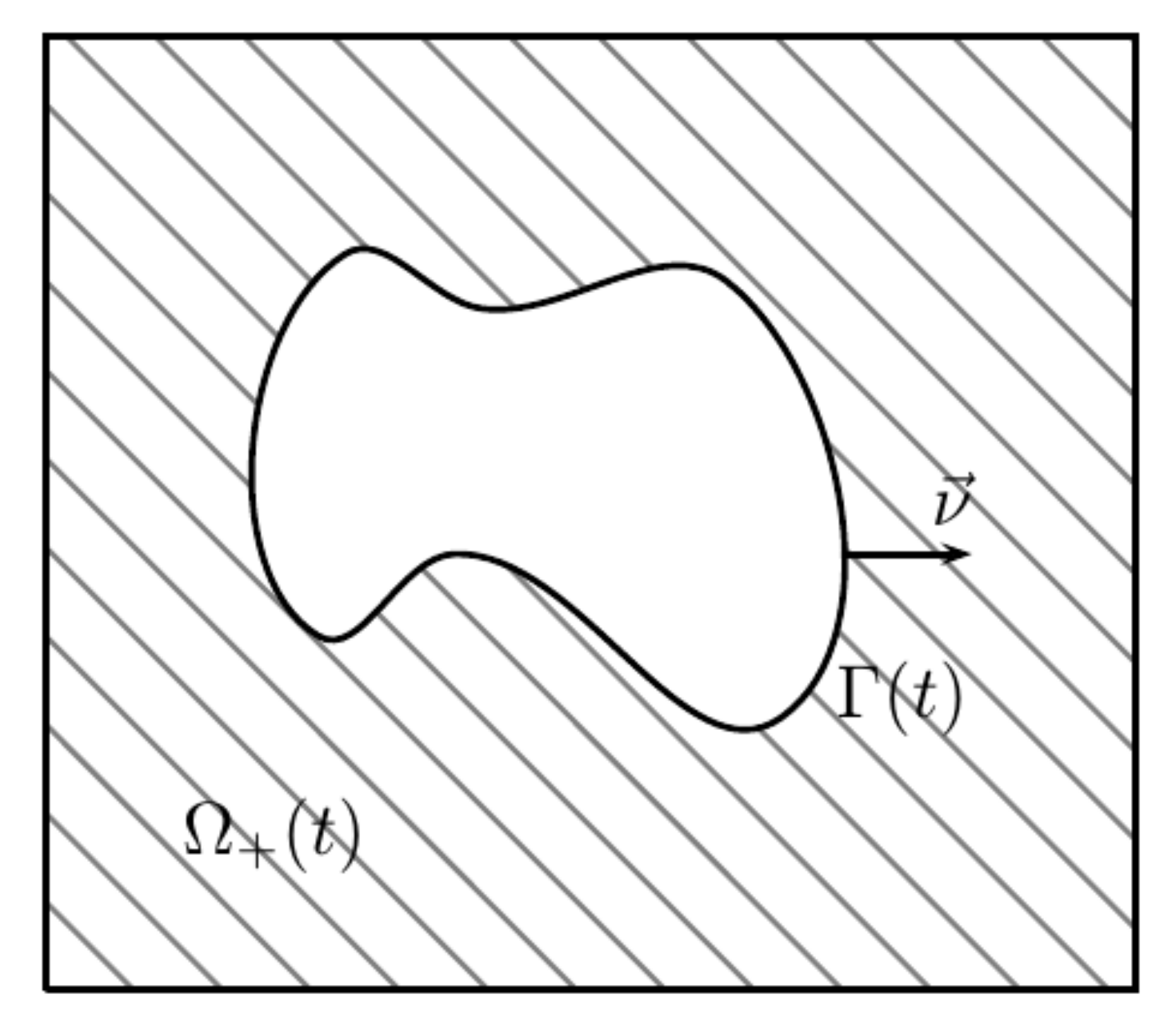}
\caption{The domain $\Omega$ in the case $d=2$.}
\label{fig:sketch}
\end{figure} 

For later use, we assume that
$(\Gamma(t))_{t\in [0,\TTime]}$ 
is a sufficiently smooth evolving
hypersurface parameterized by $\vec{x}(\cdot,t):\Upsilon\to\R^d$,
where $\Upsilon\subset \R^d$ is a given reference manifold, i.e., 
$\Gamma(t) = \vec{x}(\Upsilon,t)$. Then
$\mathcal{V} := \vec{x}_t \cdot\vec{\nu}$ is
the normal velocity of the evolving hypersurface $\Gamma$,
where $\vec\nu$ is the unit normal on $\Gamma(t)$ pointing into $\Omega_+(t)$.

We now need to find a time- and space-dependent function
$u$ defined in the liquid/gas region such that
$u(\cdot,t):\Omega_+(t)\to\mathbb{R}$ and the interface
$(\Gamma(t))_{t\in[0,\overline{T}]}$ fulfill the following conditions:
\begin{subequations}
\begin{alignat}{2}
\vartheta\,u_t - \conduct\,\Delta u & = f \qquad  \qquad \qquad &&\mbox{in } 
\Omega_+(t), \label{eq:1a} \\
\conduct\,\frac{\partial u}{\partial \vec\nu}
&=-\lambda\,{\mathcal{V}} \qquad &&\mbox{on } \Gamma(t), 
\label{eq:1b} \\  
\frac{\rho\,\mathcal{V}}{\beta(\vec \nu)} &= \alpha\,\varkappa_\gamma - 
a\,u \qquad &&\mbox{on } \Gamma(t), 
\label{eq:1c} \\  
u & = u_D \qquad &&\mbox{on } 
\partial \Omega , \label{eq:1d} \\
\Gamma(0) & = \Gamma_0 , \qquad 
\vartheta\,u(\cdot,0) = \vartheta\,u_0 \qquad && \mbox{in } 
\Omega_+(0)
\,; \label{eq:1e} 
\end{alignat}
\end{subequations}
where $\partial\Domain$ denotes the boundary of $\Domain$.
In addition, $f$ is a possible forcing term, while
$\Gamma_0 \subset\hspace*{-2pt}\subset \Omega$ and
$u_0 : \Omega_+(0) \to \R$ are given initial data. We always assume that
the solid region $\Omega_-(t)$ is compactly contained in $\Domain$. 

The unknown $u$ is, depending on the application, either a temperature
or a suitably scaled negative concentration. The orientation-dependent
function $\beta$ is a kinetic coefficient, $\gamma$ is the anisotropic
surface energy, and $\vartheta \geq 0$, $\conduct,$ $ \lambda,$ $ \rho,$ $ \alpha,$ and $ a > 0$
are constants whose physical significance is discussed in 
\cite{dendritic,jcg}.
For snow crystal growth (see 
\cite{jcg}), $-u$
is a suitably scaled concentration with $-u_D$ being the scaled supersaturation.  

It now remains to introduce the anisotropic mean curvature
$\varkappa_\gamma$. One obtains $\varkappa_\gamma$ as the first
variation of an anisotropic interface free energy
\begin{equation*}  
|\Gamma|_\gamma := \int_\Gamma \gamma(\vec\nu) \ds,
\end{equation*}
where $\gamma: \R^d \to \R_{\geq0}$,
with $\gamma(\vec{p})>0$ if $\vec{p}\ne \vec 0$, is the surface free
energy density which depends on the local orientation of the surface
via the normal $\vec{\nu}$;
and ${\mathcal{H}}^{d-1}$ denotes 
the $(d-1)$-dimensional Hausdorff measure in ${\R}^d$.
The function $\gamma$ is assumed to be
positively homogeneous of degree one, i.e., 
\begin{equation*} 
\gamma(\llambda\,\vec{p}) = \llambda\,\gamma(\vec{p}) \quad \forall \ 
\vec{p}\in \R^d,\ \forall\ \llambda \in {\mathbb R}_{>0}\quad \Rightarrow
\quad \gamma'(\vec{p}) \cdot \vec{p} = \gamma(\vec{p})
\quad \forall\ \vec{p}\in \R^d\setminus\{\vec0\},  
\end{equation*}
where $\gamma'$ is the gradient of $\gamma$.
The first variation of $|\Gamma|_\gamma$ is given by  (see
e.g.\ \cite{Giga} and \cite{ani3d}) 
\begin{equation} \label{eq:varkappa}
\varkappa_\gamma := - \nabs \cdot  \gamma'(\vec\nu),  
\end{equation}
where $\nabs .$ is the tangential divergence on $\Gamma$; i.e.,  we have
in particular that
\begin{equation} \label{eq:firstvar}
\ddt\, |\Gamma(t)|_\gamma = \ddt\, \int_{\Gamma(t)}
\gamma(\vec{\nu})\ds = -\int_{\Gamma(t)} 
\varkappa_\gamma\, \mathcal{V} \ds.
\end{equation}
We remark that in the isotropic case we have that 
\begin{equation} \label{eq:iso}
\gamma(\vec{p}) = 
\gamma_{iso}(\vec p) :=
|\vec{p}| \quad \forall\ \vec{p}\in \R^d, 
\end{equation}
which implies that $\gamma(\vec\nu)=1$; and so $|\Gamma|_{\gamma}$ 
reduces to $|\Gamma|$, the surface area of $\Gamma$. Moreover,
in the isotropic case the anisotropic mean curvature $\varkappa_\gamma$ reduces
to the usual mean curvature, i.e., 
to the sum of the principal curvatures~of~$\Gamma$.

In this paper we are interested in anisotropies of the form
\begin{equation} \label{eq:g1}
\gamma(\vec{p}) = \sum_{\ell=1}^L
\gamma_{\ell}(\vec{p}), \quad
\gamma_\ell(\vec{p}):= [{\vec{p} \cdot G_{\ell}\,\vec{p}}]^\frac12,
\end{equation}
where $G_{\ell} \in \R^{d\times d}$, for $\ell=1\to L$, 
are symmetric and positive definite matrices. 
We note that (\ref{eq:g1}) corresponds to the special choice $r=1$ for the
class of anisotropies
\begin{equation} \label{eq:g}
\gamma(\vec{p}) = \Big (\sum_{\ell=1}^L
[\gamma_{\ell}(\vec{p})]^r \Big )^{\frac1r},
\end{equation}
which has been considered by the authors in \cite{dendritic}. 
Numerical methods based on anisotropies of
the form (\ref{eq:g}) have first been considered in \cite{triplejANI} and
\cite{ani3d}, and there this choice enabled the authors to introduce 
unconditionally
stable fully discrete finite-element approximations for the anisotropic mean
curvature flow, i.e.,  (\ref{eq:1c}) with $a=0$, 
and other geometric evolution equations for an evolving interface $\Gamma$. 
Similarly, in \cite{dendritic}, the choice of
 anisotropies (\ref{eq:g}) leads to fully discrete 
approximations of the Stefan problem with very good stability properties.
We note that the simpler choice $r=1$, i.e., when $\gamma$ is of the form 
(\ref{eq:g1}), leads to a
finite-element approximation with a linear system to solve at each time level; 
see (\ref{eq:uHGa}--c). In three space dimensions, the choice (\ref{eq:g1}) 
only gives rise to a relatively small class of anisotropies, which is why the
authors introduced the more general (\ref{eq:g}) in \cite{ani3d}.
For the modelling of snow crystal growth, however, the choice (\ref{eq:g1}) 
is sufficient, and we will stick to this case in the present paper, 
but we point out that using the method from
\cite{dendritic} the approach in this paper can be easily generalized
to the more general class of anisotropies in (\ref{eq:g}). 

We now give some examples for anisotropies of the form (\ref{eq:g1}), which
later on will be used for the numerical simulations in this paper. For the
visualizations we will use the Wulff shape, \cite{Wulff1901}, defined by
\begin{equation}
{\mathcal{W}} := \{ \vec{p}\in\R^d : \vec{p} \cdot \vec{q} \leq \gamma(\vec{q})
\quad \forall\ \vec{q} \in \R^d \} . \label{eq:Wulff}
\end{equation}
Here we recall that 
the Wulff shape $\mathcal{W}$ is known to be the solution of an 
isoperimetric problem; i.e., the boundary of $\mathcal{W}$ is the minimizer of 
$|\cdot|_{\gamma}$ in the class of all  
surfaces enclosing the same volume; see e.g.\ \cite{FonsecaM91}.

Let 
$l_\epsilon(\vec{p}) := \left[ \epsilon^2\,|\vec{p}|^2 
                                 + p_1^2\,(1-\epsilon^2) \right]^{\frac12}$
for $\epsilon>0$.
\begin{figure}[t]
\def\myheight{0.1\textheight}
\center
   \includegraphics[scale=.2]{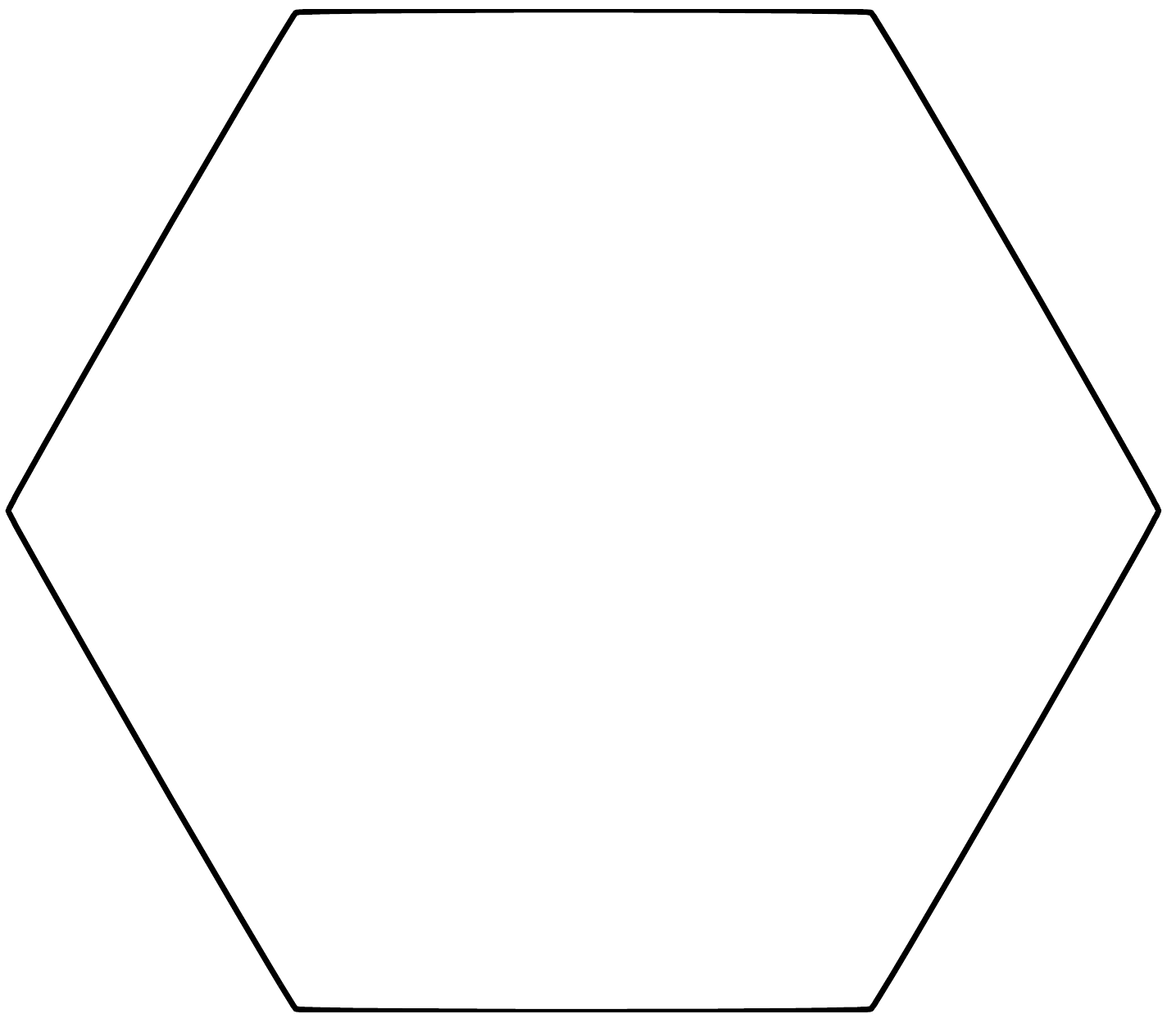}
\caption{Wulff shape in $\R^2$ for (\ref{eq:hexgamma2d}) with  
$\epsilon=0.01$ and $\theta_0=0$.}
\label{fig:Wulff2d}
\end{figure}%
Then a hexagonal anisotropy in $\R^2$ can be modelled with the choice
\begin{equation} \label{eq:hexgamma2d}
\gamma(\vec{p}) = \gamma_{hex} (\vec{p}) := \sum_{\ell = 1}^3
l_\epsilon(R(\theta_0 + \tfrac{\ell\,\pi}3)\,\vec{p}),
\end{equation}
where $R(\theta)$ denotes a clockwise rotation through the angle $\theta$,
and $\theta_0 \in [0,\frac\pi3)$ is a parameter that rotates the orientation of
the anisotropy in the plane. The Wulff shape of (\ref{eq:hexgamma2d}) 
for $\epsilon=0.01$ and $\theta_0=0$ is shown in Figure~\ref{fig:Wulff2d}.

In order to define anisotropies of the form (\ref{eq:g1}) in $\R^3$, we
introduce the rotation matrices
$$
R_{1}(\theta):=\left(
\begin{array}{rrr} \cos\theta & \sin\theta&0 \\
-\sin\theta & \cos\theta & 0 \\ 0 & 0 & 1 \end{array}\!\! \right)
\ \text{ and } \
 R_{2}(\theta):=\left(
\begin{array}{rrr} \cos\theta & 0 & \sin\theta \\
0 & 1 & 0 \\ -\sin\theta & 0 & \cos\theta \end{array}\!\! \right) .
$$
Then 
\begin{equation} \label{eq:hexgamma3d}
\gamma(\vec{p}) = 
l_\epsilon(R_2(\tfrac{\pi}2)\,\vec{p}) +
\sum_{\ell = 1}^3
l_\epsilon(R_1(\theta_0+\tfrac{\ell\,\pi}3)\,\vec{p})
\end{equation}
is one such example, where $\theta_0 \in [0,\frac\pi3)$ again rotates the
anisotropy in the $x_1$-$x_2$ plane. The anisotropy (\ref{eq:hexgamma3d}) has
been used by the authors in their numerical simulations of anisotropic
geometric evolution equations in \cite{ani3d,clust3d,ejam3d}, as well as for
their dendritic solidification computations in \cite{dendritic}. 
Its Wulff shape for $\epsilon=0.01$ is shown on the left of
Figure~\ref{fig:Wulff3d}.

\begin{figure}[t]
\def\mywidth{0.15\textwidth}
\center
   \includegraphics[angle=0,width=\mywidth]{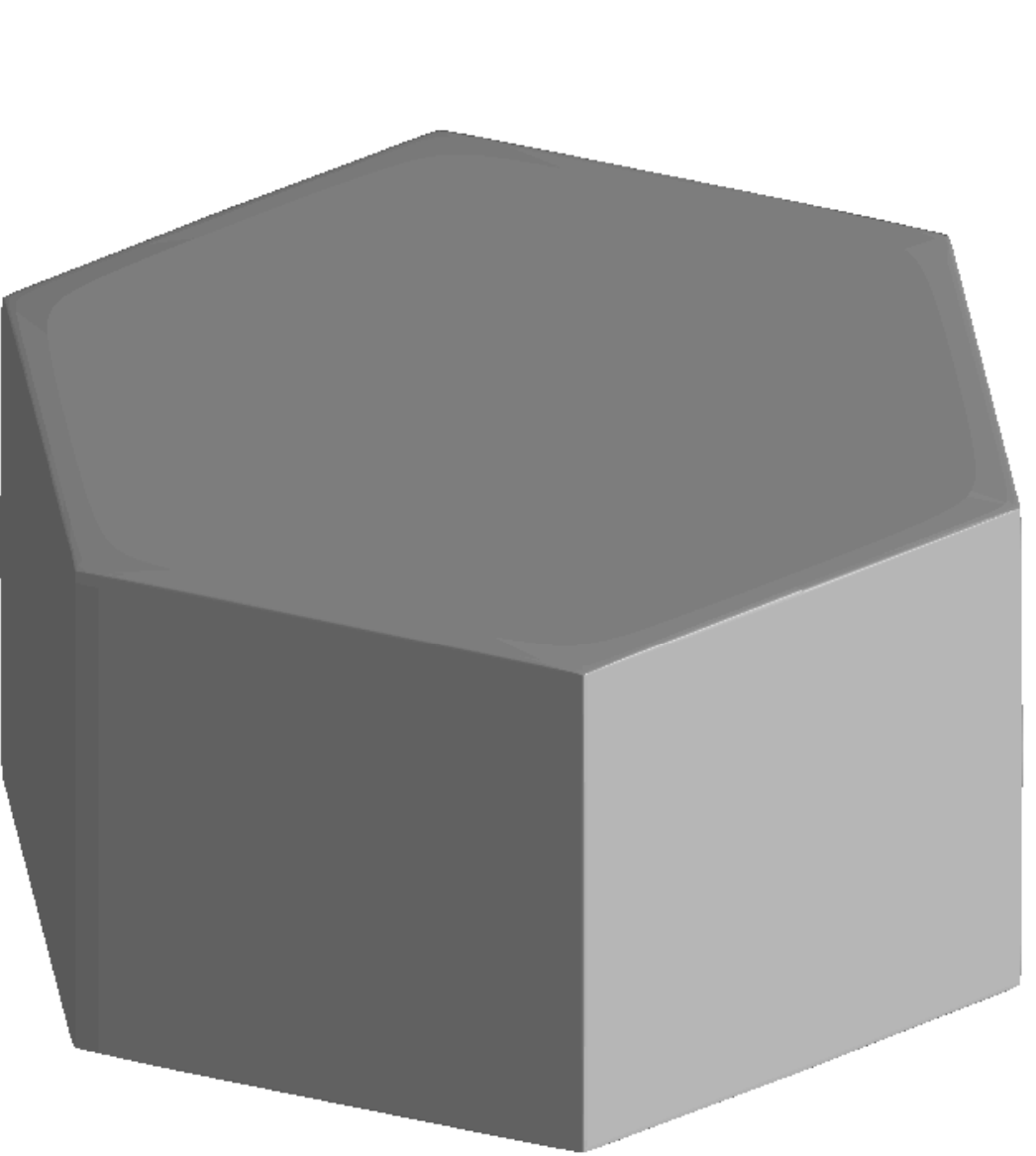} 
\qquad\qquad
   \includegraphics[angle=0,width=\mywidth]{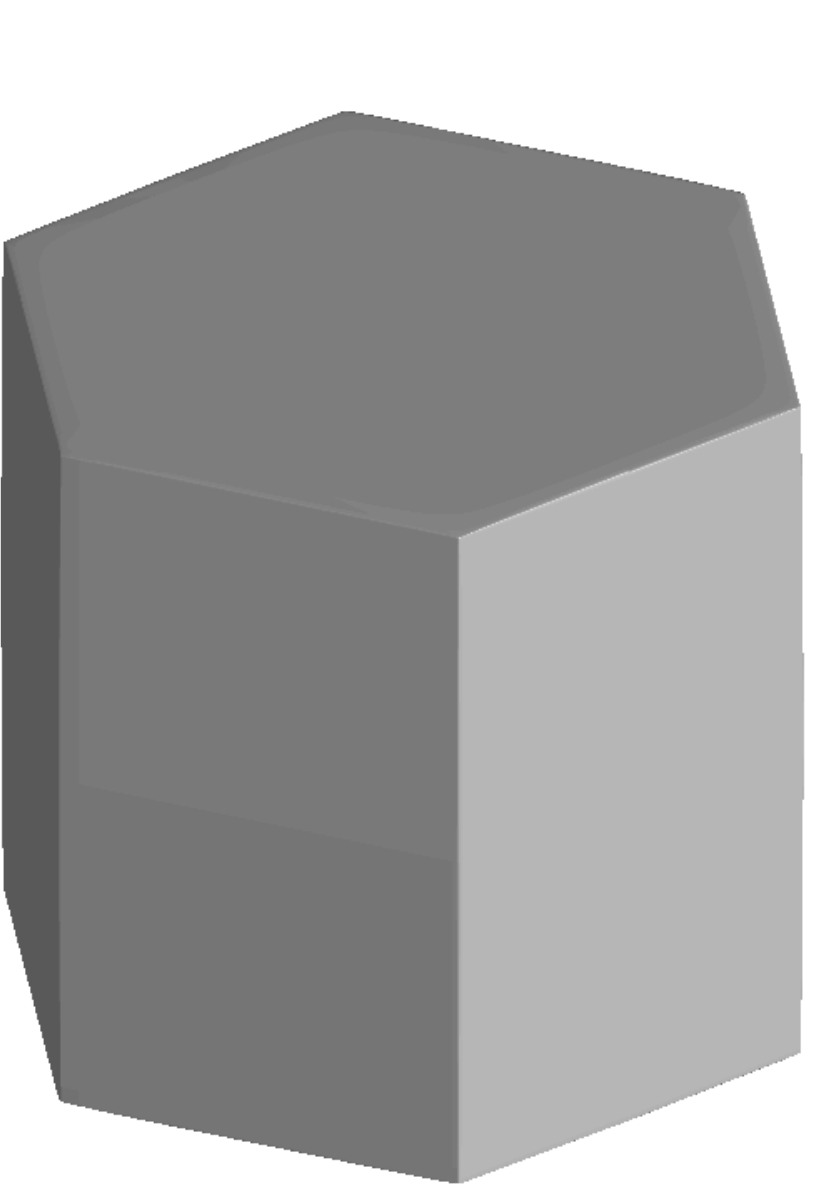} 
\caption{Scaled Wulff shape in $\R^3$ for (\ref{eq:hexgamma3d}) 
with $\epsilon=0.01$ (left).
Scaled Wulff shape in $\R^3$ for (\ref{eq:L44}) 
with $\epsilon=0.01$ (right).}
\label{fig:Wulff3d}
\end{figure}

A small modification of (\ref{eq:hexgamma3d}), which is more relevant 
for the simulation of snow flake growth, is
\begin{equation} \label{eq:L44}
\gamma(\vec{p}) = 
\gamma_{hex}(\vec{p}) := 
l_\epsilon(R_2(\tfrac{\pi}2)\,\vec{p}) + \tfrac{1}{\sqrt{3}}\,
\sum_{\ell = 1}^3
l_\epsilon(R_1(\theta_0 + \tfrac{\ell\,\pi}3)\,\vec{p}).
\end{equation}
Its Wulff shape for $\epsilon=0.01$ is
shown on the right of Figure~\ref{fig:Wulff3d}.
We note that the Wulff shape of (\ref{eq:L44}), in contrast to
(\ref{eq:hexgamma3d}), for $\epsilon\to0$ approaches a prism where every face
has the same distance from the origin. In other words, for
(\ref{eq:L44}) 
the surface energy densities in the basal and prismal directions are the same.
We remark that if $\mathcal{W}_0$ denotes the Wulff shape of
(\ref{eq:L44}) with $\epsilon=0$, 
then the authors in \cite{GravnerG09} used the scaled
Wulff shape $\tfrac12\,\mathcal{W}_0$ as the building block in their 
cellular automata algorithm. In addition, we observe that the choice
(\ref{eq:L44}) agrees well with data reported in e.g.\ 
\cite[p.~148]{PruppacherK97}, although there the ratio of basal to prismal 
energy is computed as $\gamma^{\rm B} / \gamma^{\rm P} \approx 0.92 < 1$.

In addition, we consider an example of (\ref{eq:g1}), where $L=2$ and
$G_1 = \diag(1,1,\epsilon^2)$, 
$G_2=\gamma_{\rm TB}^2\,\diag(\epsilon^2,\epsilon^2,1)$, so that
it approximates for small $\epsilon$ the anisotropy
\begin{equation} \label{eq:giga}
\gamma(\vec{p}) = \gamma_{\rm TB}\,|p_3| + (p_1^2+p_2^2)^\frac12,
\end{equation}
as considered in e.g.\ \cite{GigaR04}. See
Figure~\ref{fig:frankgiga}, where we show its Wulff shape for
$\gamma_{\rm TB} = 1$ and $\gamma_{\rm TB} = 0.1$ for
$\epsilon=10^{-2}$.
\begin{figure}[t]
\def\myheight{0.15\textheight}
\center
   \includegraphics[angle=0,height=\myheight]{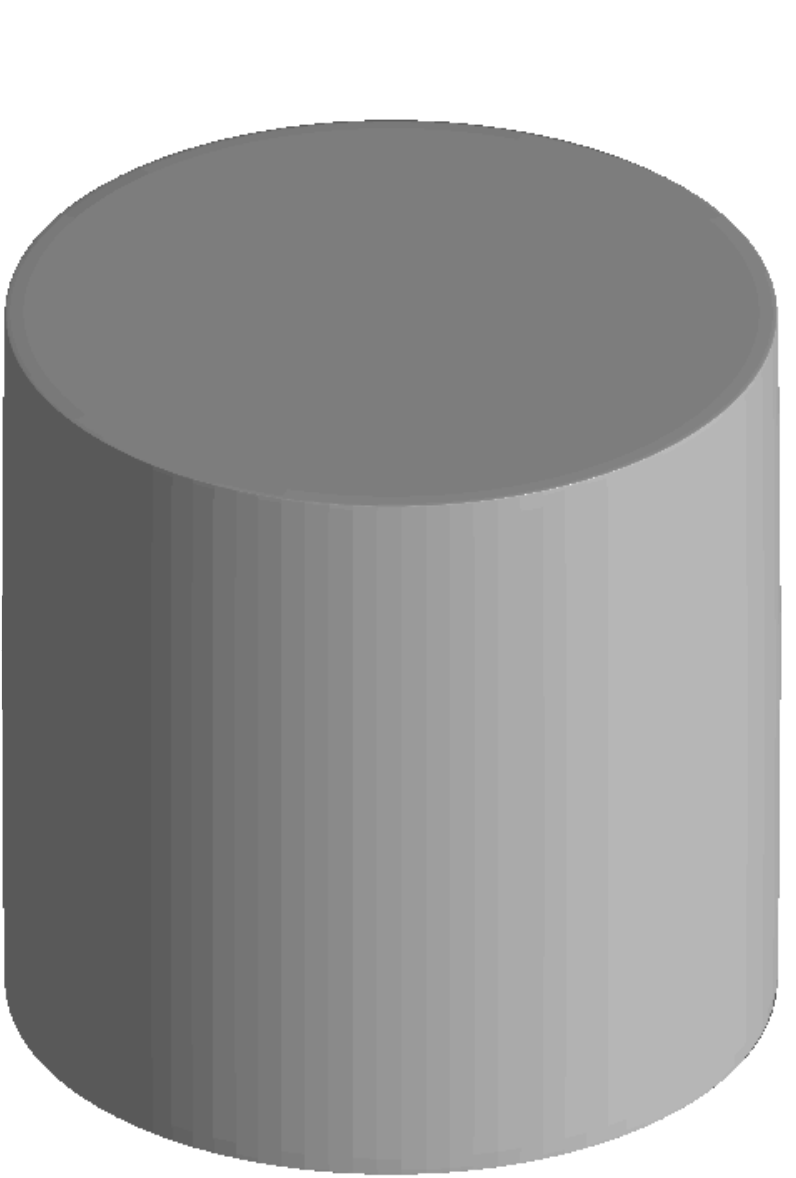} \qquad \quad
   \includegraphics[angle=0,height=\myheight]{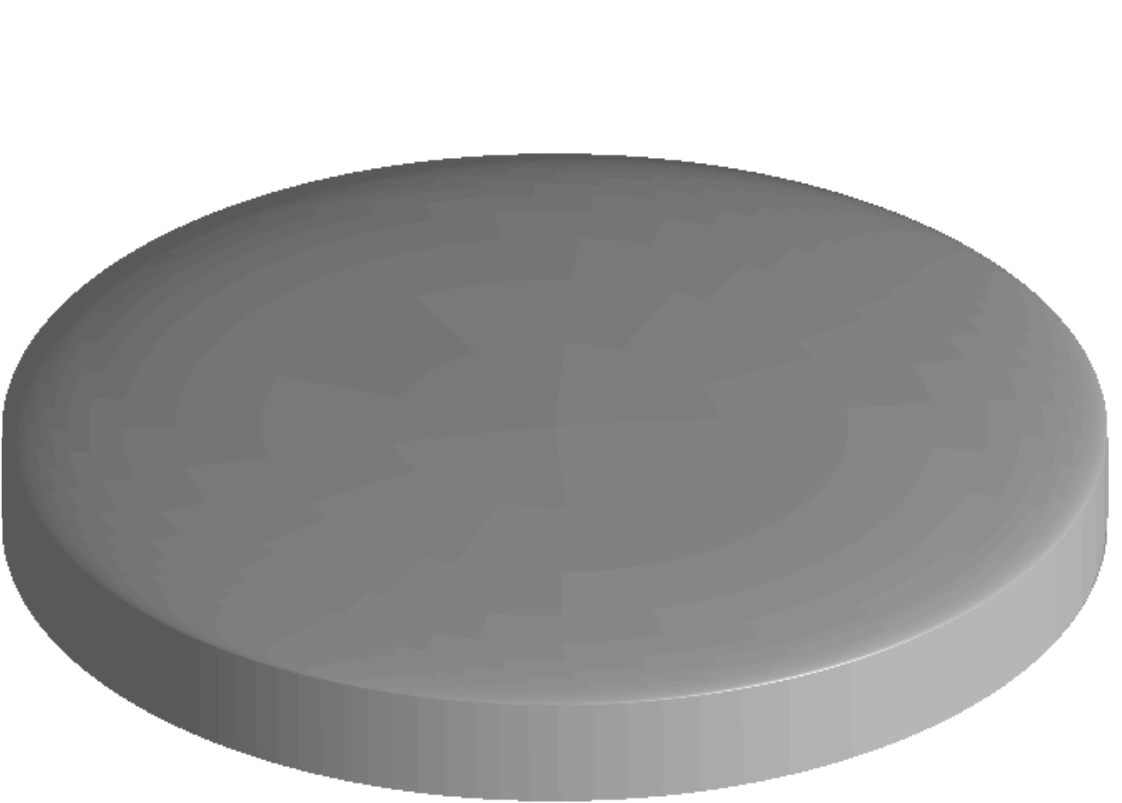} 
\caption{Scaled Wulff shape for the approximation of (\ref{eq:giga}) with 
$\gamma_{\rm TB} = 1$ (left) and $\gamma_{\rm TB} = 0.1$ (right) for
$\epsilon=10^{-2}$.}
\label{fig:frankgiga}
\end{figure}%
We note the Wulff shape of (\ref{eq:giga}) is given by a cylinder with basal
radius one and height $2\,\gamma_{\rm TB}$. Hence its ratio of
height to basal diameter is $\gamma_{\rm TB}$.

More examples
of anisotropies of the form (\ref{eq:g}) can be found in 
\cite{triplejANI,ani3d,clust3d}. 
Let us briefly discuss why the novel way that we deal with the 
anisotropy makes it possible to compute evolution equations
resulting from nearly crystalline surface energies,
i.e., when the Wulff shape has sharp corners and flat parts. 
Energies of the form 
(\ref{eq:hexgamma2d}) and (\ref{eq:hexgamma3d}) have as building blocks 
simple quadratic expressions, and for $\epsilon$ close to zero they reduce to 
crystalline surface energies. 
It is now possible to discretize these energies,
such that the resulting discrete equations  are linear and such that 
they allow for a stability bound; compare Theorem~\ref{thm:stab} below
and \cite{triplejANI,ani3d}. Stability        
bounds for nearly crystalline energies are very difficult to obtain. The     
fact that
we obtain stability bounds for small $\epsilon$, and hence nearly            
crystalline energies, together with the good mesh properties                   
of our discrete approximation of the interface                                  
enable us to perform numerical computations in situations which involve        
nearly crystalline surface energies. In this context let us mention that the    
good mesh quality results from a tangential redistribution of the mesh, where
the tangential velocity arises naturally from the discretization of 
a variational formulation of (\ref{eq:varkappa}).

Crystal growth in general, and snow crystal growth in particular, is a
highly anisotropic mechanism. In snow crystal growth the morphologies that
appear depend strongly on the environment and, in particular, on the
temperature and the supersaturation,
which influence the values of
$\alpha$ and $u_D$, respectively, in (\ref{eq:1a}--e). 
This can be seen in the famous
Nakaya diagram; see Figure~\ref{fig:libbrecht}. Depending on
these parameters, 
either solid prisms, needles, thin plates, hollow columns
or dendrites appear in snow crystal growth. The anisotropy of the
surface energy can be responsible for the hexagonal symmetry, but
probably also an anisotropic $\beta$ has an influence on the shapes
appearing in snow crystal growth; see e.g.\ \cite{Libbrecht05} and
\cite{Yokoyama93}. Depending on the size of the crystal, either the kinetic
anisotropy or the anisotropy in the surface energy dominates; see
\cite{YokoyamaS92} or \cite{KobayashiG01}.
It is one of the goals of this paper to study the influence of the
anisotropies in $\beta$ and $\gamma$ on the growth morphologies. It
was discussed in \cite{Libbrecht05} that the kinetic coefficient can
vary drastically between the directions of the two basal hexagonal facets
and the directions of the six prismal facets. Depending on the environmental
conditions either flat crystals or column crystals appear; see 
Figure~\ref{fig:libbrecht}. 

 \begin{figure}[t]
\center
   \includegraphics[angle=0,height=6.5cm]{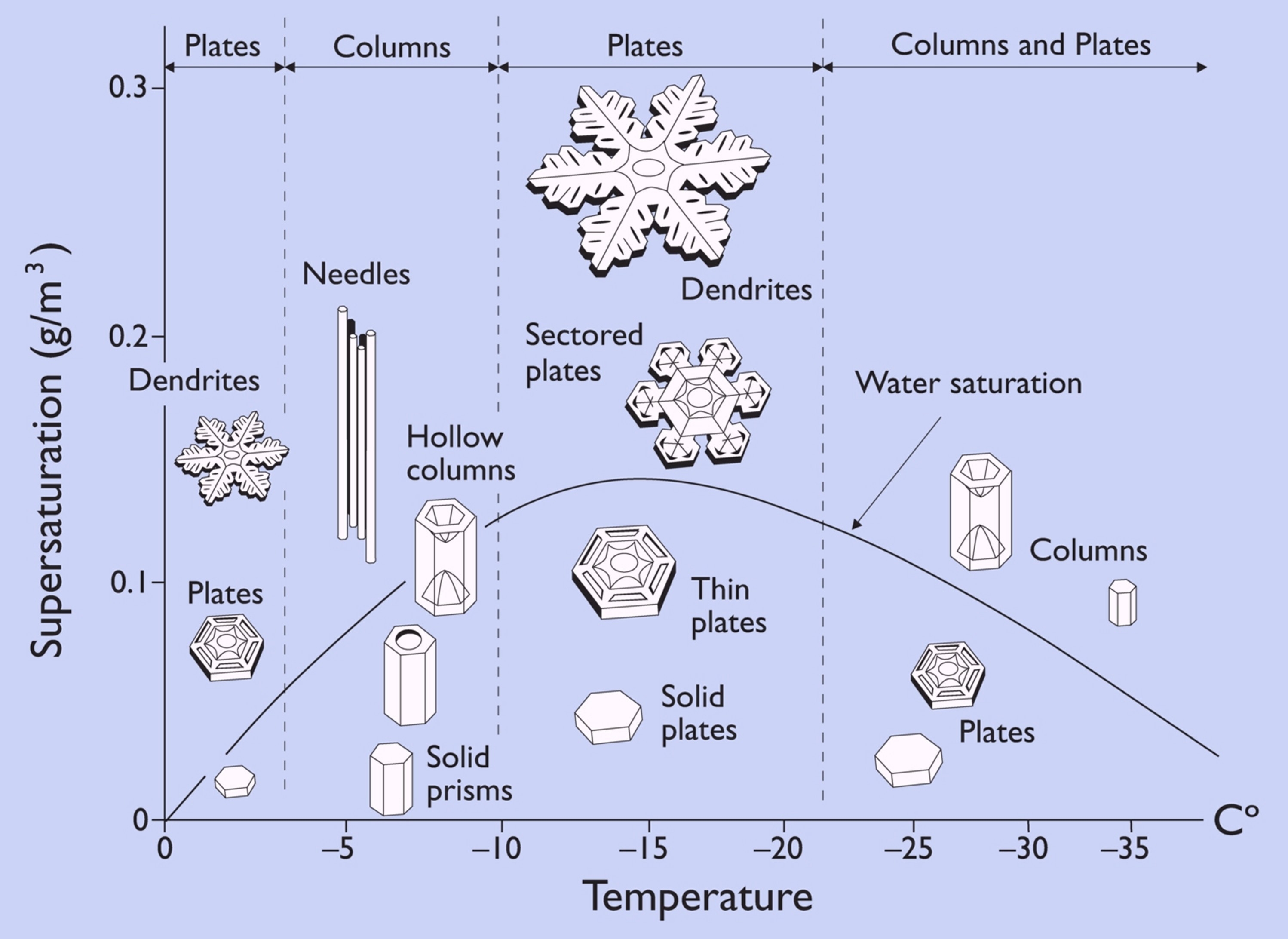}
\caption{The Nakaya diagram illustrates which snow crystal forms
  appear at different temperatures and supersaturations. This figure
  is taken from \cite{Libbrecht05}.}
\label{fig:libbrecht}
\end{figure}

A derivation of the set of equations (\ref{eq:1a}--e) can be found in
\cite{Gurtin93} and \cite{Davis01}. The evolution of interfaces driven by
anisotropic curvature has been studied by many authors, and we refer
to \cite{Giga} for an overview. For the full problem (\ref{eq:1a}--e),
to the knowledge of the authors, no existence result seems to be
known,
although there are results for two-sided variants; see
\cite{Luckhaus90} in the isotropic case and 
\cite{GarckeS11} in the anisotropic case. We remark that also cases where the
Wulff shape is crystalline  
have been studied. In this case nonlocal curvature quantities have to
be considered, and the geometric equation (\ref{eq:1c}) for the
interface is of singular diffusion type. 
Then local existence to (\ref{eq:1a}--e) has been obtained for 
anisotropies where the Wulff shape is a prism with polygonal base,
for a restricted class of $\Gamma_0$ and on assuming that no facet bending or
facet breaking occurs; see \cite{Gorka08,Gorka08a}.
In addition, it was shown in \cite{GigaR04}  
that self-similar solutions for (\ref{eq:1a}--e) exist in a
situation  where the Wulff shape is a cylinder. We will attempt to
compute such self-similar solutions in Section~\ref{sec:6}.

In snow crystal growth often flat parts appear, and in some cases they
become unstable and break; see Figure~\ref{fig:libbrecht},
\cite{Libbrecht05} and \cite{GondaY82}. Only recently  have researchers 
studied facet breaking from a mathematical point of view.  
The three-dimensional case  
has been considered in 
\cite{BellettiniNP99} and \cite{GigaG98} for geometrical evolution
equations---see also the numerical studies in \cite{ani3d}. A
full crystalline model of solidification facet breaking has, so far,
only been studied analytically in \cite{GigaR06} and numerically in
\cite{dendritic}. Clearly from the Nakaya diagram, facet
breaking is an important issue in snow crystal growth, and we will
study this aspect numerically in Section~\ref{sec:6}.

Numerical approaches for dendritic solidification  that are based on
the Stefan problem with the Gibbs--Thomson law are often restricted to two
space dimensions; see e.g.\ \cite{YokoyamaS92,Schmidt98} 
and \cite{BanschS00}, where in the latter article 
the coupling to a fluid flow is also considered. 
The first implementations in three
space dimensions are due to Schmidt  (see \cite{Schmidt93,Schmidt96}), 
and the present authors later proposed a stable variant of Schmidt's approach
which could also handle the anisotropy in a more physically rigorous
way; see \cite{dendritic}. We also would like to refer to the
fascinating results on snow crystal growth, which were established in
\cite{GravnerG09}, using a cellular automata model. They were
able to compute a large variety of forms, which resemble snow crystals
in nature, even though the overall approach does not
stem from basic physical conservation laws and it is difficult to
relate its parameters to physical quantities. 

The outline of this paper is as follows.
In Section~\ref{sec:2} we introduce a weak formulation of the one-sided 
Stefan problem and the one-sided Mullins--Sekerka problem, which we consider in 
this paper.
Based on this weak formulation, we then 
introduce our numerical approximation of these problems
in Section~\ref{sec:3}. In particular, on utilizing techniques 
from \cite{dendritic}, we derive a coupled finite-element approximation for the interface evolution and the diffusion equation in
the bulk. Moreover, we show well-posedness and stability results for our
numerical approximation. Solution methods for the discrete equations and
implementation issues are discussed in Section~\ref{sec:4}.
In addition, a non-dimensionalization of a model for snow crystal growth from
\cite{Libbrecht05}, which allows us to derive physically relevant parameter
ranges, is recalled in Section~\ref{sec:nondim}.
Finally, we present several numerical experiments, including
simulations of snow crystal formations in three space dimensions, in
Section~\ref{sec:6}.

\setcounter{equation}{0}
\section{Weak formulation} \label{sec:2}

In this section we state a weak formulation of the problem
(\ref{eq:1a}--e) and derive a formal energy bound. Recall that
$\vartheta \geq 0$ and $\conduct, \lambda, \rho, \alpha, a > 0$
are physical parameters that are discussed in more detail 
in \cite{dendritic} and in \cite{jcg}. 

We introduce the function spaces
\begin{align*}
S_{0,+}(t) & := \{ \phi \in H^1(\Omega_+(t)) : \phi = 0 \ \mbox{ on } 
\partial\Omega \} \\ 
\quad\mbox{and}\quad
S_{D,+}(t) & := \{ \phi \in H^1(\Omega_+(t)) : \phi = u_D \ \mbox{ on } 
\partial\Omega\}
.  
\end{align*}
In addition, we define
$
\underline{V} := H^1(\Upsilon,\R^d) $ and $
W := H^1(\Upsilon,\R),  
$
where we recall that $\Upsilon$ is a given reference manifold.
A possible weak formulation of 
\mbox{(\ref{eq:1a}--e)},
which utilizes the novel weak representation of $\varkappa_\gamma\,\vec\nu$
introduced in \cite{ani3d}, is then given as follows.
Find time-dependent functions $u$, $\vec{x}$, and $\varkappa_\gamma$ such
that $u(\cdot,t)\in S_{D,+}(t)$, $\vec{x}(\cdot,t) \in \underline{V}$, 
$\varkappa_\gamma(\cdot,t)\in W$, and
\begin{subequations}
\begin{align}
& \vartheta\,(u_t, \phi)_+ + \conduct\,(\nabla\,u, \nabla\,\phi)_+ 
- (f, \phi)_+ \nonumber \\ & \hspace{5mm}
 = - \conduct\,\int_{\Gamma(t)} 
\frac{\partial u}{\partial \vec\nu} \, \phi \ds =
 \lambda\int_{\Gamma(t)} \vec{x}_t \cdot \vec\nu\,\phi \ds 
\quad\forall\ \phi\in S_{0,+}(t),
\label{eq:weaka} \\
&\rho\,\int_{\Gamma(t)} \frac{\vec{x}_t \cdot \vec\nu\,\chi}{\beta(\vec \nu)}\;\dH
= \int_{\Gamma(t)} \left[\alpha\,\varkappa_\gamma - a\,u\right] \chi \ds
\quad \forall\ \chi \in W
,\label{eq:weakb} \\
&\int_{\Gamma(t)} \varkappa_\gamma\,\vec\nu \cdot \vec\eta \ds
+ \langle \nabs^\tG\,\vec{x} , \nabs^\tG\,\vec\eta \rangle_\gamma
 = 0  \quad\forall\ \vec\eta \in \underline{V}
\label{eq:weakc}
\end{align}
\end{subequations}
hold for almost all times $t\in (0,\TTime]$, as well as the initial conditions
(\ref{eq:1e}). 
Here $(\cdot,\cdot)_+$ denotes the $L^2$-inner product on $\Omega_+(t)$.

We note that, for convenience, we have adopted a slight abuse
of notation
in (\ref{eq:weaka}--c). Here, and throughout  this paper, we will identify functions defined
on the reference manifold $\Upsilon$ with functions defined on $\Gamma(t)$.
In particular, we identify $v \in W$ with $v \circ \vec{x}^{-1}$ on 
$\Gamma(t)$, where we recall that $\Gamma(t) = \vec{x}(\Upsilon,t)$, and we
denote both functions simply as $v$.
For example, $\vec{x} \equiv \vec{\rm id}$ is also 
the identity function on $\Gamma(t)$.
In addition, we have
introduced the shorthand notation 
$\langle \nabs^\tG\,\cdot, \nabs^\tG\,\cdot \rangle_\gamma$ for the 
inner product defined in \cite{ani3d}. 
In particular, on recalling (\ref{eq:g1}), we define the symmetric positive-definite matrices $\tG_\ell$ with the associated inner products 
$(\cdot,\cdot)_{\tG_\ell}$ on ${\mathbb R}^d$ by
\begin{equation*} 
\tG_{\ell} := [\det G_\ell]^{\frac1{2}}\,[G_\ell]^{-1}
\ \mbox{ and } \ 
(\vec{v},\vec{w})_{\tG_\ell} = \vec{v} \cdot \tG_\ell\,\vec{w}
\quad \forall \ \vec{v},\,\vec{w} \in {\mathbb R}^d, 
\quad \ell = 1 \to L. 
\end{equation*}
Then we have that
\begin{equation}
\langle \nabs^\tG \vec\chi, \nabs^\tG \vec\eta \rangle_\gamma :=
\sum_{\ell=1}^{L} \int_{\Gamma(t)} 
(\nabs^{\tG_{\ell}} \vec\chi , \nabs^{\tG_{\ell}} \vec\eta)
_{\tG_{\ell}} \gamma_{\ell}(\vec\nu) \ds 
\quad \forall\ \vec\chi , \vec\eta \in \underline V ,
\label{eq:ani3d}
\end{equation}
where 
\begin{equation*}
(\nabs^{\tG_\ell}\,\vec\eta , \nabs^{\tG_\ell}\,\vec\chi)_{\tG_\ell} 
:= \sum_{j=1}^{d-1} (\partial_{\vec{t}_{j}^{(\ell)}}\, \vec\eta, 
\partial_{\vec{t}_{j}^{(\ell)}}\, \vec\chi)_{\tG_\ell} 
\end{equation*}
with $\{\vec{t}_{1}^{(\ell)},\ldots,\vec{t}_{d-1}^{(\ell)}\}$ 
being an orthonormal basis with respect to the $\tG_\ell$ inner product for
the tangent space of $\Gamma(t)$; see \cite{ani3d} for further details. 

Assuming, for simplicity, that the Dirichlet data 
$u_D$ is constant, we can establish the
following formal a priori bound. 
Choosing $\phi=u-u_D$ in (\ref{eq:weaka}), 
$\chi=\frac\lambda{a}\,\vec{x}_t \cdot \vec\nu$ in (\ref{eq:weakb}), and 
$\vec\eta=\frac{\alpha\,\lambda}a\,\vec{x}_t$ in (\ref{eq:weakc})  
we obtain, on using the identities
\begin{equation}
\ddt\, \int_{\Omega_+(t)} g \dL
= \int_{\Omega_+(t)} g_t \dL 
- \int_{\Gamma(t)} g\,\mathcal{V} \ds, 
\label{eq:dtvol}
\end{equation}
with $\mathcal{L}^d$ denoting the Lebesgue measure in $\R^d$ 
(see e.g.\ \cite{DeckelnickDE05}) and 
\begin{equation}
\ddt\,|\Gamma(t)|_\gamma = \ddt\,\int_{\Gamma(t)} \gamma(\vec\nu) \ds=
\langle \nabs^\tG\,\vec{x},\nabs^\tG\,\vec{x}_t \rangle_\gamma 
\label{eq:dtarea}
\end{equation}
(see \cite{ani3d}), that
\begin{align}
\label{eq:testD}
& \ddt \Big (\frac\vartheta2\,|u-u_D|^2_{\Omega_+} + 
\frac{\alpha\,\lambda}a\,
|\Gamma(t)|_\gamma
-\lambda\,u_D\,\vol(\Omega_+(t)) \Big ) + \conduct\,(\nabla\,u, \nabla\,u)_+
  \\ & \hspace{5mm} 
+\frac{\lambda\,\rho}{a}\, 
\int_{\Gamma(t)} \frac{\mathcal{V}^2}{\beta(\vec \nu)} \ds
= - \frac\vartheta2\,\int_{\Gamma(t)} \mathcal{V}\,|u-u_D|^2 \ds  +
(f, u - u_D)_+ , \notag 
\end{align}
where $|\cdot|_{\Omega_+}$ denotes the $L^2$-norm on $\Omega_+(t)$. In
particular, the bound (\ref{eq:testD}) for $\vartheta>0$ gives a formal a priori control on
$u$ and $\Gamma(t)$ only if $\mathcal{V} \geq 0$, i.e., when the solid region is
not shrinking.

\setcounter{equation}{0}
\section{Finite-element approximation} \label{sec:3}
Let $0= t_0 < t_1 < \dots < t_{M-1} < t_M = \TTime$ be a
partitioning of $[0,\TTime]$ into possibly variable time steps 
$\tau_m := t_{m+1} -
t_{m}$, $m=0\to M-1$. We set $\tau := \max_{m=0\to M-1}\tau_m$.
First we introduce standard finite-element spaces of piecewise-linear
functions on $\Omega$.

Let $\Omega$ be a polyhedral domain. For $m\geq0$, let $\mathcal{T}^m$ 
be a regular partitioning of $\Omega$ into disjoint open simplices, so that 
$\overline{\Omega}=\cup_{\sigmaO^m\in\mathcal{T}^m}\overline\sigmaO^m$. 
Let $J_\Omega^m$ be the number of elements in $\mathcal{T}^m$, so that
$\mathcal{T}^m=\{ \sigmaO^m_j : j = 1 \to J^m_\Omega\}$.
Associated with $\mathcal{T}^m$ is the finite-element space 
\begin{equation} \label{eq:Sh}
 S^m := \{\chi \in C(\overline{\Omega}) : \chi\!\mid_{\sigmaO^m} \mbox{ is
linear } \forall\ \sigmaO^m \in \mathcal{T}^m\} \subset H^1(\Omega). 
\end{equation}
Let $K_\Omega^m$ be the number of nodes of $\mathcal{T}^m$, and  
let $\{\vec{p}^m_{j}\}_{j=1}^{K_\Omega^m}$ be the coordinates of these nodes.
Let $\{\phi_{j}^m\}_{j=1}^{K_\Omega^m}$ be the standard basis functions 
for $\Sm$.
We introduce $I^m:C(\overline{\Omega})\to S^m$, the interpolation
operator, such that $(I^m \eta)(\vec{p}_k^m)= \eta(\vec{p}_k^m)$ 
for $k=1\to K_\Omega^m$. 
A discrete semi-inner product on $C(\overline{\Omega})$ is then defined by 
$ 
(\eta_1,\eta_2)^h_m := (I^m[\eta_1\,\eta_2],1) ,
$
with the induced semi-norm given by 
$|\eta|_{\Omega,m}  := [\,(\eta,\eta)^h_m\,]^{\frac{1}{2}}$
for $\eta \in C(\overline{\Omega})$.

The test and trial spaces for our finite-element approximation of the bulk
equation (\ref{eq:weaka}) are then defined by
\begin{equation} \label{eq:Sh0}
 \Smhom:= \{\chi \in \Sm : \chi = 0 \ \mbox{ on $\partial\Omega$} \}  
\ \mbox{ and } \
 \SmD:= \{\chi \in \Sm : \chi = I^m u_D \ \mbox{ on $\partial\Omega$} \} 
,
\end{equation}
where in the definition of $\SmD$ we allow for  
$u_D \in H^\frac12(\partial\Omega) \cap C(\partial\Omega)$.
Without loss of generality, let
$\{\phi_{j}^m\}_{j=1}^{K_{\Omega,D}^m}$ be the standard basis functions 
for $\Smhom$.

The parametric finite-element spaces in order to approximate $\vec{x}$ and
$\varkappa_\gamma$ in (\ref{eq:weaka}--c), are defined as follows.
Similarly to \cite{gflows3d}, 
we introduce the following discrete spaces, based on the seminal
paper \cite{Dziuk91}.
Let $\Gamma^{m}\subset\R^d$ be a 
\mbox{$(d-1)$}-dimensional {\em polyhedral surface},
i.e., a union of non-degenerate $(d-1)$-simplices with no hanging vertices
(see \cite[p.~164]{DeckelnickDE05} for $d=3$), approximating the
closed surface $\Gamma(t_m)$, $m=0 \to M$.
In particular, let $\Gamma^m=\bigcup_{j=1}^{J^m_\Gamma} 
\overline\sigma^m_j$,
where $\{\sigma^m_j\}_{j=1}^{J^m_\Gamma}$ is a family of mutually disjoint open 
$(d-1)$-simplices 
with vertices $\{\vec{q}^m_k\}_{k=1}^{K^m_\Gamma}$. 
Then for $m =0 \to M-1$, let
\begin{align*}
\Vh & := \{\vec\chi \in C(\Gamma^m,\R^d):\vec\chi\!\mid_{\sigma^m_j}
\mbox{ is linear}\ \forall\ j=1\to J^m_\Gamma\} \\ & 
=: [\Wh]^d \subset H^1(\Gamma^m,\R^d),
\end{align*}
where $\Wh \subset H^1(\Gamma^m,\R)$ is the space of scalar continuous
piecewise-linear functions on $\Gamma^m$, with 
$\{\chi^m_k\}_{k=1}^{K^m_\Gamma}$ denoting the standard basis of $\Wh$.
For later purposes, we also introduce 
$\pi^m: C(\Gamma^m,\R)\to \Wh$, the standard interpolation operator
at the nodes $\{\vec{q}_k^m\}_{k=1}^{K^m_\Gamma}$,
and similarly $\vec\pi^m: C(\Gamma^m,\R^d)\to \Vh$.
Throughout this paper, we will parameterize the new closed surface 
$\Gamma^{m+1}$ over $\Gamma^m$, with the help of a parameterization
$\vec{X}^{m+1} \in \Vh$, i.e., $\Gamma^{m+1} = \vec{X}^{m+1}(\Gamma^m)$.
Moreover, for $m \geq 0$, we will often identify $\vec{X}^m$ with 
$\vec{\rm id} \in \Vh$, the identity function on $\Gamma^m$. 

For scalar and vector  
functions $v,w\in L^2(\Gamma^m,\R^{(d)})$ 
we introduce the $L^2$ inner product $\langle\cdot,\cdot\rangle_m$ over
the current polyhedral surface $\Gamma^m$  
as follows:
\begin{equation*}  
\langle v, w\rangle_m  := \int_{\Gamma^m} v \cdot w \, \dH.
\end{equation*}
Here and throughout this paper, $\cdot^{(\ast)}$ denotes an
expression with or
without the superscript $\ast$, and similarly for subscripts.
If $v$ and $w$ are piecewise continuous, with possible jumps
across the edges of $\{\sigma_j^m\}_{j=1}^{J^m_\Gamma}$,
we introduce the mass lumped inner product
$\langle\cdot,\cdot\rangle^h_m$ as
\begin{equation}
\langle v, w \rangle^h_m  :=
\tfrac1d \sum_{j=1}^{J^m_\Gamma} |\sigma^m_j|\sum_{k=1}^{d} 
(v \cdot w)((\vec{q}^m_{j_k})^-),
\label{def:ip0}
\end{equation}
where $\{\vec{q}^m_{j_k}\}_{k=1}^{d}$ are the vertices of $\sigma^m_j$,
and where
we define $v((\vec{q}^m_{j_k})^-):=
\underset{\sigma^m_j\ni \vec{p}\to \vec{q}^m_{j_k}}{\lim}\, v(\vec{p})$.
Here $|\sigma^m_j| = \frac{1}{(d-1)!}\,
|(\vec{q}^m_{j_2}-\vec{q}^m_{j_1}) \land \cdots\land
(\vec{q}^m_{j_{d}}-\vec{q}^m_{j_1})|$ is the measure of $\sigma^m_j$,
where $\land$ is the standard wedge product on $\R^d$.
Moreover, we set $|\cdot|_{m(,h)}^2 := \langle \cdot, \cdot \rangle^{(h)}_m$.

Given $\Gamma^m$, we 
let $\Omega^m_+$ denote the exterior of $\Gamma^m$ and let
$\Omega^m_-$ denote the interior of $\Gamma^m$, so that
$\Gamma^m = \partial \Omega^m_- = \overline\Omega^m_- \cap 
\overline\Omega^m_+$. 
In addition, we define the piecewise-constant unit normal 
$\vec{\nu}^m$ to $\Gamma^m$ by
\begin{equation*}
\vec{\nu}^m_j := \vec{\nu}^m\!\mid_{\sigma^m_j} :=
\frac{(\vec{q}^m_{j_2}-\vec{q}^m_{j_1}) \land \cdots \land
(\vec{q}^m_{j_d}-\vec{q}^m_{j_1})}{|(\vec{q}^m_{j_2}-\vec{q}^m_{j_1}) \land
\cdots \land (\vec{q}^m_{j_d}-\vec{q}^m_{j_1})|},
\end{equation*}
where we have assumed that the vertices $\{\vec{q}^m_{j_k}\}_{k=1}^d$
of $\sigma_j^m$ are ordered such that $\vec\nu^m:\Gamma^m\to\R^d$ 
induces an orientation on $\Gamma^m$, and such that $\vec\nu^m$ points into
$\Omega^m_+$.

Before we can introduce our approximation to (\ref{eq:weaka}--c), we have to
introduce the notion of a vertex normal on $\Gamma^m$. We will combine this
definition with a natural assumption that is needed in order to show existence
and uniqueness, where applicable, for the introduced finite-element 
approximation.
\begin{itemize}
\item[$(\mathcal{A})$]
We assume for $m=0\to M-1$ that $|\sigma^m_j| > 0$ for all $j=1\to J^m_\Gamma$,
and that $\Gamma^m \subset \overline\Omega$.
For $k= 1 \rightarrow K^m_\Gamma$, let
$\Xi_k^m:= \{\sigma^m_j : \vec{q}^m_k \in \overline \sigma^m_j\}$
and set
\begin{equation*}
\Lambda_k^m := \cup_{\sigma^m_j \in \Xi_k^m} \overline \sigma^m_j
 \qquad \mbox{and} \qquad
\vec\omega^m_k := \frac{1}{|\Lambda^m_k|}
\sum_{\sigma^m_j\in \Xi_k^m} |\sigma^m_j|\;\vec{\nu}^m_j.  
\end{equation*}
Then we further assume that $\vec\omega^m_k\not=\vec0$, $k=1\to K^m_\Gamma$, 
and that $\dim \spa\{\vec{\omega}^m_k\}_{k=1}^{K^m_\Gamma} = d$,
$m=0\to M-1$.
\end{itemize}

Given the above definitions, we also introduce the piecewise-linear 
vertex normal function
\begin{equation*}  
\vec\omega^m := \sum_{k=1}^{K^m_\Gamma} \chi^m_k\,\vec\omega^m_k \in \Vh ,
\end{equation*}
and note that 
\begin{equation} 
\langle \vec{v}, w\,\vec\nu^m\rangle_m^h =
\langle \vec{v}, w\,\vec\omega^m\rangle_m^h \qquad \forall\ \vec{v} \in 
\Vh,\ w \in \Wh .
\label{eq:NI}
\end{equation}

Following \cite{BarrettE82}, we consider the following unfitted finite-element
approximation of (\ref{eq:weaka}--c).
First we need to introduce the appropriate discrete trial and test function
spaces. To this end, let $\Omega^{m,h}_+$ be an approximation to
$\Omega^m_+$ and set $\Omega^{m,h}_- := \Domain \setminus
\overline \Omega^{m,h}_+$. We stress that $\Omega^{m,h}_+$ need not necessarily
be a union of elements from $\mathcal{T}^m$. Moreover, it need not hold that
$\Gamma^m \subset \overline \Omega^{m,h}_+$.
Then we define the finite-element spaces
\begin{align}
\Sml & := \{ \chi \in \Sm : \chi(\vec p^m_j) = 0 \text{ if }
\supp \phi^m_j \subset \overline\Omega^{m,h}_- \}, \nonumber \\ 
\quad \Smhoml & := \Smhom \cap \Sml,
\quad \SmDl := \SmD \cap \Sml.
\label{eq:Sml}
\end{align}

Our finite-element approximation is then given as follows.
Let $\Gamma^0$, an approximation to $\Gamma(0)$, 
and, if $\vartheta>0$, $U^0\in S^0_D$ be given.
For $m=0\to M-1$, find $U^{m+1} \in \SmDl$, $\vec{X}^{m+1}\in\Vh$, and 
$\kappa^{m+1}_\gamma \in \Wh$ such that
for all $\varphi \in \Smhoml$, $\chi \in \Wh$, and  $\vec\eta \in \Vh$,
\begin{subequations}
\begin{align}
&
\vartheta \Big (\frac{U^{m+1}-U^m}{\tau_m}, \varphi  \Big )^h_{m,+} +
 \conduct\,(\nabla\,U^{m+1}, \nabla\,\varphi)_{m,+} \nonumber \\ & \hspace{2cm}
 - \lambda\,
 \Big \langle \pi^m \Big [
\frac{\vec{X}^{m+1}-\vec{X}^m}{\tau_m}  \cdot \vec\omega^m  \Big ], \varphi
 \Big \rangle_m   
= (f^{m+1}, \varphi)^h_{m,+}
, \label{eq:uHGa}\\
& \rho \Big \langle 
[\beta(\vec\nu^m)]^{-1}\,\frac{\vec{X}^{m+1}-\vec{X}^m}{\tau_m},
\chi\,\vec\omega^m  \Big \rangle_m^h - \alpha\, \langle \kappa^{m+1}_\gamma, 
\chi\rangle_m^h
+ a \,\langle U^{m+1}, \chi \rangle_m =0  
, \label{eq:uHGb} \\
& \langle \kappa^{m+1}_\gamma\,\vec\omega^m, \vec\eta \rangle_m^h + 
\langle \nabs^\tG\,\vec{X}^{m+1}, \nabs^\tG\,\vec\eta \rangle_{\gamma,m}
= 0  
, \label{eq:uHGc} 
\end{align}
\end{subequations}
and set $\Gamma^{m+1} = \vec{X}^{m+1}(\Gamma^m)$.
Here we define
\begin{subequations}
\begin{align} \label{eq:intml}
(\nabla\,\chi , \nabla\,\varphi)_{m,+} & :=
\int_{\Omega^{m,h}_+} \nabla\,\chi  \cdot  \nabla\,\varphi \dL \nonumber \\ &
= \sum_{j=1}^{J^m_\Omega} \tfrac{|o^m_j \cap \Omega^{m,h}_+|}{|o^m_j|}\, 
\int_{o^m_j} \nabla\,\chi  \cdot  \nabla\,\varphi \dL
\quad \forall\ \chi, \varphi \in \Sm,
\end{align}
and, in a similar fashion, 
\begin{equation} \label{eq:inthml}
(\chi, \varphi)_{m,+}^h := \sum_{j=1}^{J^m_\Omega}
\tfrac{|o^m_j \cap \Omega^{m,h}_+|}{|o^m_j|}
\, \int_{o^m_j} I^m [\chi \,\varphi ] \dL
\quad \forall\ \chi, \varphi \in \Sm.
\end{equation}
\end{subequations}
For later use we note that it follows immediately from (\ref{eq:Sml}) 
and (\ref{eq:intml}) that
\begin{equation} \label{eq:Omegah}
(\nabla \varphi, \nabla \varphi)_{m,+} > 0 \qquad \forall\ \varphi \in 
\Smhoml \setminus \{0\}.
\end{equation}
In addition, we set $f^{m+1}(\cdot) := f(\cdot,t_{m+1})$,
where we assume for convenience that $f$ is defined on $\Omega$.
In addition, for $\vartheta>0$, $U^0\in S^0_D$ is given by
$U^{0} = I^0[u_0]$, where $u_0 \in C(\overline\Omega)$ is an appropriately
defined extension to $\overline\Omega$ of the given initial data from  
(\ref{eq:1e}).

Moreover, 
$\langle \nabs^\tG\, \cdot, \nabs^\tG\, \cdot \rangle_{\gamma,m}$ 
in (\ref{eq:uHGc}) is the discrete inner product defined by
\begin{align}
\langle \nabs^\tG\, \vec\chi, \nabs^\tG\, \vec\eta \rangle_{\gamma,m} & :=
 \sum_{\ell=1}^{L} \int_{\Gamma^m} 
(\nabs^{\tG_{\ell}}\,\vec{\chi},\nabs^{\tG_{\ell}}\,
\vec\eta)_{\tG_{\ell}} \,\gamma_{\ell}(\vec\nu^m) \ds \nonumber \\ & 
\hspace{6cm} \forall\ \vec\chi,\,\vec\eta \in \Vh.
\label{eq:dip}
\end{align}
Note that (\ref{eq:dip}) is a natural discrete analogue of (\ref{eq:ani3d}); 
see \cite{ani3d} for details. This choice of discretization will lead to 
unconditionally
stable approximations in certain situations; see Theorem~\ref{thm:stab}, below.

\begin{rem} \label{rem:theta}
\rm
We note that for $\vartheta>0$ the approximation {\rm (\ref{eq:uHGa}--c)} is
only meaningful when the discrete solid region does not shrink. To see this,
assume that the discrete solid region shrinks at some time step, so that
$\Smhoml \setminus S^{m-1}_{0,+} \not= \emptyset$ for some $m > 1$. Assume for
simplicity that $\mathcal{T}^m = \mathcal{T}^{m-1}$, so that $S^m = S^{m-1}$.
Now let $\phi^m_j \in \Smhoml \setminus S^{m-1}_{0,+}$, 
which means that the node $\vec p^m_j$ is an active node in $\Smhoml$, 
but was inactive in 
$S^{m-1}_{0,+}$; i.e.,  $U^m(\vec p^m_j) = 0$ since $U^m \in S^{m-1}_{D,+}$.
Here the value $U^m(\vec p^m_j) = 0$ is arbitrary, and has no physical meaning.
Crucially, however, this value will play a role on the discrete level, since
choosing $\varphi = \phi^m_j$ in {\rm (\ref{eq:uHGa})}, and noting that
$(\phi^m_j, \phi^m_j)^h_{m,+} >0$, means that $U^{m+1}$ will depend on
$U^m(\vec p^m_j)$.

In practice this technical restriction is not very relevant, since in physically
meaningful simulations the solid region typically 
never shrinks. Here we also recall that
the formal energy bound {\rm (\ref{eq:testD})}, for $\vartheta >0$, is also only meaningful,
when the solid region is not shrinking.
\end{rem}

\begin{thm} \label{thm:stab}
Let the assumption $( \mathcal{A} )$ hold. Then
there exists a unique solution  
$(U^{m+1}, \vec{X}^{m+1}, \kappa^{m+1}_\gamma) 
\in \SmDl\times \Vh \times \Wh$ to {\rm (\ref{eq:uHGa}--c)}. 
Let $u_D \in \R$ and define
\begin{equation}
\mathcal{E}^m(U^m,\vec{X}^m) := 
\frac\vartheta2\,|U^{m} - u_D|_{\Omega,m,+}^2 +
\frac{\alpha\,\lambda}a\,|\Gamma^m|_\gamma,
\label{eq:E}
\end{equation}
where $|\cdot|_{\Omega,m,+} := [(\cdot, \cdot)^h_{m,+}]^\frac12$.
Then the solution to {\rm (\ref{eq:uHGa}--c)} satisfies
\begin{align}
& \mathcal{E}^m(U^{m+1}, \vec{X}^{m+1})
+\lambda\,u_D\,\langle \vec{X}^{m+1}-\vec{X}^m, \vec\omega^m\rangle_m^h
+ \frac\vartheta2\,|U^{m+1} - U^m|_{\Omega,m,+}^2
\nonumber \\ & \hspace{0.2cm}
+ \tau_m\,\conduct\,(\nabla\,U^{m+1}, \nabla\,U^{m+1})_{m,+}
+\tau_m\,\frac{\lambda\,\rho}{a} \Big |[\beta(\vec\nu^m)]^{-\frac12}\,
\frac{\vec{X}^{m+1}-\vec{X}^m}{\tau_m}  \cdot \vec\omega^m \Big |_{m,h}^2
\nonumber \\ & \hspace{2cm}
\leq \mathcal{E}^m(U^m,\vec{X}^m) 
+ \tau_m\,(f^{m+1}, U^{m+1} - u_D)^h_{m,+}.
\label{eq:stab}
\end{align}
\end{thm}

\begin{proof}[\bf Proof.]
As the system (\ref{eq:uHGa}--c) is linear, existence follows from uniqueness.
In order to establish the latter, we consider the following system:
Find $(U, \vec{X}, \kappa_\gamma) \in \Smhoml\times \Vh \times \Wh$ such that
\begin{subequations}
\begin{align}
& \vartheta\left(U, \varphi \right)^h_{m,+} + 
\tau_m\,\conduct\,(\nabla\, U, \nabla\,\varphi)_{m,+} - \lambda 
\left\langle \pi^m  [\vec{X}  \cdot \vec\omega^m ], \varphi 
\right\rangle_m = 0 
\quad \forall\ \varphi \in \Smhoml, \label{eq:proofa}\\
& \frac\rho{\tau_m}\left\langle 
[\beta(\vec\nu^m)]^{-1}\,\vec{X},
\chi\,\vec\omega^m \right\rangle_m^h - \alpha \langle \kappa_\gamma, 
\chi\rangle_m^h
+ a \langle U, \chi \rangle_m =0 \quad \forall\ \chi \in \Wh,
\label{eq:proofb} \\
& \langle \kappa_\gamma\,\vec\omega^m, \vec\eta \rangle_m^h + 
\langle \nabs^\tG\,\vec{X}, \nabs^\tG\,\vec\eta \rangle_{\gamma,m}
= 0 \quad \forall\ \vec\eta \in \Vh.
\label{eq:proofc} 
\end{align}
\end{subequations}
Choosing $\varphi=U$ in (\ref{eq:proofa}), 
$\chi = \frac\lambda{a}\,\pi^m[\vec{X} \cdot \vec\omega^m]$ in (\ref{eq:proofb}),
and 
$\vec\eta=\frac{\alpha\,\lambda}a\,\vec{X}$ in (\ref{eq:proofc})
yields, on noting (\ref{eq:NI}), that
\begin{align}
& \vartheta\,(U, U)^h_{m,+} + 
\tau_m\,\conduct\,(\nabla\, U, \nabla\, U)_{m,+}
+\frac{\lambda\,\rho}{\tau_m\,a}\,\left|[\beta(\vec\nu^m)]^{-\frac12}\,
\vec{X} \cdot \vec\omega^m\right|_{m,h}^2 \nonumber \\ & \qquad\qquad\qquad
+ \frac{\alpha\,\lambda}a\, \langle \nabs^\tG\,\vec{X}, \nabs^\tG\,\vec{X} 
\rangle_{\gamma,m} 
=0. \label{eq:proof2}
\end{align}
It immediately follows from (\ref{eq:proof2}) 
and (\ref{eq:Omegah}) that $U = 0 \in \Smhoml$.
In addition, on recalling that $\alpha,\lambda > 0$, it holds that
$\vec{X}\equiv \vec{X}_c \in \R^d$. Together with (\ref{eq:proof2}), 
for $U=0$, and the assumption $(\mathcal{A})$ this
immediately yields that $\vec{X} \equiv \vec0$, while
(\ref{eq:proofb}) with $\chi=\kappa_\gamma$  
implies that $\kappa_\gamma \equiv 0$; compare Theorem 3.1 in \cite{ani3d}.
Hence there exists a unique solution
$(U^{m+1}, \vec{X}^{m+1}, \kappa^{m+1}_\gamma) \in \SmDl\times \Vh \times \Wh$.

It remains to establish the bound (\ref{eq:stab}). 
Let $\charfcn{\mathcal{A}}$ denote the characteristic function of a set
$\mathcal{A}$.
Choosing $\varphi=U^{m+1}-u_D\,I^m\,\charfcn{\overline\Omega^{m,h}_+}$ 
in (\ref{eq:uHGa}), 
$\chi = \frac\lambda{a}\,\pi^m[({\vec{X}^{m+1}-\vec{X}^m}) \cdot \vec\omega^m]$ 
in (\ref{eq:uHGb}), and
$\vec\eta=\frac{\alpha\,\lambda}a\,({\vec{X}^{m+1}-\vec{X}^m})$ 
in (\ref{eq:uHGc}) yields that
\begin{align*}
& \vartheta\,(U^{m+1}-U^m, U^{m+1} - u_D)^h_{m,+} + 
\tau_m\,\conduct\,(\nabla\, U^{m+1}, \nabla\, U^{m+1})_{m,+}
\nonumber \\ & \hspace{1cm}
+ \frac{\alpha\,\lambda}a\,
 \langle \nabs^\tG\,\vec{X}^{m+1}, \nabs^\tG\,(\vec{X}^{m+1} - \vec{X}^m) 
\rangle_{\gamma,m}
\nonumber \\ & \hspace{1cm}
+\tau_m\,\frac{\lambda\,\rho}{a} \Big |[\beta(\vec\nu^m)]^{-\frac12}\,
\frac{\vec{X}^{m+1}-\vec{X}^m}{\tau_m}  \cdot \vec\omega^m \Big |_{m,h}^2
\nonumber \\ & \hspace{2cm}
= -\lambda\,u_D\,\langle \vec{X}^{m+1}-\vec{X}^m, \vec\omega^m\rangle_m^h
+ \tau_m\,(f^{m+1}, U^{m+1} - u_D)^h_{m,+},
\end{align*}
and hence (\ref{eq:stab}) follows immediately, where
we have used the result that
\begin{equation*}
\langle \nabs^\tG\,\vec{X}^{m+1}, \nabs^\tG\,(\vec{X}^{m+1} - \vec{X}^m) 
\rangle_{\gamma,m}
\geq |\Gamma^{m+1}|_\gamma - |\Gamma^{m}|_\gamma 
;  
\end{equation*}
see e.g.\ \cite{triplejANI}  
and \cite{ani3d}  
for the proofs for $d=2$ and $d=3$, respectively.
\end{proof}

The above theorem allows us to prove unconditional stability for our scheme
under certain conditions.

\begin{thm} \label{thm:stabstab}
Let the assumptions of Theorem~$\ref{thm:stab}$ hold with $u_D=0$.
In addition, assume that either $\vartheta=0$ or that
$U^m \in \Smhom$ 
and $\Omega^{m,h}_+ \subset \Omega^{m-1,h}_+$
for $m=1\to M-1$. Then it holds that
\begin{align}
&\mathcal{E}^m(U^{m+1}, \vec{X}^{m+1}) 
\nonumber \\ & \hspace{5mm}
+ \sum_{k=0}^m \tau_k\,\conduct \Big [ (\nabla\,U^{k+1}, \nabla\,
U^{k+1})_{k,+}
+ \frac{\lambda\,\rho}{a} \Big |[\beta(\vec\nu^k)]^{-\frac12}\,
\frac{\vec{X}^{k+1}-\vec{X}^k}{\tau_k}  \cdot \vec\omega^k \Big |_{k,h}^2 \Big ]
\nonumber \\ & \hspace{4cm}
\leq \mathcal{E}^0(U^0, \vec{X}^0)
+ \sum_{k=0}^m \tau_k\,(f^{k+1}, U^{k+1})^h_{k,+}
\label{eq:stabstab}
\end{align}
for $m=0\to M-1$.
\end{thm}

\noindent
{\bf Proof.}
The result immediately follows from (\ref{eq:stab}) on noting that, if
$\vartheta>0$, it follows from our assumptions that
$\mathcal{E}^m(U^{m}, \vec{X}^{m}) \leq \mathcal{E}^{m-1}(U^{m}, \vec{X}^{m})$
for $m=1\to M-1$, since then
\begin{align*} 
\int_{\Omega^{m,h}_+} I^m [ (U^m)^2 ] \dL    
&
\leq
\int_{\Omega^{m-1,h}_+} I^m [ (U^m)^2 ] \dL   
\\
&
=
\int_{\Omega^{m-1,h}_+} I^{m-1} [ (U^m)^2 ] \dL  .
\tag*{\qed}
\end{align*} 

\begin{rem}\label{rem:stab}
\rm
{\rm Theorem~\ref{thm:stabstab}} establishes the unconditional stability of our 
scheme {\rm (\ref{eq:uHGa}--c)} 
under certain conditions. Of course, if $u_D \not= 0$, 
analogous weaker stability results based on {\rm (\ref{eq:stab})} can 
be derived.
We note that the condition $U^m\in \SmD$ is trivially satisfied if 
$S^{m-1}_{D} \subset \SmD$, e.g.,  when mesh refinement routines without
coarsening are employed. 
The condition $\Omega^{m,h}_+ \subset \Omega^{m-1,h}_+$, on the other hand,
is ensured whenever the discrete solid region is
not shrinking. 
This is in line with the corresponding continuous 
energy law {\rm (\ref{eq:testD})}.
Note also that the condition $U^m\in \SmDl$ would be too strong, as in
physically meaningful computations the solid region grows, and so the
condition would enforce that $U^m = 0$ at vertices which are now in the
solid region, but were degrees of freedom in $S^{m-1}_{D,+}$.
In the simpler case that $\vartheta=0$, the stability
bound {\rm (\ref{eq:stab})} 
is independent of $U^m$, and so
here the stability bound {\rm (\ref{eq:stabstab})} holds 
for arbitrary choices of bulk meshes $\mathcal{T}^m$.
\end{rem}

\begin{rem} \label{rem:twosided}
\rm
With the techniques introduced in this paper, it is a simple matter to extend
the finite-element approximation introduced in \cite{dendritic} for the
two-sided Stefan problem with constant heat conductivity 
$\conduct = \conduct_s = \conduct_l$ to the case 
$\conduct_s - \conduct_l \not= 0$, where we have adopted the notation 
from \cite[(2.1a--e)]{dendritic}. Here the subscripts $s$ and $l$ refer to the
solid and liquid phase, respectively.

Our finite-element approximation for this problem is then given as follows.
Let $\Gamma^0$ be given.
For $m=0\to M-1$, find $U^{m+1} \in \SmD$, $\vec{X}^{m+1}\in\Vh$, and 
$\kappa^{m+1}_\gamma \in \Wh$ such that
for all $\varphi \in \Smhom$, $\chi \in \Wh$, and $\vec\eta \in \Vh$,
\begin{subequations}
\begin{align}
& \vartheta \Big (\frac{U^{m+1}-U^m}{\tau_m}, \varphi  \Big )^h_{m} +
\sum_{i\in\{l,s\}}  \Big [
 \conduct_i\,(\nabla\,U^{m+1}, \nabla\,\varphi)_{m,i} 
- (f^{m+1}_i, \varphi )^h_{m,i}
 \Big ]
\nonumber \\ & \hspace{4cm}
- \lambda\,
 \Big \langle \pi^m \Big [
\frac{\vec{X}^{m+1}-\vec{X}^m}{\tau_m}  \cdot \vec\omega^m  \Big ], \varphi
 \Big \rangle_m = 0  
, \label{eq:2HGa}\\
& \rho \Big \langle 
[\beta(\vec\nu^m)]^{-1}\,\frac{\vec{X}^{m+1}-\vec{X}^m}{\tau_m},
\chi\,\vec\omega^m  \Big \rangle_m^h - \alpha\, \langle \kappa^{m+1}_\gamma, 
\chi\rangle_m^h
+ a \,\langle U^{m+1}, \chi \rangle_m =0  
, \label{eq:2HGb} \\
& \langle \kappa^{m+1}_\gamma\,\vec\omega^m, \vec\eta \rangle_m^h + 
\langle \nabs^\tG\,\vec{X}^{m+1}, \nabs^\tG\,\vec\eta \rangle_{\gamma,m}
= 0  
, \label{eq:2HGc} 
\end{align}
\end{subequations}
and set $\Gamma^{m+1} = \vec{X}^{m+1}(\Gamma^m)$.
Here $(\nabla\,\chi , \nabla\,\varphi)_{m,i}$ and
$(\chi, \varphi)_{m,i}^h$, for $i \in \{s,l\}$ and for $\chi, \varphi \in \Sm$,
are defined analogously to (\ref{eq:intml},b),
where $\Omega^{m,h}_l := \Omega^{m,h}_+$ and 
$\Omega^{m,h}_s := \Omega^{m,h}_-$ 
represent approximations to the ``liquid'' and ``solid'' phases in this
two-sided Stefan
problem.
\end{rem}

\setcounter{equation}{0}
\section{Solution of the discrete system}  \label{sec:4}
Introducing the obvious abuse of notation, the linear system (\ref{eq:uHGa}--c) 
can be formulated as follows: Find $(U^{m+1},\kappa^{m+1}_\gamma,\delta\vec{X}^{m+1})$ 
such that 
\begin{equation}
\begin{pmatrix}
\frac1{\tau_m}\,M_\Omega + A_\Omega & 0 &
-\frac\lambda{\tau_m}\,\Nbulk \\
-a\,\Mbulk & \alpha\,M_\Gamma &
- \frac\rho{\tau_m}\,[\vec{N}_\Gamma^{(\beta)}]^T \\
0 & \vec{N}_\Gamma & \vec{A}_\Gamma 
\end{pmatrix}  
\begin{pmatrix} U^{m+1} \\  \kappa^{m+1}_\gamma \\ 
\delta\vec{X}^{m+1} \end{pmatrix}
=
\begin{pmatrix} \frac1{\tau_m}\,M_\Omega\,U^m + g^m \\  
0 \\ -\vec{A}_\Gamma\,\vec{X}^{m} \end{pmatrix} ,
\label{eq:lin}
\end{equation}
where $(U^{m+1},\kappa^{m+1}_\gamma,\delta\vec{X}^{m+1}) \in
\R^{K^m_\Omega}\times \R^{K^m_\Gamma}\,\times (\R^d)^{K^m_\Gamma}$
here denote the coefficients of these finite-element functions with respect to
the standard bases of $\Sm$, $\Wh$, and $\Vh$, respectively.
The definitions of the matrices in (\ref{eq:lin}) directly follow from
(\ref{eq:uHGa}--c), but we state them here for completeness.
Let $i,j = 1 \to K_\Omega^m$ and $k,l = 1 \to K^m_\Gamma$.
Then, on recalling (\ref{eq:Sh0}), we have that
\begin{align}
& [M_\Omega]_{ij} := \vartheta\,(\phi_j^m, \phi_i^m)^h_{m,+}, \nonumber \\
& [\widetilde A_\Omega]_{ij} := 
\begin{cases}
\conduct\,(\nabla\,\phi_j^m, \nabla\,\phi_i^m)_{m,+} & 
1 \leq i \leq K^m_{\Omega,D} \\
\delta_{i,j} & K^m_{\Omega,D} < i \leq K^m_{\Omega} 
\end{cases}
, \nonumber \\ &
[\Mbulk]_{li} := 
\langle \phi_i^m, \chi^m_l \rangle_m
, \nonumber \\ & [\vec{N}_{\Gamma,\Omega}]_{li} := \left( 
\langle \phi_i^m, \pi^m\left[(\chi^m_l\,\vec{e}_j) \cdot \vec\omega^m\right] 
\rangle_m \right)_{j=1}^d = 
\langle \phi_i^m, \chi^m_l \rangle_m\,\vec\omega^m_l, 
\nonumber \\
& [M_\Gamma]_{kl} := \langle \chi^m_l, \chi^m_k \rangle_m^h, \quad 
[\vec{A}_\Gamma]_{kl} := \left(
\langle \nabs^\tG\,(\chi^m_l\,\vec{e}_i), \nabs^\tG\,(\chi^m_k\,\vec{e}_j)
\rangle_{\gamma,m} \right)_{i,j=1}^d , \nonumber \\
& [\vec{N}_\Gamma]_{kl} := \langle 
\chi^m_l, \chi^m_k\,\vec\omega^m \rangle_m^h , \nonumber \\ & 
[\vec{N}^{(\beta)}_\Gamma]_{kl} := \langle [\beta(\vec\nu^m)]^{-1}
\chi^m_l, \chi^m_k\,\vec\omega^m \rangle_m^h
=\langle [\beta(\vec\nu^m)]^{-1}
\chi^m_l, \chi^m_k\rangle_m^h\,\vec\omega^m_l \label{eq:MAM}
, 
\end{align}
where $\{\vec{e}_i\}_{i=1}^d$ denotes the standard basis in $\R^d$
and where we have used the convention that the subscripts in the matrix
notation refer to the test and trial domains, respectively. A single subscript
is used where the two domains are the same.
We note that the special definition of $\widetilde A_\Omega$, together
with $g^m$ in (\ref{eq:lin}), accounts for the Dirichlet boundary conditions
of $U^{m+1} \in \SmD$. Here $g^m$ is defined by
\begin{equation} \label{eq:gm}
g^m_i = \begin{cases}
(f^{m+1} , \phi^m_i)^h_{m,+} & 1 \leq i \leq K^m_{\Omega,D} ,\\
u_D & K^m_{\Omega,D} < i \leq K^m_{\Omega} .
\end{cases}
\end{equation}
Clearly, the matrix $\widetilde A_\Omega$ will in
general be singular. In particular, it will have zero diagonal entries for
every vertex $\vec p^m_j \in \overline\Omega^{m,h}_-$. Hence we enforce
$U^{m+1} \in \SmDl$ by setting
\begin{equation} \label{eq:A}
[A_\Omega]_{ij} = \begin{cases}
[\widetilde A_\Omega]_{ij} & [\widetilde A_\Omega]_{ij} \not= 0  ,\\
\delta_{i,j} & [\widetilde A_\Omega]_{ij} = 0 \,;
\end{cases}
\end{equation}
i.e.,  we replace zero diagonal entries by $1$.

The assembly of the matrices in (\ref{eq:MAM}), apart from 
$\widetilde A_\Omega$, is described in \cite[Section~4]{dendritic}. The assembly of
$\widetilde A_\Omega$, and in particular the possible definitions of the
region $\Omega^{m,h}_+$, will be discussed in Section~\ref{sec:41} below.
The linear system (\ref{eq:lin}) can be efficiently solved with iterative
solvers applied to a Schur complement formulation; see \cite{dendritic} for
details. For completeness we state that for the application of preconditioners
and for the solution of subproblems we make use of the packages
LDL and AMD; see \cite{Davis05,AmestoyDD04}.

\subsection{Definition of the discrete liquid/gas region} \label{sec:41}
We now discuss possible choices of $\Omega^{m,h}_+$ in 
(\ref{eq:Sml}) and (\ref{eq:intml},b). 
To this end, we partition the elements of the  bulk mesh 
$\mathcal{T}^m$ into liquid/gas, solid, and interfacial elements as follows.
Let
\begin{align}
\mathcal{T}^m_+ & := \{ o^m \in \mathcal{T}^m : \overline o^m \subset
\Omega^m_+ \} , \  \quad \
\mathcal{T}^m_-   := \{ o^m \in \mathcal{T}^m : \overline o^m \subset
\Omega^m_- \} , \nonumber \\
\mathcal{T}^m_{\Gamma^m} & := \{ o^m \in \mathcal{T}^m : \overline o^m \cap
\Gamma^m \not = \emptyset \} . \label{eq:partT}
\end{align}
Then $\mathcal{T}^m = \mathcal{T}^m_+ \cup \mathcal{T}^m_- \cup
\mathcal{T}^m_{\Gamma^m}$ is a disjoint partition. 

Clearly, using $\Omega^{m,h}_+ = \Omega^{m}_+$ is not very practical, since the
intersection of $\Omega^m_+$ with elements of the bulk mesh $\mathcal{T}^m$
can be complicated. Moreover, computing the domain $\Omega^{m}_+$ is unlikely
to be rewarded with lower overall approximation errors, since the trial and test
functions in (\ref{eq:intml},b) are only piecewise linear.
Instead, we consider the following approach, which defines
$\Omega^{m,h}_+$ with the help of a piecewise-linear approximation to
$\charfcn{\overline\Omega^{m}_+}$ as
\begin{equation} \label{eq:unfitBS}
\Omega^{m,h}_+ := \{ \vec p \in \Omega : 
(I^m\,\charfcn{\overline\Omega^m_+})(\vec p) > 0 \}
.
\end{equation}

Next we discuss an algorithm that computes $\Omega^{m,h}_+$ for the 
strategy (\ref{eq:unfitBS}). 
Here each element of
$\mathcal{T}^m$ is assigned to one of the three sets 
$\mathcal{T}^m_+$, $\mathcal{T}^m_-$, or $\mathcal{T}^m_{\Gamma^m}$ as 
described in Algorithm~\ref{algo:markels}.
\begin{myalgorithm}[t]
\caption{Mark all bulk mesh elements as liquid/gas, solid, or cut.} 
\label{algo:markels}
\raggedright
1. Traversing over $\Gamma^m$, find all elements of 
$\mathcal{T}^m_{\Gamma^m}$. \\
2. Set $\mathcal{T} := \mathcal{T}^m \setminus \mathcal{T}^m_{\Gamma^m}$
and $\mathcal{T}^m_+ := \emptyset$. \\
3. Move all elements that touch $\partial\Omega$
from $\mathcal{T}$ to $\mathcal{T}^m_+$.\\
4. For as long as this is possible, move neighbours of
elements in $\mathcal{T}^m_+$ from \\ \quad 
$\mathcal{T}$ to $\mathcal{T}^m_+$.\\
5. Set $\mathcal{T}^m_- := \mathcal{T}$.
\end{myalgorithm}
In addition, for later use, we need to decide for each
bulk mesh vertex $\vec p_j^m$, $j=1\to K^m_\Omega$, whether it belongs to
$\overline\Omega^m_+$ or to $\overline\Omega^m_-$. This can be done as 
described in Algorithm~\ref{algo:marknodes}. 

\begin{myalgorithm}[t]
\caption{Assign all bulk mesh vertices to $\overline\Omega^{m}_-$
or $\overline\Omega^{m}_+$.} \label{algo:marknodes}
\raggedright
1. All vertices of elements in $\mathcal{T}^m_-$ belong to
$\overline\Omega^m_-$. \\
2. All vertices of elements in $\mathcal{T}^m_+$ belong to
$\overline\Omega^m_+$. \\
3. For any remaining vertices $\{\vec p_j^m\}$, choose a $\vec p^m_j$
with 
a neighbouring \\ \quad vertex
$\vec q$ 
that is known to belong to $\overline\Omega^m_-$ or 
$\overline\Omega^m_+$. If $\Gamma^m$ cuts the \\ \quad segment 
$[\vec p_j^m, \vec q] \subset \R^d$ an even number of times, assign 
$\vec p_j^m$ to the same \\ \quad region as $\vec q$, otherwise to the
opposite 
region.
Repeat this, until all \\ \quad vertices have been assigned.
\end{myalgorithm}

\begin{rem} \label{rem:sweeping}
\rm
The global {\rm Algorithm~\ref{algo:markels}} 
is only needed at the very first time
step. For subsequent time steps, the existing marking of bulk mesh elements can
be updated depending on the movement of $\Gamma^m$. In particular, only
elements in 
$\mathcal{T}^{m-1}_{\Gamma^{m-1}} \setminus \mathcal{T}^{m}_{\Gamma^{m}}$
need to be considered. On assuming that $\Gamma^m$ has not travelled over a
whole bulk mesh element, these elements can be marked with the help of
neighbour information. This is far more efficient than employing the global
{\rm Algorithm~\ref{algo:markels}} at every time step.

In addition, in practice for a refined bulk mesh in the neighbourhood of $\Gamma^m$,
all the remaining vertices in {\rm Step 3 of Algorithm~\ref{algo:marknodes}} 
have immediately a 
neighbouring vertex
$\vec q$  
that is known to belong to $\overline\Omega^m_-$ or 
$\overline\Omega^m_+$.  
\end{rem}

An alternative approach to {\rm (\ref{eq:unfitBS})} 
would not define $\Omega^{m,h}_+$ explicitly, but rather the effect of 
$\Omega^{m,h}_+$ on the inner products defined in {\rm (\ref{eq:intml},b)}.
Here it is natural to define $\Omega^{m,h}_+$ in such a way, that
$\bigcup_{o^m \in \mathcal{T}^m_-} \overline o^m \subset
\overline\Omega^{m,h}_-$.
Then the integral in {\rm (\ref{eq:intml})} can be rewritten as
\begin{equation}
(\nabla \chi , \nabla\,\varphi)_{m,+}   
= 
\!\!\!\!\!
\sum_{o^m \in \mathcal{T}^m_+}
\int_{o^m} 
\!\!\!
\nabla\,\chi  .  \nabla\,\varphi \dL  
+
\!\!\!\!\!
 \sum_{o^m \in \mathcal{T}^m_{\Gamma^m}}
v(o^m) \int_{o^m} 
\!\!\!
\nabla \chi . \nabla \varphi \dL ,
\label{eq:unfitted2}
\end{equation}
where $v(o^m) \in [0,1]$ denotes the fraction of the element
$o^m$ that is considered to belong to the liquid/gas region
$\Omega^{m,h}_+$, and similarly for the inner product defined in
{\rm (\ref{eq:inthml})}. 
Note that {\rm (\ref{eq:unfitted2})} only implicitly defines (candidates of) 
the region $\Omega^{m,h}_+$.

In practice, several choices of $v(o^m) \in [0,1]$ can be considered.
The approach {\rm (\ref{eq:unfitBS})} corresponds to 
\begin{subequations}
\begin{equation} \label{eq:unfit0}
v(o^m) = 1 \qquad \forall\ o^m \in \mathcal{T}^m_{\Gamma^m},
\end{equation}
while the choice 
\begin{equation} \label{eq:unfit1}
v(o^m) = 0 \qquad \forall\ o^m \in \mathcal{T}^m_{\Gamma^m} 
\end{equation}
was used in \cite{BanschS00} for a two-sided Stefan problem with nonvanishing
heat conductivity coefficients. We note that for the one-sided situation
considered in this paper, the strategy {\rm (\ref{eq:unfit1})} does not make
sense, as it dramatically affects the accuracy of the approximation $U^{m+1}$
on $\Gamma^m$.
An alternative approach is the choice
\begin{equation} \label{eq:unfit3}
v(o^m) = \tfrac{k}{d+1} 
= \frac{1}{|o^m|} \int_{o^m} I^m \,\charfcn{\overline\Omega^m_+} \dL
\qquad \forall\ o^m \in \mathcal{T}^m_{\Gamma^m}, 
\end{equation}
where $k$ denotes the number of vertices of
$\overline o^m$ that lie within $\overline\Omega^m_+$. 
A simpler approach is to set
\begin{equation} \label{eq:unfit2}
v(o^m) = \tfrac12 \qquad \forall\ o^m \in \mathcal{T}^m_{\Gamma^m}.
\end{equation}
\end{subequations}
We note that for the practical implementation, the strategies
(\ref{eq:unfit0},b,d) only need the marking from Algorithm~\ref{algo:markels}.
The additional Algorithm~\ref{algo:marknodes} is only required for
the strategy
(\ref{eq:unfit3}). In practice, the three strategies (\ref{eq:unfit0},c,d) all show very
similar numerical results. Hence, in general we will employ
the simplest strategy
(\ref{eq:unfit0}).

\begin{rem} \label{rem:Oimplicit}
\rm
Of course, setting
\begin{equation} \label{eq:unfitexact}
v(o^m) = |o^m \cap \Omega^m_+| 
\qquad \forall\ o^m \in \mathcal{T}^m_{\Gamma^m}
\end{equation}
corresponds to $\Omega^{m,h}_+ = \Omega^m_+$. This will in general be too
costly to do in practice. However, we mention one possible strategy here. For
an arbitrary open bounded set $V \subset \R^d$ with Lipschitz boundary
it holds that
\begin{equation} \label{eq:volV}
\vol(V) = \int_V1 \dL = 
\tfrac1d\,\int_{\partial V} (\vec\id - \vec z_0) \cdot \vec\nu_V \dH ,
\end{equation}
where $\vec\id$ is the identity function on $\R^d$, 
$\vec z_0 \in \R^d$ is an arbitrarily fixed point, and where
$\vec\nu_V$ denotes the outer normal to $V$. Applying 
{\rm (\ref{eq:volV})} for $V = o^m \cap \Omega^m_+$, 
on noting that
$\vec\nu_V = - \vec\nu^m$ on $o^m \cap \Gamma^m$ and
$\vec\nu_V = \vec\nu_{o^m}$, the outer normal of $o^m$, on 
$\partial o^m \cap \Omega^m_+$,
yields a way of using
{\rm (\ref{eq:unfitexact})} in practice. Of course, in this case $V$ is a
polytope, with $\partial V$ being a union of flat facets. Thus the integral in
{\rm (\ref{eq:volV})} simplifies on noting that $\vec\id \cdot \vec\nu_V$ is now
constant on each facet, and vanishes on each facet that contains $\vec z_0$.
Moreover, $o^m \cap \Gamma^m$ can be computed
as in \cite[Section~4.5]{dendritic}. It remains to calculate
$\partial o^m \cap \Omega^m_+$, where for our purposes it is enough
to compute $|F_\mu \cap \Omega^m_+|$ for $\mu=1\to d+1$, where $F_\mu$ are the 
edges/faces of $o^m$; 
i.e.,  $\partial o^m = \cup_{\mu=1}^{d+1} \overline F_\mu$.
For $d=2$ this reduces to finding the lengths of $F_\mu \cap \Omega^m_+$, which
is straightforward.
For $d=3$ the set $F_\mu \cap \Omega^m_+$ in general can be the disjoint union
of possibly non-convex polygons. The oriented boundary of these polygons can be
found by suitably arranging the line segments making up 
$\partial F_\mu \cap \Gamma^m$, as well as the line segments making up
$\partial F_\mu \cap \Omega^m_+$. Then the area $|F_\mu \cap \Omega^m_+|$ can
be easily computed with Gauss' area formula.
\end{rem}

\setcounter{equation}{0}
\section{Numerical results}  \label{sec:6}
We implemented our finite-element approximation (\ref{eq:uHGa}--c) within 
the framework of the finite-element toolbox ALBERTA; see \cite{Alberta}. 
We use the bulk mesh and parametric mesh refinement strategies introduced in 
\cite[Section~5]{dendritic}. Here the bulk mesh adaptation algorithm, which
was inspired by a similar strategy proposed in \cite{voids} and \cite{voids3d}
for $d=2$ and $d=3$, respectively, results in a fine mesh of uniform mesh size 
$h_f$ around $\Gamma^m$ and a coarse mesh of uniform mesh size $h_c$ 
further away from it. Here
$h_{f} = \frac{2\,H}{N_{f}}$ and $h_{c} =  \frac{2\,H}{N_{c}}$
are given by two integer numbers $N_f >  N_c$, where we assume from now on that
$\Omega = (-H,H)^d$. For the one-sided problems considered in this paper, we
slightly amend the strategy from \cite[Section~5]{dendritic}, in that we allow an
even coarser grid inside $\Omega^{m,h}_-$. Of course, the definitions
(\ref{eq:Sml}) mean that this has no effect on the numerical results.
Moreover, the parametric mesh refinement uses bisections in order
to avoid elements getting too large over time. We stress that apart from 
this simple mesh refinement, no other changes were performed on the 
parametric mesh in any of our simulations. 
In particular, no mesh smoothing (redistribution) was required.

Throughout this section we use (almost) uniform time steps, in that 
$\tau_m=\tau$, $m=0\to M-2$, and $\tau_{M-1} = T-t_{m-1} \leq \tau$.
Unless otherwise stated we set $\Omega=(-H,H)^d$ with $H=4$.
Similarly, unless otherwise stated, we always employ the strategy
(\ref{eq:unfit0}) for the computation of $\Omega^{m,h}_+$.
The initial interface $\Gamma(0)$ is always a circle/sphere of radius 
$R_0 \in (0,H)$ around the origin.
For the Stefan problem, i.e.,  if $\vartheta > 0$, we set
\begin{equation}
u_0(\vec{z}) = \begin{cases} 0 & |\vec{z}| \leq R_0, \\
\dfrac{1 - e^{R_0 - |\vec{z}|}}{1 - e^{R_0 - H}}\,u_D 
& R_0 < |\vec{z}| < H, \\
u_D & |\vec{z}| \geq H ,
\end{cases}
\label{eq:u0expcircle}
\end{equation}
unless a true solution $u$ is given.

For later purposes, we define
$$
 \vec{X}(t):=\tfrac{t-t_{m-1}}{\tau_m}\,\vec{X}^m+ \tfrac{t_m-t}{\tau_{m-1}} 
  \vec{X}^{m-1} , \ \  t\in [t_{m-1},t_m], \ \  m\geq 1 ,
$$ 
and similarly for $U$. 

\subsection{Non-dimensionalization of a model for snow crystal growth}
\label{sec:nondim}

An aim of this paper is to be able to perform computations for the growth of
snow crystals with realistic parameters
and on physically relevant length and time  scales. 
Upon non-dimensionalizing the 
continuum model for snow crystal growth from
\cite{Libbrecht05}, it turns out that (\ref{eq:1a}--c) with
\begin{equation} \label{eq:HGparams}
\vartheta = 0, \ \ \mathcal{K} =1, \ \ \lambda = 1, \ \ \rho =
1.42\times 10^{-3}, \ \ \alpha = 10^{-5}, \ \ a = 1, \ \ f = 0
\end{equation}
is a physically realistic model. Here the typical length scale is 100\,$\mu$m, 
typical time scales vary from 100\,s to 1300\,s, $-u$ denotes a scaled
concentration of water vapour in the gas phase,
and $-u_D$ is a scaled supersaturation. We refer to \cite{jcg} for more
details on the physical interpretation of these parameters.

\subsection{Convergence experiments}
We begin with a comparison of the approximation error 
$\vol(\Omega_+(0)) - \vol(\Omega^{0,h}_+)$ for the four different strategies
(\ref{eq:unfit0}--d).
Here we set $\Omega_+(0) = \overline\Omega \setminus B_1(0)$ and, for the case
$d=2$, use the spatial discretization
parameters $N_f = K^0_\Gamma = 2^{7+i}$ and $N_c = 4^{i}$. 
An example of how the discrete interface $\Gamma^0$ cuts the bulk mesh
$\mathcal{T}^0$ is shown in Figure~\ref{fig:cutmesh}.
 
The numerical results are shown in Table~\ref{tab:volO}, where we observe that
the strategies (\ref{eq:unfit3},d) produce far smaller errors than 
(\ref{eq:unfit0},b). However, as we will see in the subsequent convergence
experiments, this does not seem to have an influence on the overall
approximation error for the underlying solutions $u$ and $\Gamma$.

For completeness, we repeat the same experiments for $d=3$, where now
$N_f = 2^{6+i}$, $N_c = 4^{i}$, and 
$K^0_\Gamma = K(i)$, with $(K(0), K(1), K(2), K(3)) = (770, 3074, 12290, 
49154)$, for $i=0\to 3$. The results are shown in Table~\ref{tab:volO3d}.

\begin{figure}[t]
\center
   \includegraphics[angle=-90,width=0.3\textwidth]{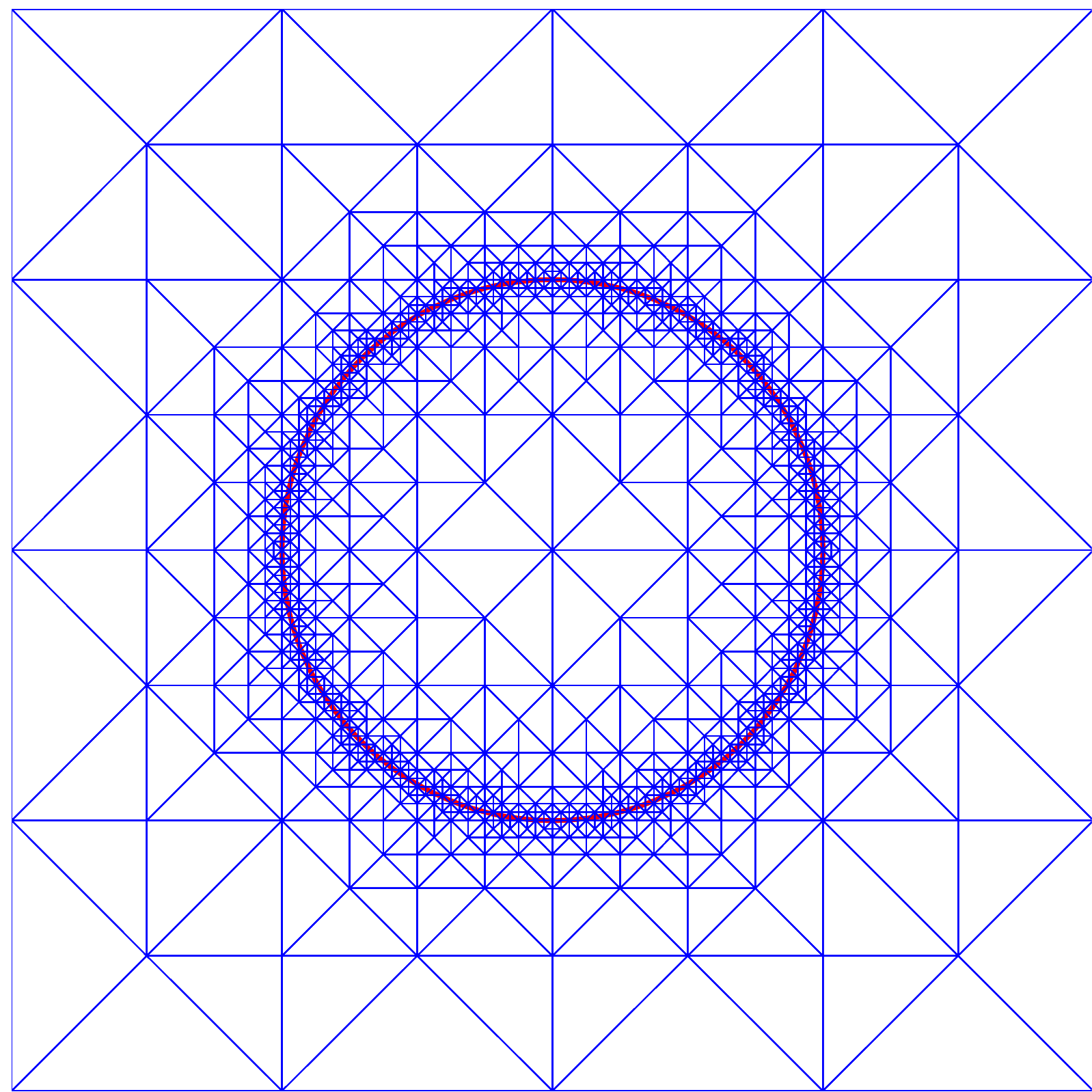}
   \includegraphics[angle=-90,width=0.3\textwidth]{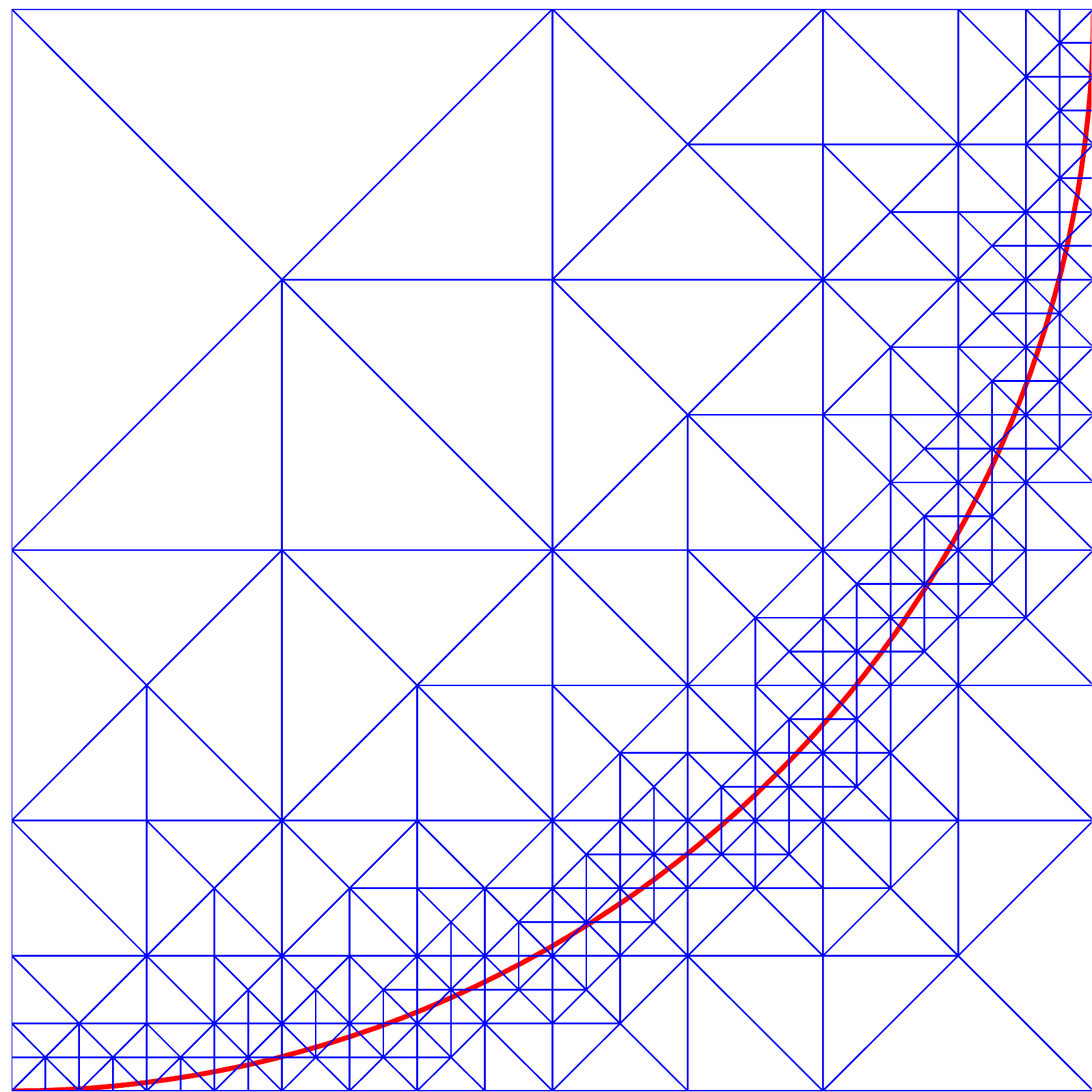}
   \includegraphics[angle=-90,width=0.3\textwidth]{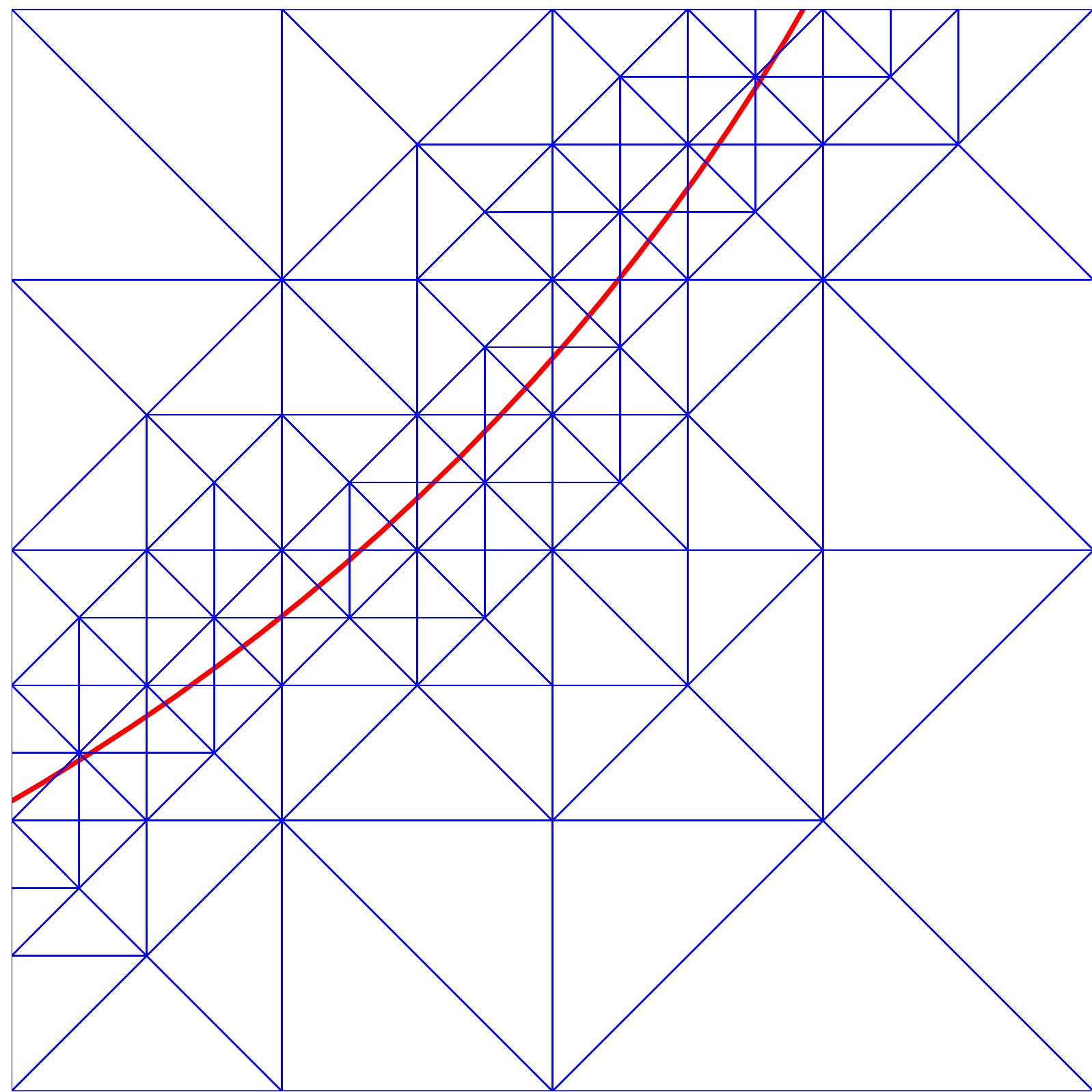}
\caption{Parts of the triangulation $\mathcal{T}^0$ and the interface 
$\Gamma^0$ when $N_f = K^0_\Gamma = 2^8$ and $N_c = 4$. From left to right
$[-2,2]^2$, $[-1,0]^2$, and $[-1,-\frac12]^2$.}
\label{fig:cutmesh}
\end{figure} 
\small
\begin{table}[t]
\center
\begin{tabular}{c|c|c|c|c|c}
 $10^3\,h_{f}$ & $h^M_\Gamma$ & (\ref{eq:unfit0}) & (\ref{eq:unfit3}) & 
(\ref{eq:unfit2}) & (\ref{eq:unfit1}) \\
\hline
62.500 & 5.4874e-02 & -2.3534e-01 & -1.1384e-02 & -8.7802e-03 &
2.1778e-01 \\
31.250 & 2.7439e-02 & -1.1425e-01 & -4.8739e-03 & -4.8739e-03 &
1.0450e-01 \\
15.625 & 1.3720e-02 & -5.5655e-02 & -1.9442e-03 & -1.4559e-03 &
5.2743e-02 \\
7.8125 & 6.8600e-03 & -2.7579e-02 & -1.2118e-03 & -7.2351e-04 &
2.6132e-02 \\
3.9062 & 3.4300e-03 & -1.4273e-02 & -7.0317e-04 & -8.7610e-04 &
1.2521e-02   \\
\end{tabular}
\caption{\small $\Omega=(-4,4)^2$. Approximation error
$\vol(\Omega_+(0))-\vol(\Omega^{0,h}_+)$ for (\ref{eq:unfit0}--d).}
\label{tab:volO}
\end{table} 
\begin{table}[t]
\center
\begin{tabular}{c|c|c|c|c|c}
 $10^2\,h_{f}$ & $h^M_\Gamma$ & (\ref{eq:unfit0}) & (\ref{eq:unfit3}) & 
(\ref{eq:unfit2}) & (\ref{eq:unfit1}) \\
\hline
12.500 & 2.0854e-01 & -8.9192e-01 & -6.7696e-02 & -1.2902e-03 & 
8.8933e-01 \\
6.2500 & 1.0472e-01 & -4.5246e-01 & -1.7403e-02 & -3.2433e-03 &
4.4598e-01 \\
3.1250 & 5.2416e-02 & -2.2370e-01 & -3.5485e-03 & 4.1878e-04 &
2.2454e-01 \\
1.5625 & 2.6215e-02 & -1.1247e-01 & -8.0954e-04 & -2.8311e-04 &
1.1190e-01 \\
\end{tabular}
\caption{$\Omega=(-4,4)^3$. Approximation error
$\vol(\Omega_+(0)) - \vol(\Omega^{0,h}_+)$ for (\ref{eq:unfit0}--d).}
\label{tab:volO3d}
\end{table} 
\begin{table}[t]
\center
\begin{tabular}{c|c|c|c|c|c}
 $h_{f}$ & $h^M_\Gamma$ & $\errorUupl$ & $\errorXx$ & $K^M_\Omega$ & $K^M_\Gamma$ 
\\
\hline
 6.2500e-02 & 5.0640e-02 & 2.4595e-01 & 9.2545e-02 & 677 & 128 \\
 3.1250e-02 & 2.7093e-02 & 7.2888e-02 & 2.1049e-02 & 1329 & 256 \\
 1.5625e-02 & 1.3740e-02 & 2.0818e-02 & 3.5439e-03 & 2753 & 512 \\
 7.8125e-03 & 6.8637e-03 & 5.2596e-03 & 6.2892e-04 & 8853 & 1024 \\
 3.9062e-03 & 3.4307e-03 & 1.2318e-03 & 2.1081e-04 & 71305 & 2048 \\
\end{tabular}
\caption{$\Omega=(-4,4)^2$ and $\TTime=1$. Convergence test for 
(\ref{eq:Trueu}) with (\ref{eq:unfit0}).}
\label{tab:HGnew2d}
\end{table} 
\begin{table}[t]
\center
\begin{tabular}{c|c|c|c|c}
 & \multicolumn{2}{c|}{(\ref{eq:unfit3})} & 
 \multicolumn{2}{c}{(\ref{eq:unfit2})} \\ \hline
 $h_{f}$ & $\errorUupl$ & $\errorXx$ & $\errorUupl$ & $\errorXx$ \\ \hline
 6.2500e-02 & 2.5105e-01 & 9.9965e-02 & 2.5204e-01 & 1.0185e-01 \\
 3.1250e-02 & 7.7931e-02 & 2.4780e-02 & 7.8798e-02 & 2.5502e-02 \\
 1.5625e-02 & 2.3909e-02 & 4.8305e-03 & 2.4363e-02 & 5.1527e-03 \\
 7.8125e-03 & 6.1309e-03 & 1.5131e-03 & 6.2823e-03 & 1.7784e-03 \\
 3.9062e-03 & 1.7662e-03 & 7.6027e-04 & 1.8835e-03 & 9.4426e-04 \\
\end{tabular}
\caption{$\Omega=(-4,4)^2$ and $\TTime=1$. Convergence test for 
(\ref{eq:Trueu}) with (\ref{eq:unfit3}) and (\ref{eq:unfit2}).}
\label{tab:HGnew2di}
\end{table} 
\normalsize

\subsubsection{One-sided Stefan problem}
Next we investigate the approximative \linebreak[4] properties of our algorithm 
(\ref{eq:uHGa}--c) for the
following exact solution to the one-sided Stefan problem (\ref{eq:1a}--e), 
in the case of 
the isotropic surface energy (\ref{eq:iso}). 
Here we adapt the following expanding circle/sphere solution for the two-phase
Stefan problem in \cite[(6.5)]{dendritic},
where the radius of the circle/sphere is given by $r(t)$,
and so $\Omega_+ (t)= \Domain \setminus \overline B_{r(t)}(0)$. Assume that
$ \vartheta = \conduct = \lambda = \rho = \alpha = a = 1 $ 
and let
$$
r(t) = (r^2(0) + t)^\frac12 , 
\ \  
w(t) = -\frac{d-\frac12}{r(t)}, 
\ \
v(s) = - \frac{e^{\frac14}}2\,\int_1^s \frac{e^{-\frac14\,z^2}}{z^{d-1}}\;{\rm
d}z.  
$$
Then it is easy to see that on letting

$$
f(\vec{z},t) = \ddt w(t) = \frac{d - \frac12}{2\,r^3(t)},
$$
the solution $u$ to (\ref{eq:1a}--e), with $u_D$ in (\ref{eq:1d}) replaced by
$u\!\mid_{\partial_D\Domain}$, is given by the restriction to $\Omega_+(t)$ of
\begin{equation}
u(\vec{z},t) = 
\begin{cases}
w(t) & \vec{z} \in \overline\Omega_-(t), \\
w(t) + v\left(\frac{|\vec{z}|}{r(t)}\right) 
& \vec{z} \in \Omega_+(t).
\end{cases}
\label{eq:Trueu}
\end{equation}
For $d=2$, we perform the following convergence experiment for the solution
(\ref{eq:Trueu}), where we use $r(0) = R_0 = 0.5$.   
For $i=0\to 4$, we set $N_f = 2 K^0_\Gamma = 2^{7+i}$, 
$N_c = 4^{i}$,
and $\tau= 4^{3-i}\times10^{-3}$. 
The errors $\errorUupl$ and $\errorXx$ on the
interval $[0,\TTime]$ with $\TTime=1$, so that $r(\TTime) \approx 1.12$, 
are displayed in Table~\ref{tab:HGnew2d}.
Here 
$$
\errorUupl := \max_{m=1 \to M} \|U^m - I^{m-1}\,u(t_m)\|_{\infty,m-1,+} ,
$$
where 
$$
\|U^m - I^{m-1}\,u(t_m)\|_{\infty,m-1,+} :=
\max_{\vec p \in \mathcal{N}^{m-1}_+} |U^m(\vec p) - u(t_m, \vec p)|
$$
and 
$$
\mathcal{N}^{m-1}_+ := \{ \vec p^{m-1}_j : j = 1 \to K^{m-1}_\Omega \}
\cap \overline\Omega^{m-1}_+ \cap \overline \Omega_+(t_m) .
$$
Moreover, 
$$
\errorXx :=
\max_{m=1\to M} \|\vec{X}^m - \vec{x}(\cdot,t_m)\|_{L^\infty} , 
$$
where
$\|\vec{X}(t_m) - \vec{x}(\cdot,t_m)\|_{L^\infty}
:= \max_{k=1\to K^m_\Gamma} 
\left\{\min_{\vec{y}\in \Upsilon}|\vec{X}^m(\vec{q}^m_k) -
\vec{x}(\vec{y},t_m)|\right\}$, and
$h^M_\Gamma := \max_{j = 1 \to J^M_\Gamma} 
\diam(\sigma^M_j)$.
Note that $K^M_\Gamma = 2\,K^0_\Gamma$ due to the growth of the interface.

In addition, we use the convergence experiment in order to compare the
different strategies (\ref{eq:unfit3}) and (\ref{eq:unfit2}). See
Table~\ref{tab:HGnew2di}, where we present the same computations as in
Table~\ref{tab:HGnew2d}, but now for (\ref{eq:unfit3}) and
(\ref{eq:unfit2}). For the new results we omit the additional mesh statistics,
as they are very similar to the results for (\ref{eq:unfit0}) shown in
Table~\ref{tab:HGnew2d}.
 
We also compare the numbers in Tables~\ref{tab:HGnew2d} and \ref{tab:HGnew2di}
with 
the corresponding errors for the approximation from \cite{dendritic} for the
two-phase Stefan problem  (see (2.1a--e) in \cite{dendritic}), 
with the same choice of parameters. Note that $u(\cdot,t):\Omega \to \R$ as
defined in (\ref{eq:Trueu}) then is the desired true solution.
The corresponding errors, where $\errorUu :=
\max_{m=1\to M}\|U^m - I^{m-1}\,u(\cdot,t_m)\|_{L^\infty}$, 
can be seen in Table~\ref{tab:tpsp}.

\begin{table} 
\center
\begin{tabular}{c|c|c|c|c|c}
 $h_{f}$ & $h^M_\Gamma$ & $\errorUu$ & $\errorXx$ & $K^M_\Omega$ & $K^M_\Gamma$ 
\\
\hline
 6.2500e-02 & 5.0474e-02 & 2.4940e-01 & 9.7039e-02 & 645 & 128 \\
 3.1250e-02 & 2.7082e-02 & 7.3208e-02 & 2.2291e-02 & 1353 & 256 \\
 1.5625e-02 & 1.3739e-02 & 2.0678e-02 & 3.9277e-03 & 2753 & 512 \\
 7.8125e-03 & 6.8641e-03 & 4.9403e-03 & 7.2470e-04 & 9017 & 1024 \\
 3.9062e-03 & 3.4309e-03 & 1.2377e-03 & 2.8003e-04 & 74589 & 2048 \\
\end{tabular}
\caption{$\Omega=(-4,4)^2$ and $\TTime=1$. Convergence test for the
two-phase Stefan problem.}
\label{tab:tpsp}
\end{table} 

Similarly to Table~\ref{tab:HGnew2d}, we perform a convergence test for the 
solution (\ref{eq:Trueu}) to the one-sided Stefan problem, 
now for $d=3$, leaving all the remaining parameters fixed as before.
To this end, for $i=0\to 3$, we set $N_f = 2^{6+i}$, $N_c = 4^{i}$, and 
$K^0_\Gamma = K(i)$, where $(K(0), K(1), K(2), K(3)) = (770, 3074, 12290, 
49154)$,
and $\tau= 4^{3-i}\times10^{-3}$. The errors $\errorUupl$ and $\errorXx$ on the
interval $[0,\TTime]$ with $\TTime=0.1$, so that $r(\TTime) \approx 0.59$,
are displayed in Table~\ref{tab:HGnew3d}. 

\begin{table}[t]
\center
\begin{tabular}{c|c|c|c|c|c}
 $h_{f}$ & $h^M_\Gamma$ & $\errorUupl$ & $\errorXx$ & $K^M_\Omega$ & $K^M_\Gamma$
 \\
\hline
 1.2500e-01 & 1.1309e-01 & 1.9195e-01 & 5.1473e-02 & 1655 & 770 \\
 6.2500e-02 & 5.9856e-02 & 8.7871e-02 & 2.0037e-02 & 5353 & 3074 \\
 3.1250e-02 & 3.0712e-02 & 2.8850e-02 & 5.2297e-03 & 26221 & 12290 \\
 1.5625e-02 & 1.5464e-02 & 8.3717e-03 & 1.0781e-03 & 356903 & 49154 \\
\end{tabular}
\caption{$\Omega=(-4,4)^3$ and $\TTime=0.1$. Convergence test for 
(\ref{eq:Trueu}) with (\ref{eq:unfit0}).}
\label{tab:HGnew3d}
\end{table} 
In addition, we use the convergence experiment in order to compare the
different strategies (\ref{eq:unfit3}) and (\ref{eq:unfit2}). See
Table~\ref{tab:HGnew3di}, where we present the same computations as in
Table~\ref{tab:HGnew3d}, but now for (\ref{eq:unfit3}) and
(\ref{eq:unfit2}).

\begin{table}[t]
\center
\begin{tabular}{c|c|c|c|c}
 & \multicolumn{2}{c|}{(\ref{eq:unfit3})} & 
 \multicolumn{2}{c}{(\ref{eq:unfit2})} \\ \hline
 $h_{f}$ & $\errorUupl$ & $\errorXx$ & $\errorUupl$ & $\errorXx$ \\ \hline
 1.2500e-01 & 1.9608e-01 & 5.1569e-02 & 1.9699e-01 & 5.1769e-02 \\
 6.2500e-02 & 8.7603e-02 & 2.0783e-02 & 9.0314e-02 & 2.1104e-02 \\
 3.1250e-02 & 2.8999e-02 & 5.9048e-03 & 3.0329e-02 & 6.1388e-03 \\
 1.5625e-02 & 9.3255e-03 & 1.4821e-03 & 9.9485e-03 & 1.6143e-03 \\
\end{tabular}
\caption{$\Omega=(-4,4)^3$ and $\TTime=0.1$. Convergence test for 
(\ref{eq:Trueu}) with (\ref{eq:unfit3}) and (\ref{eq:unfit2}).}
\label{tab:HGnew3di}
\end{table} 
We also compare the numbers in Tables~\ref{tab:HGnew3d} and \ref{tab:HGnew3di}
with
the corresponding errors for the approximation from \cite{dendritic} for the
two-phase Stefan problem with the same choice of parameters. 
The corresponding errors can be seen in Table~\ref{tab:tpsp3d}.

\begin{table}[t]
\center
\begin{tabular}{c|c|c|c|c|c}
 $h_{f}$ & $h^M_\Gamma$ & $\errorUu$ & $\errorXx$ & $K^M_\Omega$ & $K^M_\Gamma$ 
\\
\hline
 1.2500e-01 & 1.1297e-01 & 1.9491e-01 & 5.2057e-02 & 1781 & 770 \\
 6.2500e-02 & 5.9798e-02 & 8.3255e-02 & 2.0582e-02 & 5353 & 3074 \\
 3.1250e-02 & 3.0700e-02 & 2.7380e-02 & 5.4506e-03 & 26221 & 12290 \\
 1.5625e-02 & 1.5462e-02 & 8.1295e-03 & 1.1521e-03 & 356909 & 49154 \\
\end{tabular}
\caption{$\Omega=(-4,4)^3$ and $\TTime=0.1$. Convergence test for the
two-phase Stefan problem.}
\label{tab:tpsp3d}
\end{table}

\subsubsection{One-sided Mullins--Sekerka problem}
We start with a comparison of our algorithm (\ref{eq:uHGa}--c) for the
following exact solution to the one-sided Mullins--Sekerka 
problem (\ref{eq:1a}--e) with $\vartheta= 0$, 
in the case of 
the isotropic surface energy (\ref{eq:iso}).
Here we use the following expanding circle/sphere solution,
where the radius of the circle/sphere is given by $r(t)$. Assume that
$ 
\vartheta= 0 ,
$
$
\conduct = \lambda = \rho = \alpha = a = 1 ,
$ and 
$ f = 0$,  
and let
$
r(t) = (r^2(0) + 2\,t)^\frac12  .
$
Then it is easy to see that the solution $u$ to 
(\ref{eq:1a}--e), 
with $u_D$ in (\ref{eq:1d}) replaced by $u\!\mid_{\partial_D\Omega}$, 
is given by the restriction to $\Omega_+(t)$ of
\begin{equation}
u(\vec{z},t) = 
\begin{cases}
- \frac d{r(t)} & \vec{z} \in \overline\Omega_-(t) , \\
\begin{cases} 
-\ln \frac{|\vec{z}|}{r(t)} - \frac2{r(t)} & d = 2, \\
\frac{r(t)}{|\vec{z}|} - 1 - \frac3{r(t)} & d = 3,
\end{cases} & \vec{z} \in \Omega_+(t).
\end{cases}
\label{eq:TrueuMS}
\end{equation}
For $d=2$, we perform the following convergence experiment for the solution
(\ref{eq:TrueuMS}), where we use $r(0) = R_0 = 1$. 
For $i=0\to 4$, we set $N_f = K^0_\Gamma = 2^{7+i}$, 
$N_c = 4^{i}$, and $\tau= 4^{2-i}\times10^{-3}$. 
The errors $\errorUupl$ and $\errorXx$ on the
interval $[0,\TTime]$ with $\TTime=1$, so that $r(\TTime) \approx 1.73$, 
are displayed in Table~\ref{tab:MS2d}.

\begin{table}[t]
\center
\begin{tabular}{c|c|c|c|c|c}
 $h_{f}$ & $h^M_\Gamma$ & $\errorUupl$ & $\errorXx$ & $K^M_\Omega$ & $K^M_\Gamma$ 
\\
\hline
 6.2500e-02 & 8.5583e-02 & 5.9751e-02 & 1.1650e-02 & 1005 & 128 \\
 3.1250e-02 & 4.2909e-02 & 3.7601e-02 & 1.6311e-02 & 1981 & 256 \\
 1.5625e-02 & 2.1304e-02 & 9.0157e-03 & 4.0322e-03 & 4069 & 512 \\
 7.8125e-03 & 1.0632e-02 & 1.5531e-03 & 6.7227e-04 & 11149 & 1024 \\
 3.9062e-03 & 5.3145e-03 & 4.7394e-04 & 2.0761e-04 & 70733 & 2048 \\
\end{tabular}
\caption{$\Omega=(-4,4)^2$ and $\TTime=1$. Convergence test for 
(\ref{eq:TrueuMS}) with (\ref{eq:unfit0}).}
\label{tab:MS2d}
\end{table} 
In addition, we use the convergence experiment in order to compare the
different strategies (\ref{eq:unfit3}) and (\ref{eq:unfit2}). See
Table~\ref{tab:MS2di}, where we present the same computations as in
Table~\ref{tab:MS2d}, but now for (\ref{eq:unfit3}) and
(\ref{eq:unfit2}).
\begin{table}[t]
\center
\begin{tabular}{c|c|c|c|c}
 & \multicolumn{2}{c|}{(\ref{eq:unfit3})} & 
 \multicolumn{2}{c}{(\ref{eq:unfit2})} \\ \hline
 $h_{f}$ & $\errorUupl$ & $\errorXx$ & $\errorUupl$ & $\errorXx$ \\ \hline
 6.2500e-02 & 7.0732e-02 & 6.1554e-03 & 7.5587e-02 & 5.0174e-03 \\
 3.1250e-02 & 4.1221e-02 & 1.3540e-02 & 4.3588e-02 & 1.2923e-02 \\
 1.5625e-02 & 1.1504e-02 & 2.6877e-03 & 1.2409e-02 & 2.3901e-03 \\
 7.8125e-03 & 3.2383e-03 & 3.8846e-05 & 3.4367e-03 & 1.7735e-04 \\
 3.9062e-03 & 1.2919e-03 & 1.4815e-04 & 1.3623e-03 & 2.1749e-04 \\
\end{tabular}
\caption{$\Omega=(-4,4)^2$ and $\TTime=1$. Convergence test for 
(\ref{eq:TrueuMS}) with (\ref{eq:unfit3}) and (\ref{eq:unfit2}).}
\label{tab:MS2di}
\end{table} 

We also compare the numbers in Tables~\ref{tab:MS2d} and \ref{tab:MS2di} with
the corresponding errors for the approximation from \cite{dendritic} for the
two-phase Mullins--Sekerka problem with the same choice of parameters, when
the function $u(\cdot,t):\Omega \to \R$ from (\ref{eq:TrueuMS}) 
is the desired true solution.
The corresponding errors can be seen in Table~\ref{tab:tpMS}.
\begin{table}[th]
\center
\begin{tabular}{c|c|c|c|c|c}
 $h_{f}$ & $h^M_\Gamma$ & $\errorUu$ & $\errorXx$ & $K^M_\Omega$ & $K^M_\Gamma$ 
\\
\hline
 6.2500e-02 & 8.5582e-02 & 5.5854e-02 & 1.1640e-02 & 1005 & 128 \\
 3.1250e-02 & 4.2910e-02 & 3.3181e-02 & 1.6328e-02 & 1981 & 256 \\
 1.5625e-02 & 2.1305e-02 & 8.6904e-03 & 4.0428e-03 & 4073 & 512 \\
 7.8125e-03 & 1.0632e-02 & 1.5719e-03 & 6.8315e-04 & 11493 & 1024 \\
 3.9062e-03 & 5.3145e-03 & 4.7787e-04 & 2.1309e-04 & 79197 & 2048 \\
\end{tabular}
\caption{$\Omega=(-4,4)^2$ and $\TTime=1$. Convergence test for the
two-phase Mullins--Sekerka problem.}
\label{tab:tpMS}
\end{table} 

Similarly to Table~\ref{tab:MS2d}, 
we perform a convergence experiment for the true solution
(\ref{eq:TrueuMS}) to the one-sided Mullins--Sekerka problem, now for $d=3$,
leaving all the remaining parameters fixed as before.
To this end, for $i=0\to 3$, we set $N_f = 2^{5+i}$, $N_c = 4^{i}$, and 
$K^0_\Gamma=K(i)$, 
where $(K(0), K(1), K(2), K(3)) = (770, 3074, 12290, 49154)$,
and $\tau= 4^{3-i}\times10^{-3}$. The errors $\errorUupl$ and $\errorXx$ on the
interval $[0,\TTime]$ with $\TTime=0.1$, so that $r(\TTime) \approx 1.1$ 
are displayed in Table~\ref{tab:MS3d}.
\begin{table}[t]
\center
\begin{tabular}{c|c|c|c|c|c}
 $h_{f}$ & $h^M_\Gamma$ & $\errorUupl$ & $\errorXx$ & $K^M_\Omega$ & $K^M_\Gamma$ 
\\ \hline
 2.5000e-01 & 2.2637e-01 & 1.8264e-01 & 1.3621e-02 & 1437 & 770 \\
 1.2500e-01 & 1.1441e-01 & 8.2741e-02 & 2.6208e-03 & 4769 & 3074 \\
 6.2500e-02 & 5.7328e-02 & 3.2617e-02 & 8.0637e-04 & 22659 & 12290 \\
 3.1250e-02 & 2.8688e-02 & 5.8383e-03 & 2.4496e-04 & 339431 & 49154 \\
\end{tabular}
\caption{$\Omega=(-4,4)^3$ and $\TTime=0.1$. Convergence test for 
(\ref{eq:TrueuMS}) with (\ref{eq:unfit0}).}
\label{tab:MS3d}
\end{table} 
In addition, we use the convergence experiment in order to compare the
different strategies (\ref{eq:unfit3}) and (\ref{eq:unfit2}). See
Table~\ref{tab:MS3di}, where we present the same computations as in
Table~\ref{tab:MS3d}, but now for (\ref{eq:unfit3}) and
(\ref{eq:unfit2}).
\begin{table}[th]
\center
\begin{tabular}{c|c|c|c|c}
 & \multicolumn{2}{c|}{(\ref{eq:unfit3})} & 
 \multicolumn{2}{c}{(\ref{eq:unfit2})} \\ \hline
 $h_{f}$ & $\errorUupl$ & $\errorXx$ & $\errorUupl$ & $\errorXx$ \\ \hline
 2.5000e-01 & 1.7194e-01 & 1.5249e-02 & 1.7596e-01 & 1.4567e-02 \\
 1.2500e-01 & 7.1850e-02 & 2.3731e-03 & 7.8187e-02 & 2.8742e-03 \\
 6.2500e-02 & 2.9357e-02 & 5.3446e-04 & 3.2027e-02 & 8.1515e-04 \\
 3.1250e-02 & 9.6310e-03 & 2.7820e-04 & 1.0533e-02 & 4.1290e-04 \\
\end{tabular}
\caption{$\Omega=(-4,4)^3$ and $\TTime=0.1$. Convergence test for 
(\ref{eq:TrueuMS}) with (\ref{eq:unfit3}) and (\ref{eq:unfit2}).}
\label{tab:MS3di}
\end{table} 
We also compare the numbers in Tables~\ref{tab:MS3d} and \ref{tab:MS3di} with
the corresponding errors for the approximation from \cite{dendritic} for the
two-phase Mullins--Sekerka problem with the same choice of parameters. 
The corresponding errors can be seen in Table~\ref{tab:tpMS3d}.
\begin{table}[th]
\center
\begin{tabular}{c|c|c|c|c|c}
 $h_{f}$ & $h^M_\Gamma$ & $\errorUu$ & $\errorXx$ & $K^M_\Omega$ & $K^M_\Gamma$ 
\\
\hline
 2.5000e-01 & 2.2681e-01 & 1.8285e-01 & 1.2023e-02 & 1563 & 770 \\
 1.2500e-01 & 1.1458e-01 & 6.7414e-02 & 1.3748e-03 & 4847 & 3074 \\
 6.2500e-02 & 5.7385e-02 & 2.2704e-02 & 7.5695e-04 & 22773 & 12290 \\
 3.1250e-02 & 2.8688e-02 & 6.4026e-03 & 2.5641e-04 & 340087 & 49154 \\
\end{tabular}
\caption{$\Omega=(-4,4)^3$ and $\TTime=0.1$. Convergence test for the
two phase Mullins--Sekerka problem.}
\label{tab:tpMS3d}
\end{table} 

What all of the numerical results in Tables~\ref{tab:HGnew2d}--\ref{tab:tpMS3d}
reveal is that the three strategies (\ref{eq:unfit0},c,d) all behave very
similarly in practice, with the simple strategy (\ref{eq:unfit0}) surprisingly
showing the smallest errors in general. This, combined with the fact that
implementing this strategy requires the fewest computational steps, means that
from now on we will always use (\ref{eq:unfit0}) in our experiments. Lastly we
note that also in the anisotropic setting the different strategies 
(\ref{eq:unfit0},c,d) perform very similarly. For example, 
when we compared the numerical simulations in 
Figure~\ref{fig:2dhex004beta}, below, for the two strategies (\ref{eq:unfit0})
and (\ref{eq:unfit3}), the numerical results were virtually identical.

\subsection{Crystal growth simulations for $d=2$}
Throughout this subsection we use the parameters in (\ref{eq:HGparams})
and $\gamma=\gamma_{hex}$ defined by (\ref{eq:hexgamma2d}) with $\epsilon=0.01$
and $\theta_0=\frac\pi{12}$.
We use this rotation of the anisotropy $\gamma_{hex}$, so that the dominant growth
directions are not exactly aligned with the underlying finite-element meshes
$\mathcal{T}^m$. For the kinetic coefficient we usually set $\beta = \gamma$.
Moreover, the radius of the initial crystal seed $\Gamma(0)$ is always chosen
to be $R_0=0.05$.

We begin with a value of $u_D = -0.004$. 
The results are shown in Figure~\ref{fig:2dhex004}. 
We also show the same experiment for $\beta=1$; 
see Figure~\ref{fig:2dhex004beta}.
\begin{figure}[t]
\center
\mbox{
\hspace{-1.3cm}
 \includegraphics[angle=-90,width=0.5\textwidth]{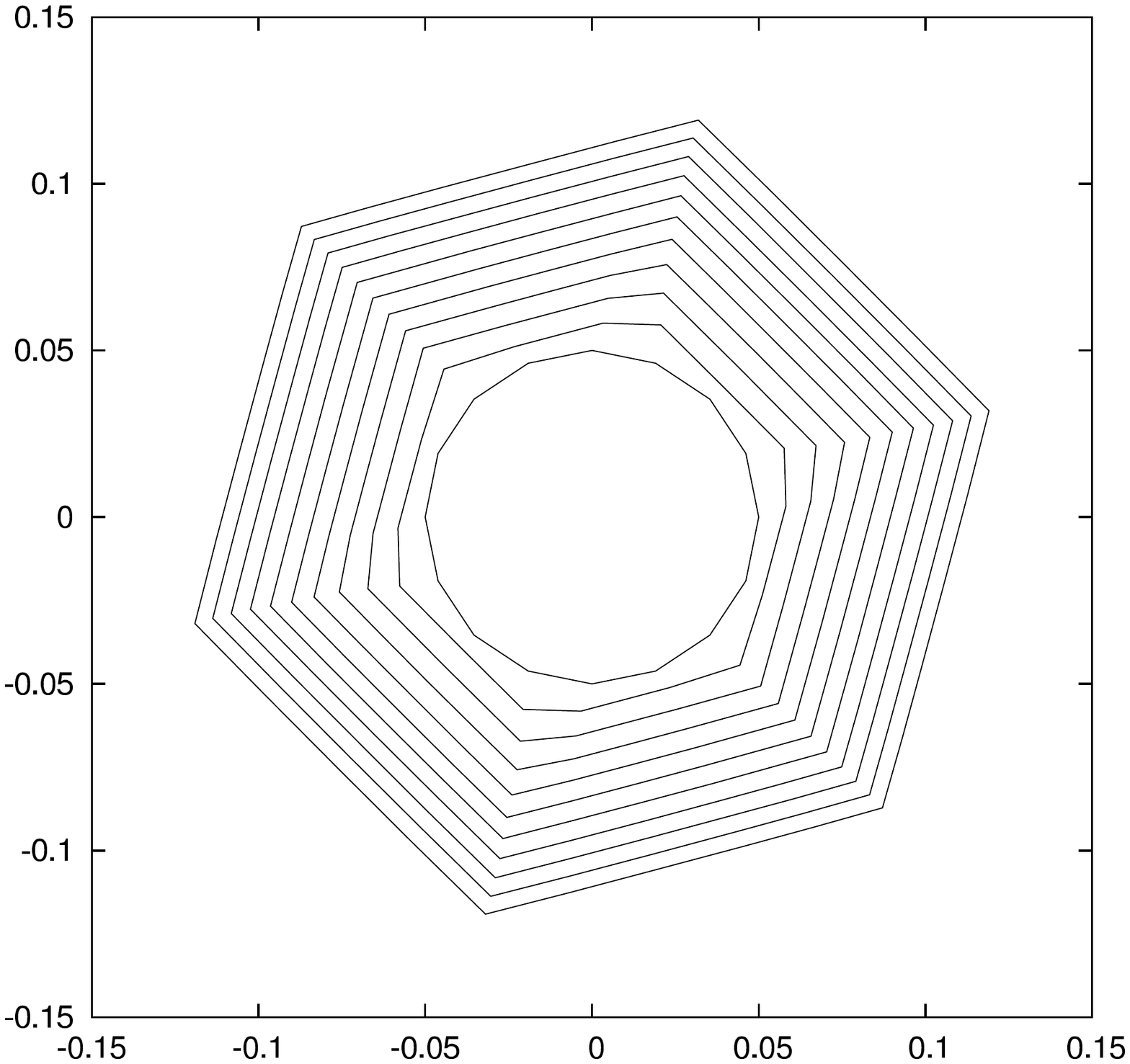}
\hspace{-2.3cm}
 \includegraphics[angle=-90,width=0.5\textwidth]{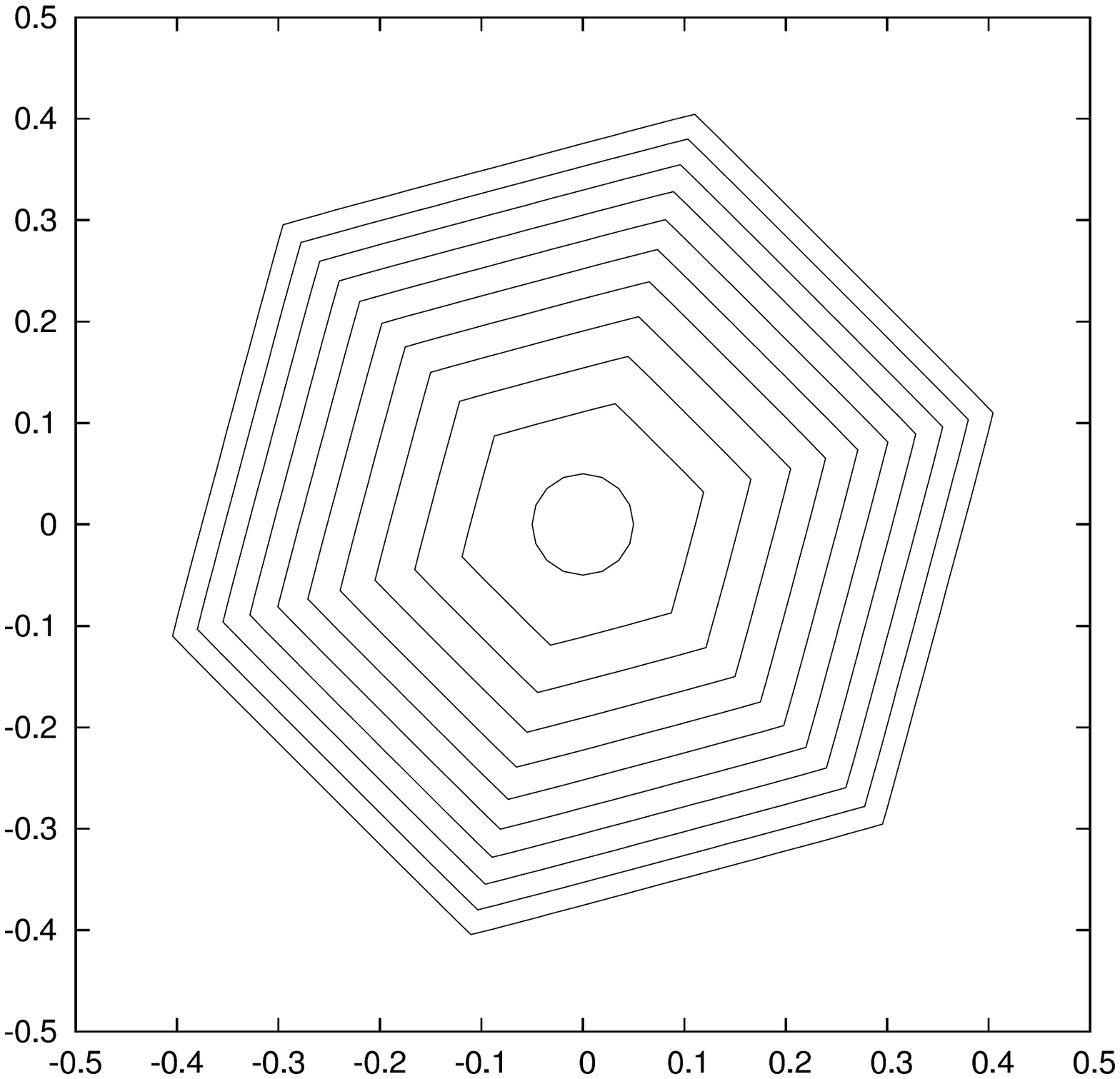}
\hspace{-2.3cm}
 \includegraphics[angle=-90,width=0.5\textwidth]{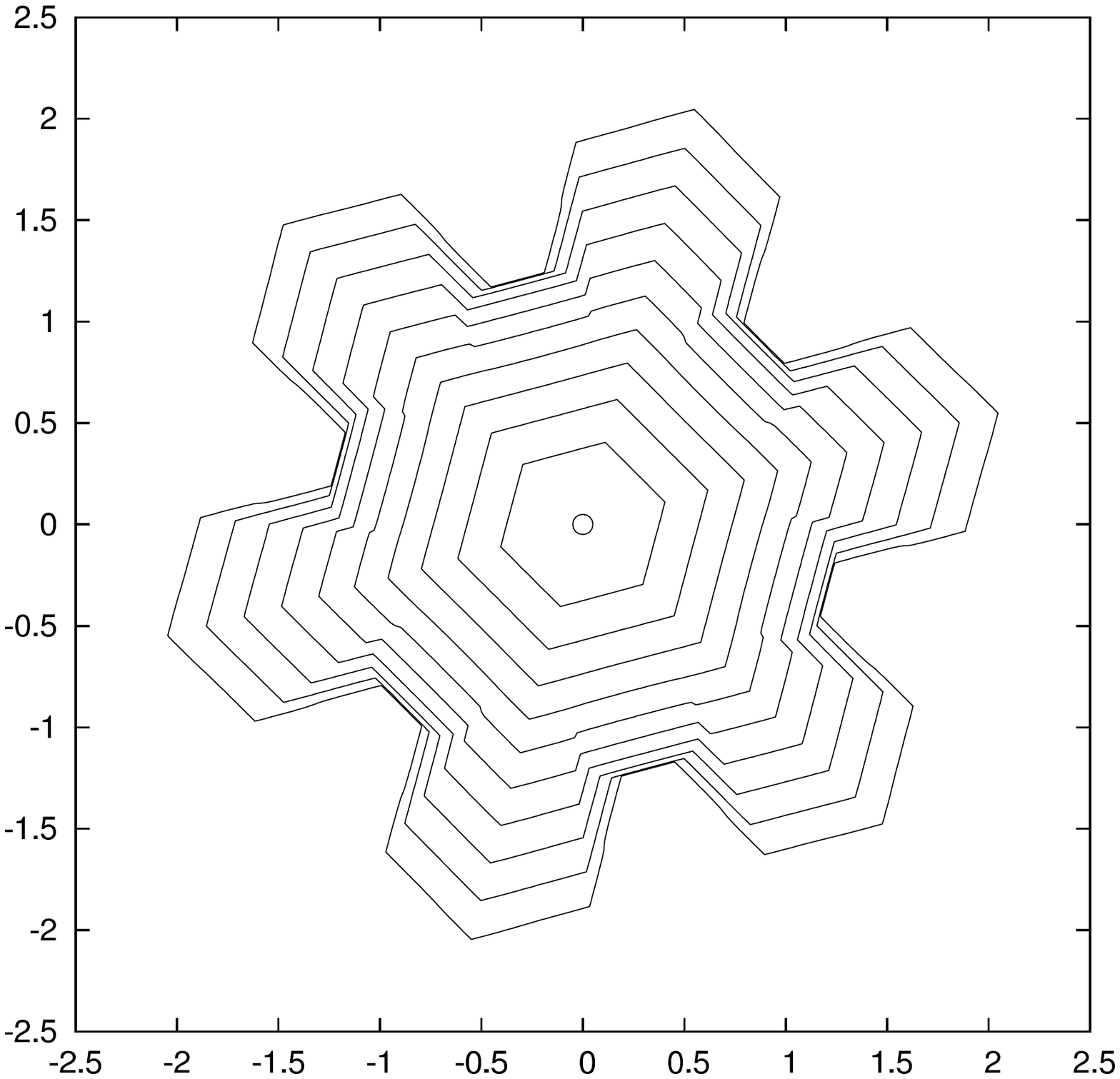}
}
\caption{($\Omega=(-4,4)^2$, $u_D = -0.004$, $\gamma=\beta = \gamma_{hex}$)
$\vec{X}(t)$ for $t=0,\,0.5,\ldots,5$ (left), 
for $t=0,\,5,\ldots,50$ (middle), 
and for $t=0,\,50,\ldots,500$ (right).
Parameters are $N_f=256$, $N_c=4$, $K^0_\Gamma = 16$, and $\tau=0.1$.
}
\label{fig:2dhex004}
\end{figure} 
\begin{figure}
\center
\mbox{
\hspace{-1.3cm}
 \includegraphics[angle=-90,width=0.5\textwidth]{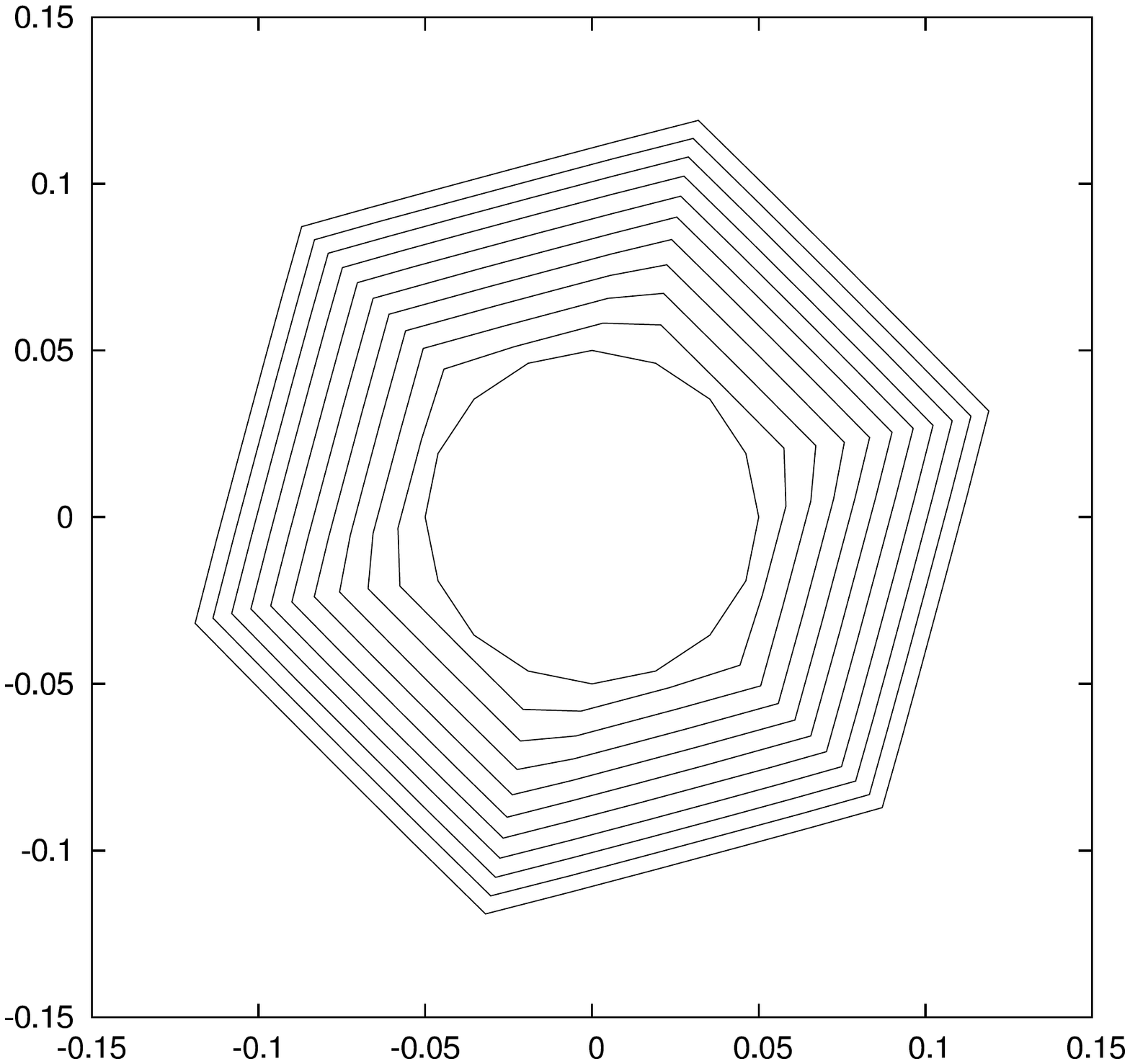}
\hspace{-2.3cm}
 \includegraphics[angle=-90,width=0.5\textwidth]{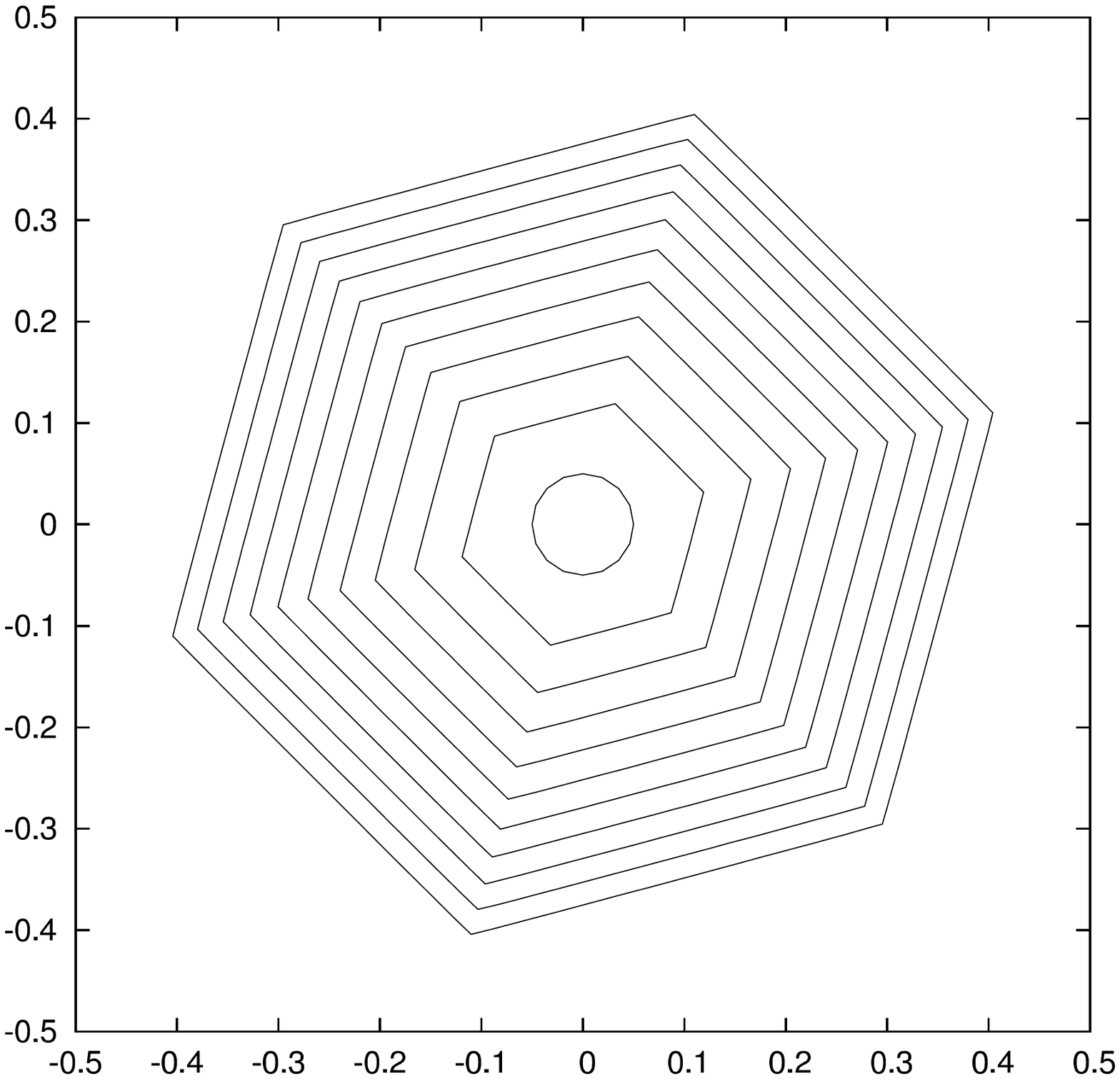}
\hspace{-2.3cm}
 \includegraphics[angle=-90,width=0.5\textwidth]{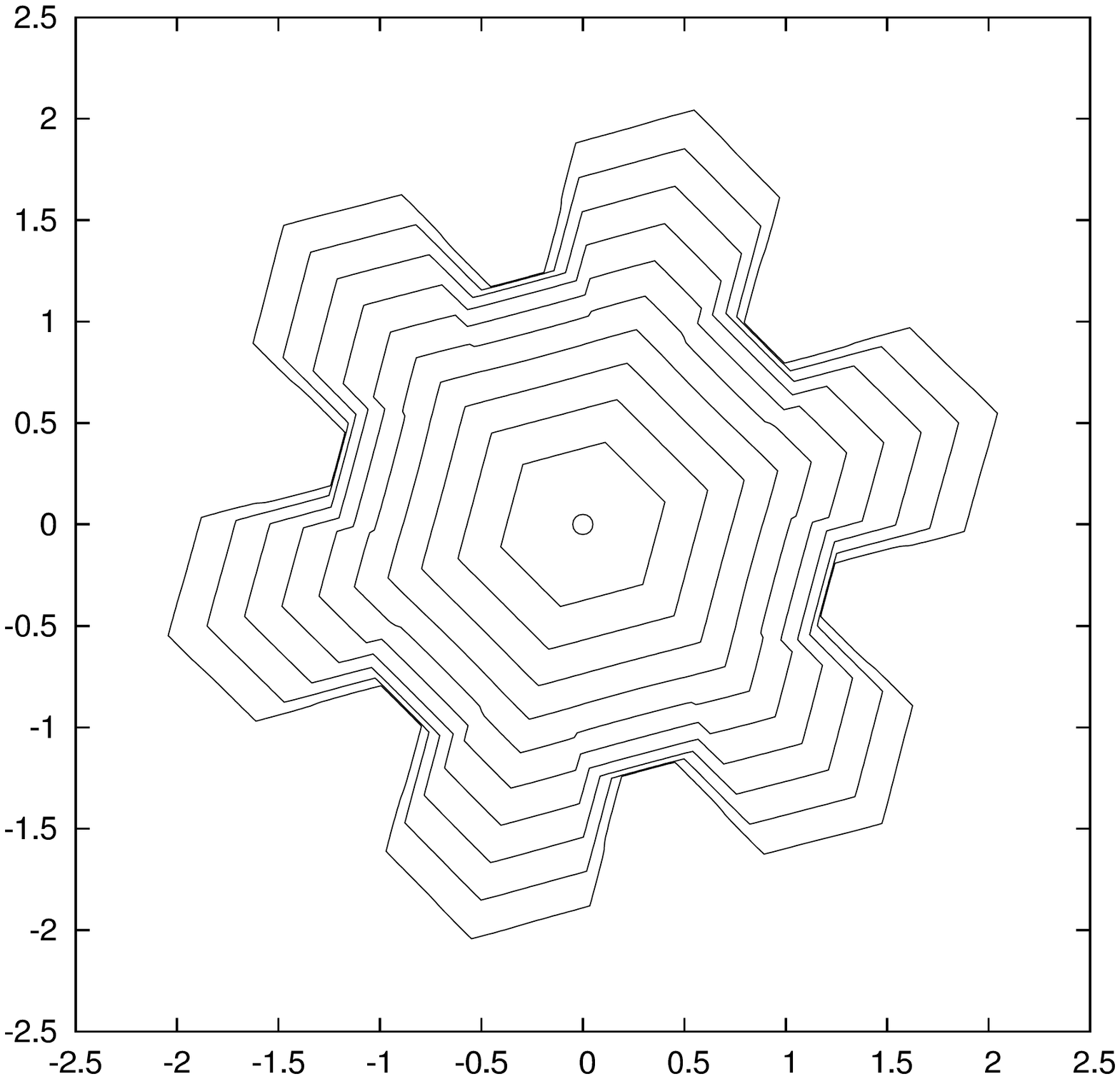}
}
\caption{($\Omega=(-4,4)^2$, $u_D = -0.004$, $\gamma=\gamma_{hex}$, $\beta = 1$)
$\vec{X}(t)$ for $t=0,\,0.5,\ldots,5$ (left), 
for $t=0,\,5,\ldots,50$ (middle), 
and for $t=0,\,50,\ldots,500$ (right).
Parameters are $N_f=256$, $N_c=4$, $K^0_\Gamma = 16$, and $\tau=0.1$.
}
\label{fig:2dhex004beta}
\end{figure} 
We observe that for this experiment, the kinetic coefficient $\beta$ appears to
have hardly any influence on the growth of the crystal. Moreover, we can
observe that the initially circular crystal seed almost immediately assumes a
shape that is favoured by the anisotropy $\gamma$, i.e.,  a shape that is close
to the Wulff shape. This shape then expands at first in a self-similar fashion,
before dendritic arms start to grow at the vertices of the shape.
In order to underline the different effects of $\gamma$ and $\beta$, we compare
the results in Figure~\ref{fig:2dhex004beta} with an experiment where we reverse the
choices of $\gamma$ and $\beta$; i.e.,  we choose an isotropic surface energy
density $\gamma = \gamma_{iso}$ as in (\ref{eq:iso}), while the kinetic
coefficient is defined by $\beta = \gamma_{hex}$; recall (\ref{eq:hexgamma2d}). 
The numerical results for this experiment can be seen in
Figure~\ref{fig:2dhex004beta10}.
\begin{figure}[t]
\center
\mbox{
\hspace{-1.3cm}
 \includegraphics[angle=-90,width=0.5\textwidth]{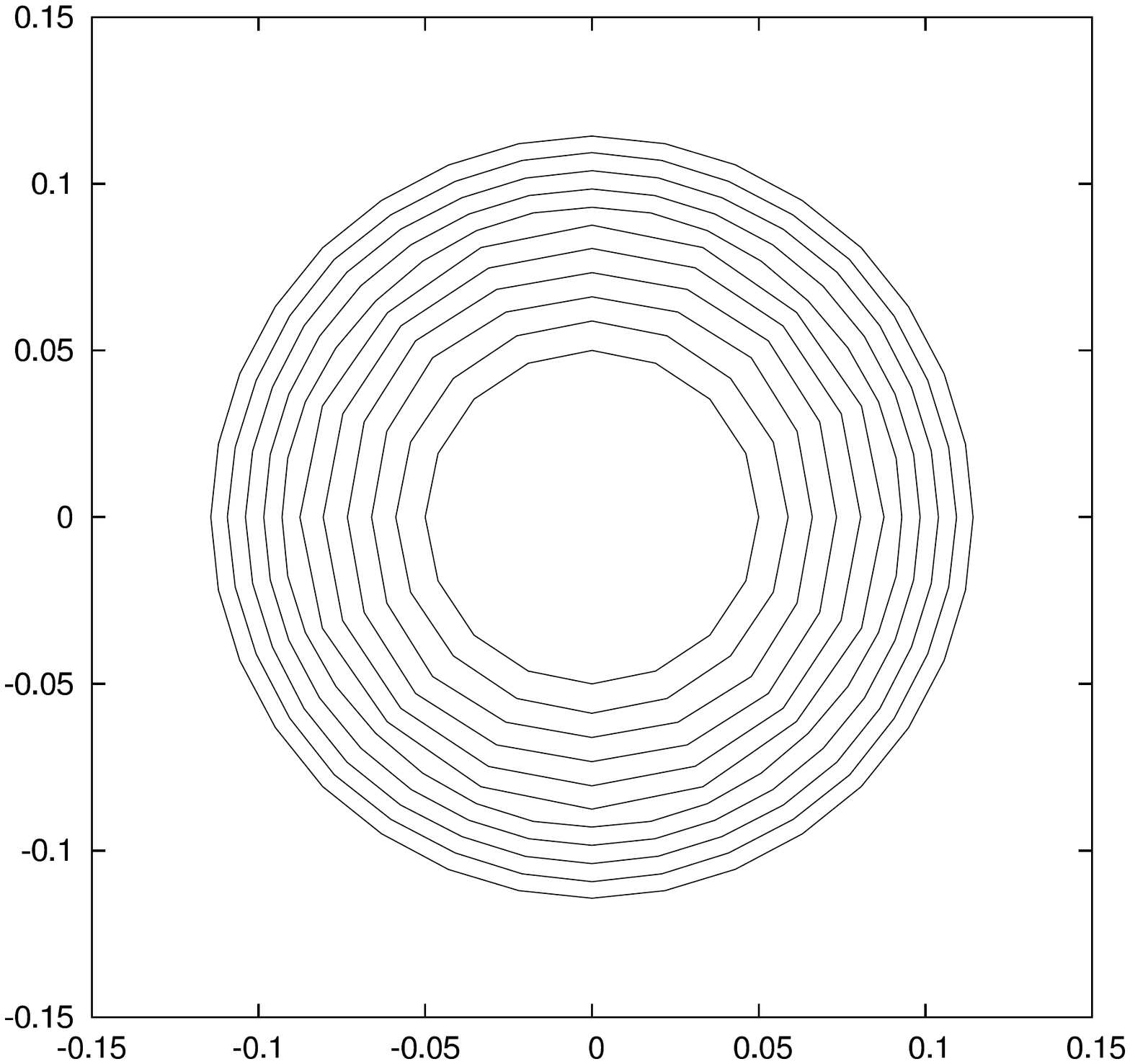}
\hspace{-2.3cm}
 \includegraphics[angle=-90,width=0.5\textwidth]{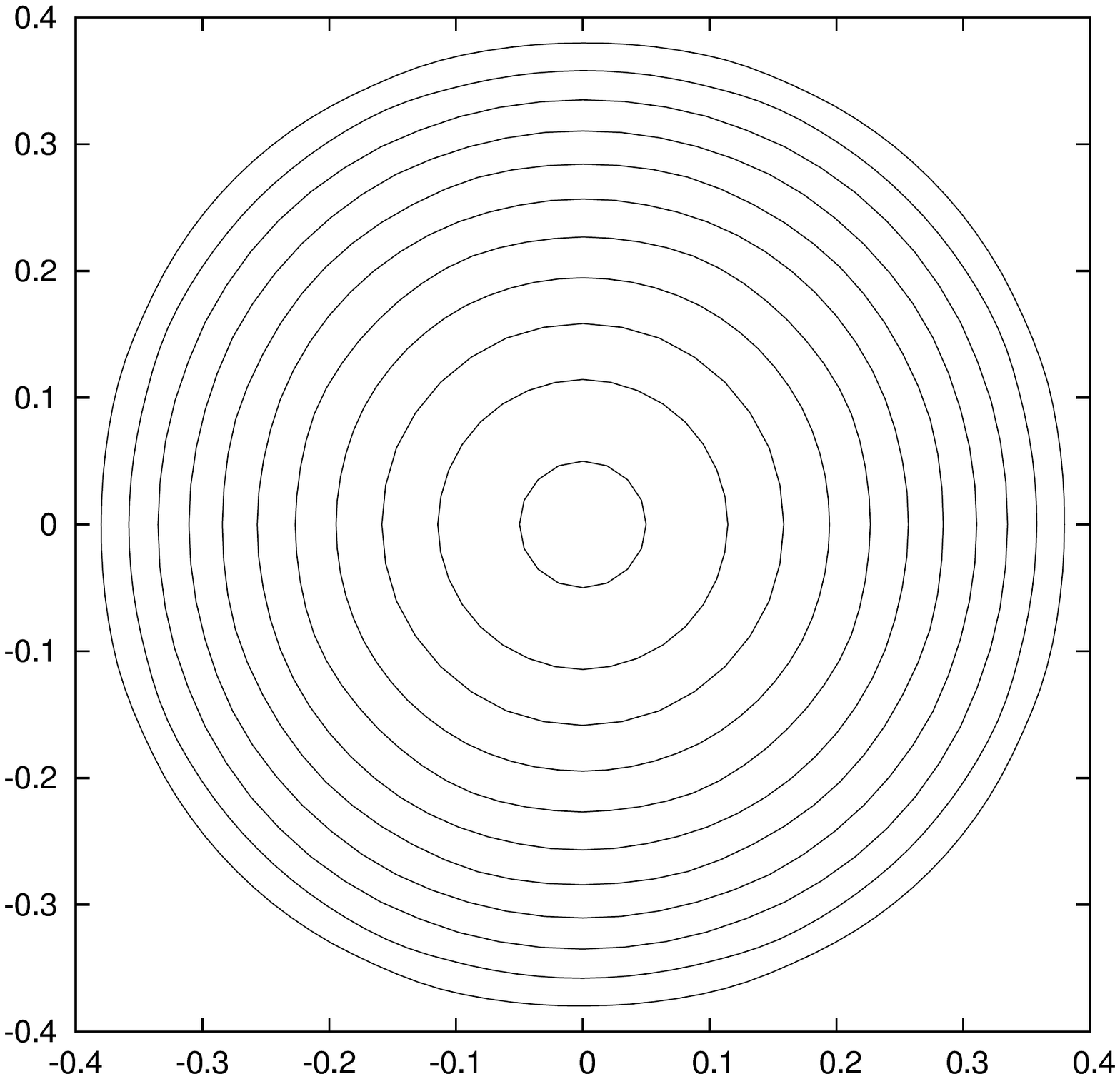}
\hspace{-2.3cm}
 \includegraphics[angle=-90,width=0.5\textwidth]{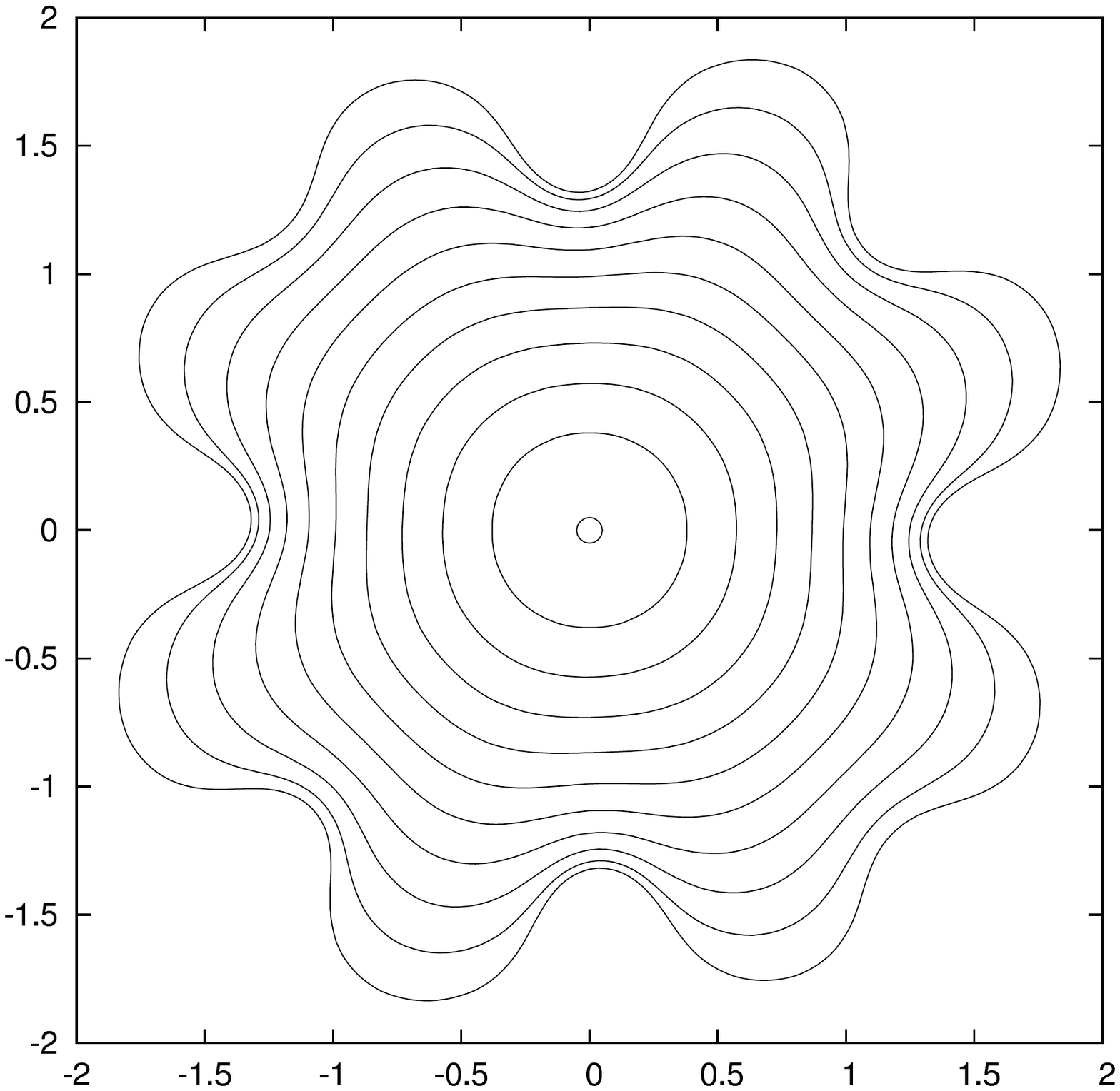}
}
\caption{($\Omega=(-4,4)^2$, $u_D = -0.004$, $\gamma = \gamma_{iso}$, 
$\beta = \gamma_{hex}$)
$\vec{X}(t)$ for $t=0,\,0.5,\ldots,5$ (left), 
for $t=0,\,5,\ldots,50$ (middle), 
and for $t=0,\,50,\ldots,500$ (right).
Parameters are $N_f=256$, $N_c=4$, $K^0_\Gamma = 16$, and $\tau=0.1$.
}
\label{fig:2dhex004beta10}
\end{figure} 

Before we look at experiments with larger values of $|u_D|$, we present the
results for a run with $u_D = -0.004$, but now run on the larger domain
$\Omega=(-8,8)^2$ and until the later time $T=2500$. See
Figure~\ref{fig:2dhex004long} for the results, where the different effects of
$\gamma$ and $\beta$ are once again visible.
In fact, the results for the isotropic surface energy $\gamma =
\gamma_{iso}$ seem to indicate that the orientation of the underlying finite
element mesh has a larger influence on the directions, in which the unstable
interface grows, than the kinetic coefficient $\beta = \gamma_{hex}$ itself.
To confirm this interpretation, we present a further comparison. This time, we
choose all coefficients as isotropic, so that $\gamma = \gamma_{iso}$ and $\beta = 1$. 
The corresponding result is shown on the right of
Figure~\ref{fig:2dhex004long}. Once again it appears that the role that $\beta$
plays here is insignificant. We observe that in the case that $\gamma$
is isotropic a tip-splitting instability occurs.
\begin{figure}[t]
\center
\mbox{
\hspace{-1.3cm}
 \includegraphics[angle=-90,width=0.5\textwidth]{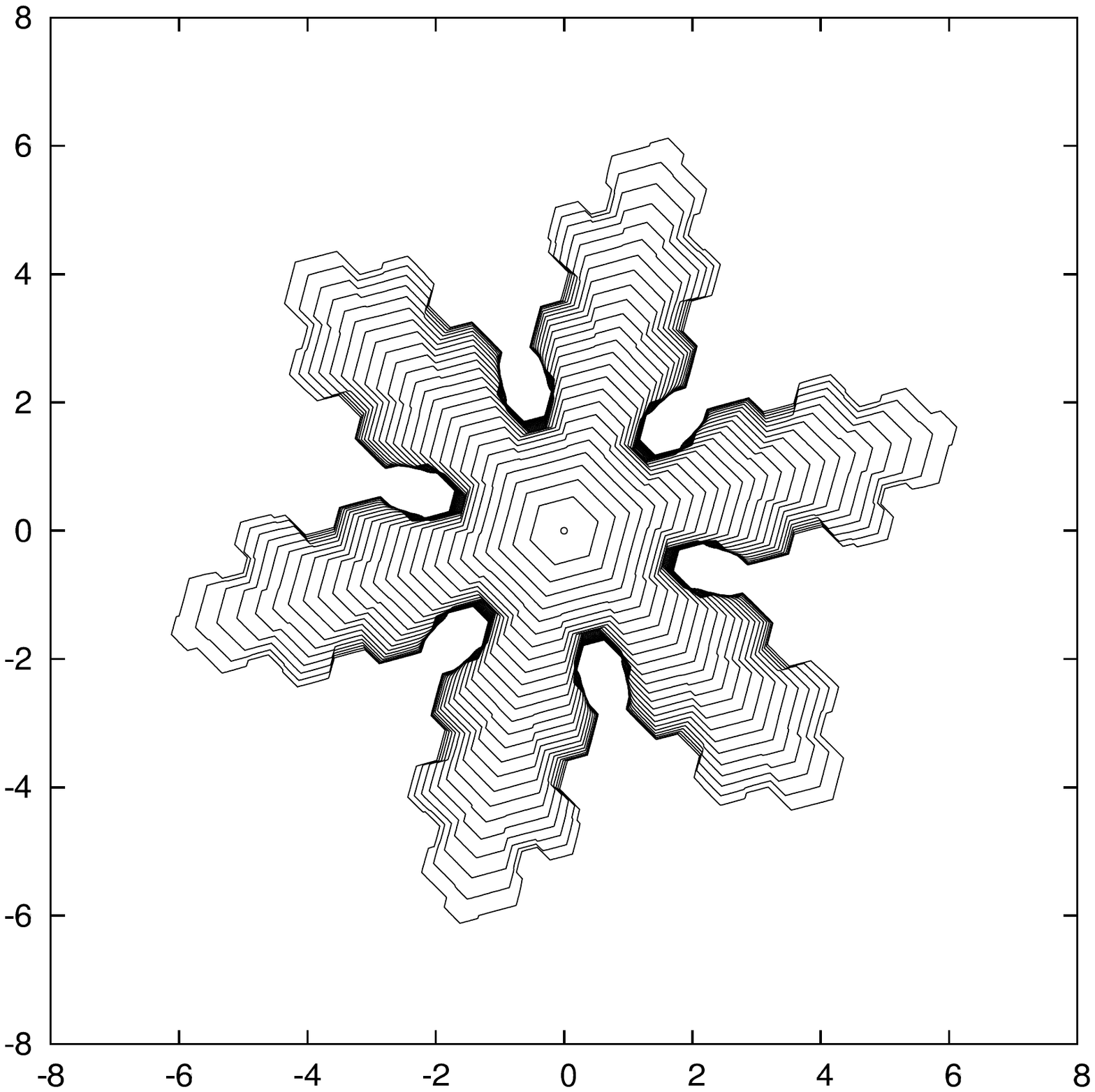}
\hspace{-2.3cm}
 \includegraphics[angle=-90,width=0.5\textwidth]{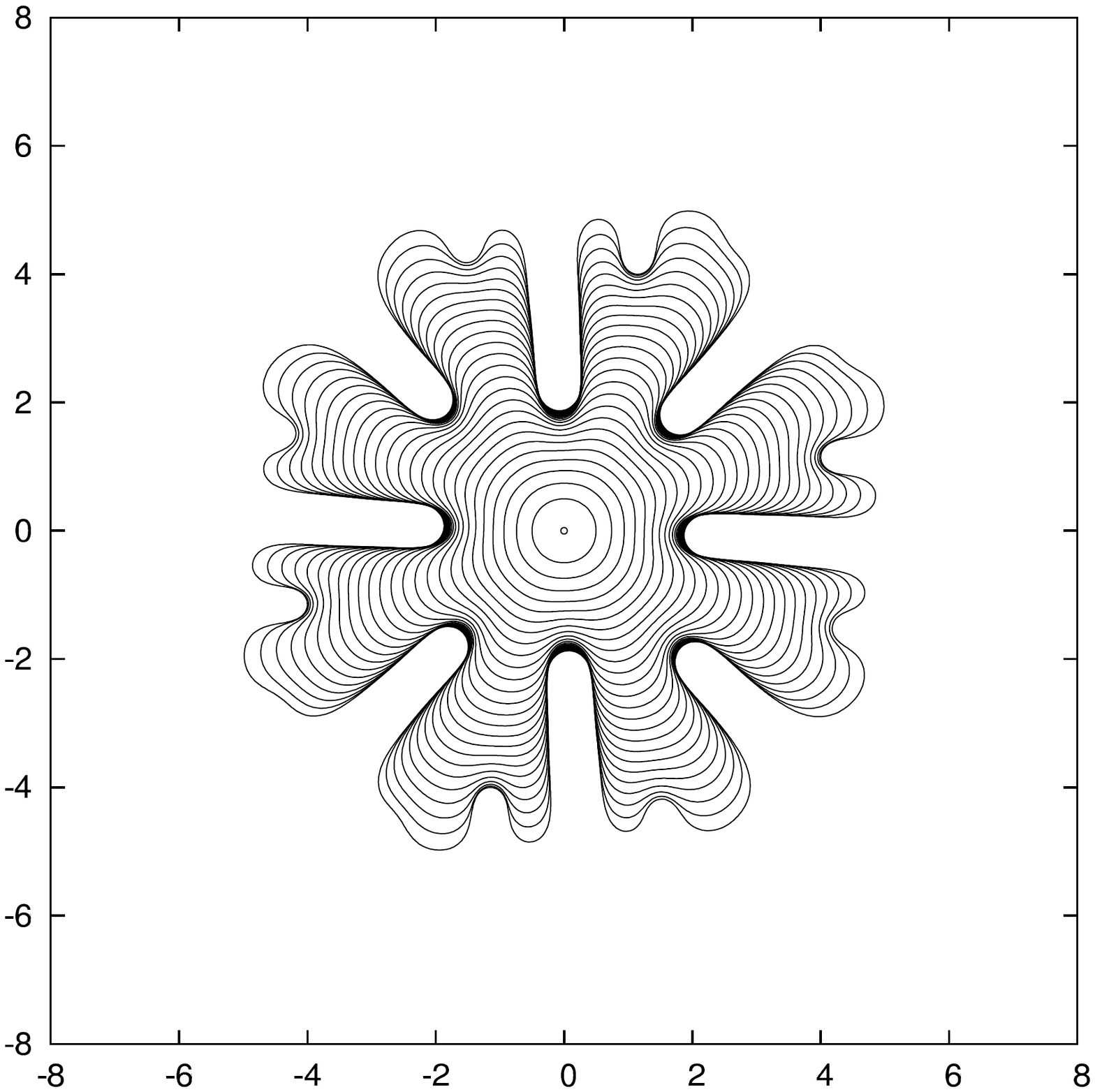}
\hspace{-2.3cm}
 \includegraphics[angle=-90,width=0.5\textwidth]{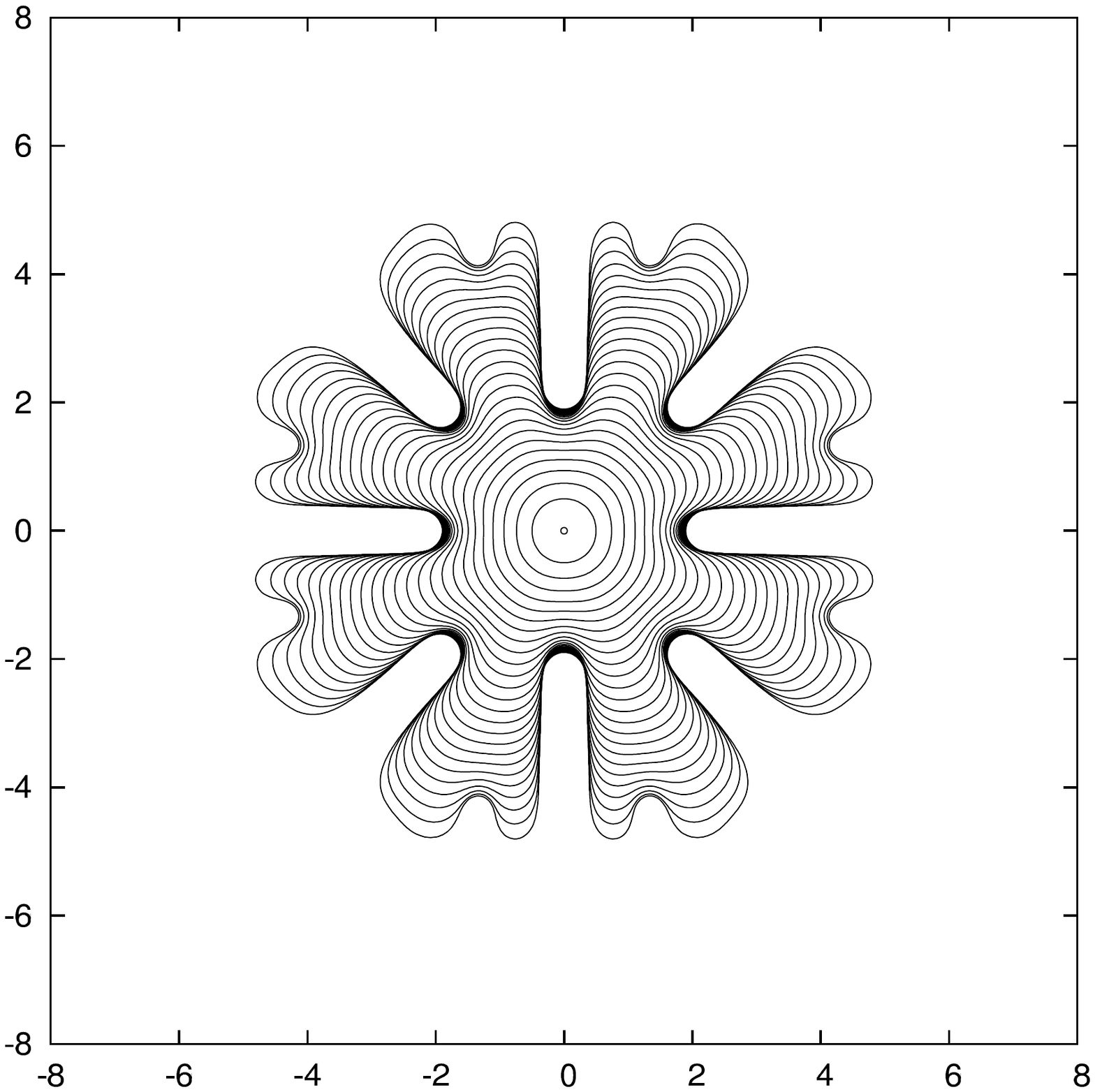} 
}
\caption{($\Omega=(-8,8)^2$, $u_D = -0.004$, 
 $\gamma = \gamma_{hex}$, $\beta = 1$ (left), $\gamma =
\gamma_{iso}$, $\beta = \gamma_{hex}$ (middle),  $\gamma = \gamma_{iso}$, $\beta = 1$, (right))
$\vec{X}(t)$ for $t=0,\,100,\ldots,2500$.
Parameters are $N_f=512$, $N_c=8$, $K^0_\Gamma = 16$, and $\tau=0.1$.
}
\label{fig:2dhex004long}
\end{figure}

In the next experiment, we set $u_D = -0.01$ for $\gamma=\beta=\gamma_{hex}$. The
results are shown in Figure~\ref{fig:2dhex01coarse}, and we observe
that a larger supersaturation enhances the unstable behaviour. 
\begin{figure}[ht]
\center
\mbox{
 \includegraphics[angle=-90,width=0.5\textwidth]{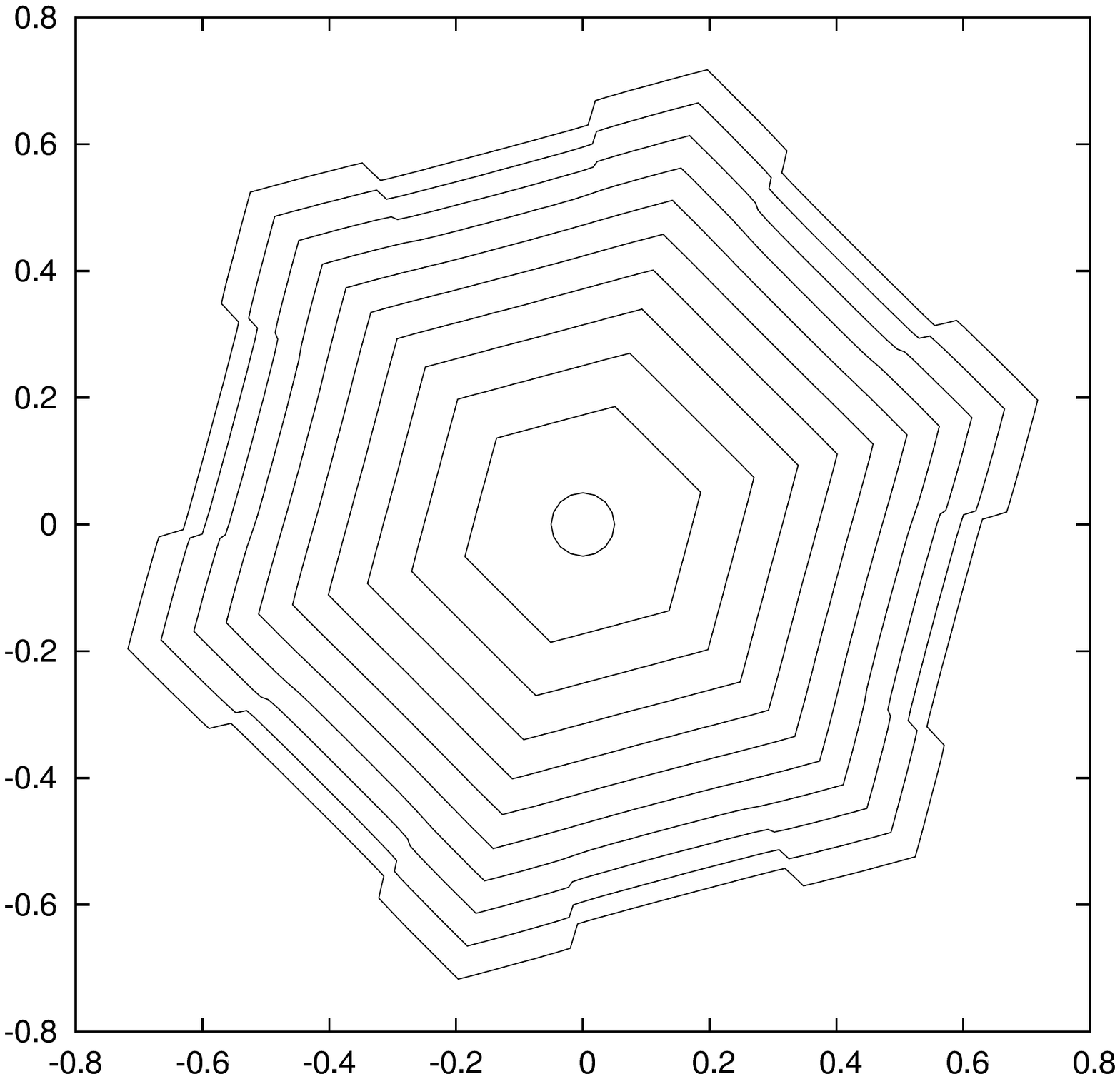}
\quad 
 \includegraphics[angle=-90,width=0.5\textwidth]{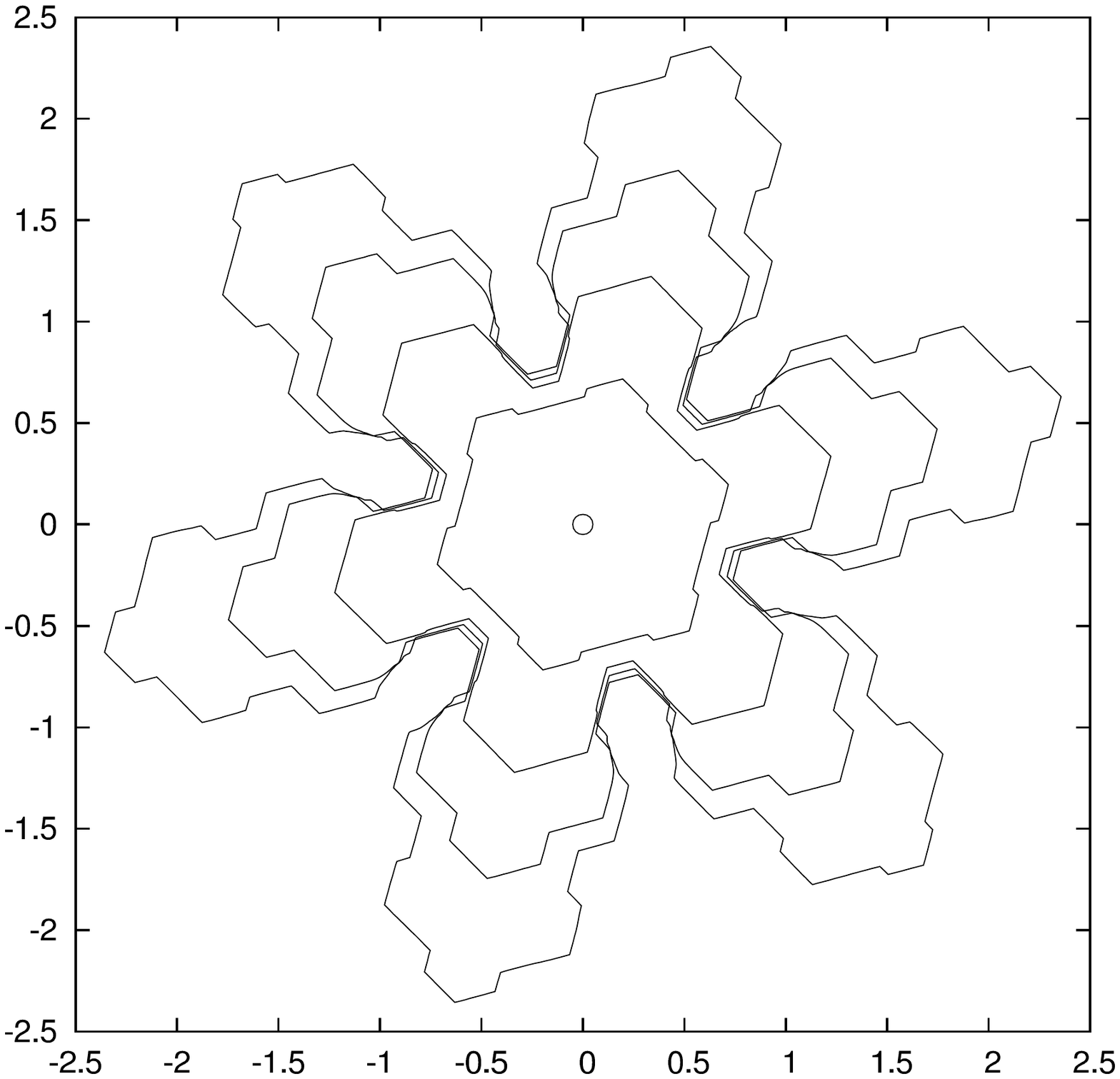}
}
\caption{($\Omega=(-4,4)^2$, $u_D = -0.01$, $\gamma=\beta=\gamma_{hex}$)
$\vec{X}(t)$ for $t=0,\,5,\ldots,50$ (left), 
and for $t=0,\,50,\ldots,200$ (right).
Parameters are $N_f=512$, $N_c= K^0_\Gamma = 16$, and $\tau=5\times10^{-3}$.
}
\label{fig:2dhex01coarse}
\end{figure} 

In the next experiment, we set $u_D = -0.04$. The
results are shown in Figure~\ref{fig:2dhex04}.
\begin{figure}[ht]
\center
\mbox{
 \includegraphics[angle=-90,width=0.5\textwidth]{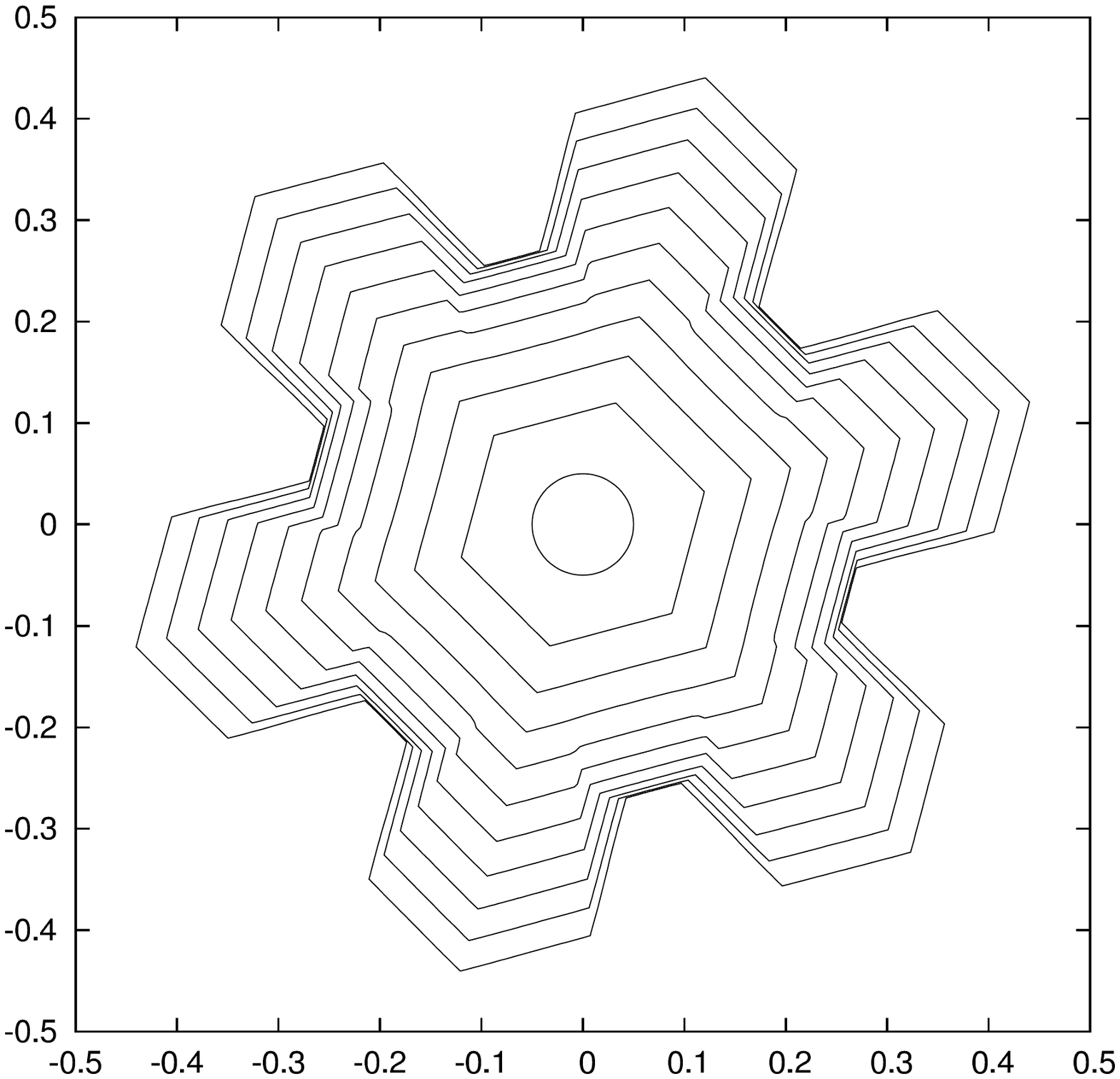}  
\quad
 \includegraphics[angle=-90,width=0.5\textwidth]{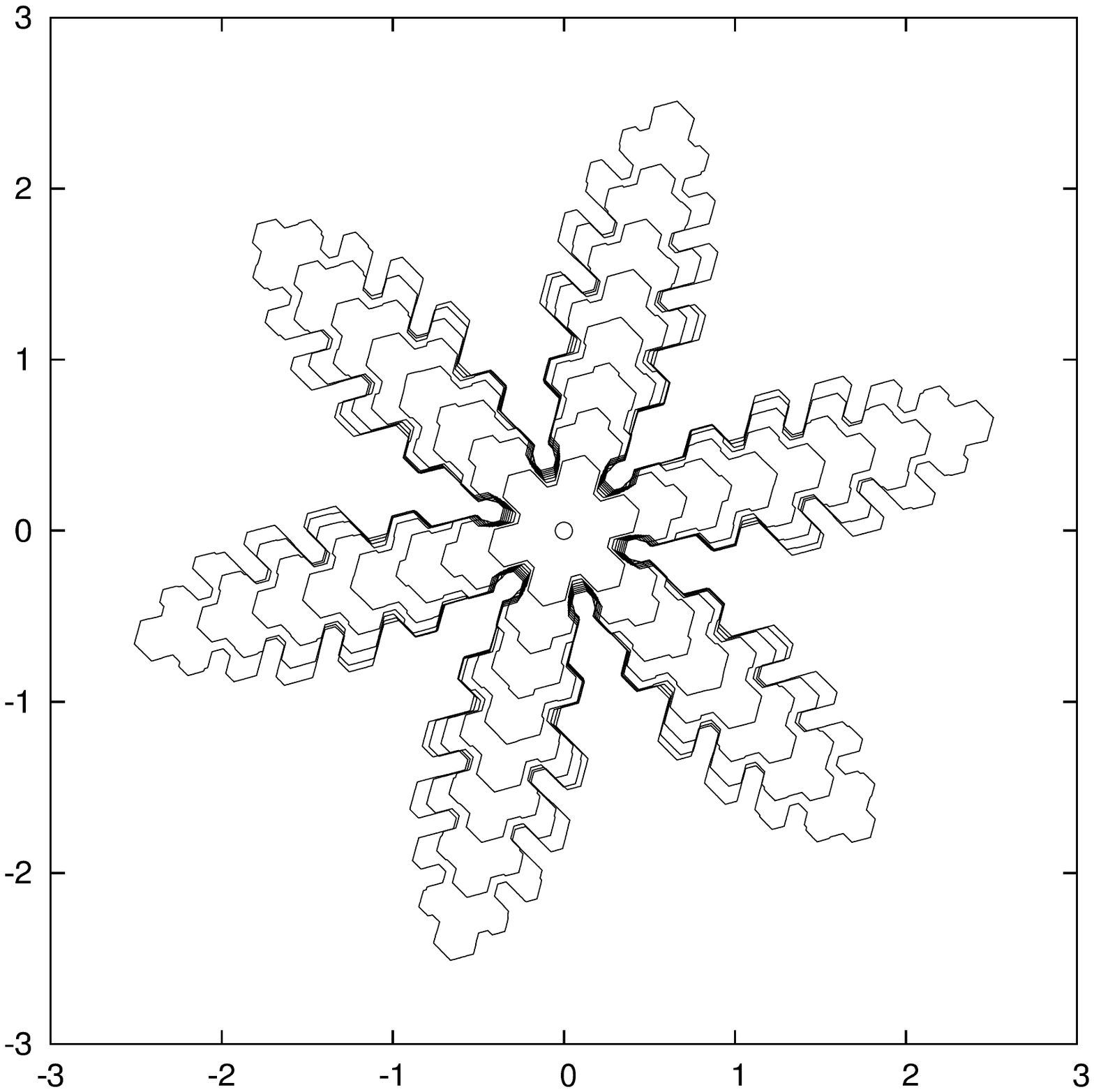}
}
\caption{($\Omega=(-4,4)^2$, $u_D = -0.04$, $\gamma=\beta=\gamma_{hex}$)
$\vec{X}(t)$ for $t=0,\,0.5,\ldots,5$ (left), 
and for $t=0,\,5,\ldots,40$ (right).
Parameters are $N_f=1024$, $N_c=K^0_\Gamma = 64$, and $\tau=2.5\times10^{-3}$.
}
\label{fig:2dhex04}
\end{figure} 
The distribution of $U$ at time $t=40$ can be seen in
Figure~\ref{fig:2dhex04_tempX}. Here we note that, according to the definitions
(\ref{eq:Sml}), in these plots $U$  is set to zero
inside the solid phase.
\begin{figure}[ht]
\center
\mbox{
\hspace{-1.3cm}
 \includegraphics[angle=-90,width=0.45\textwidth]{figures/2dhex_c04_T40_ps}
\hspace{-1.3cm}
 \includegraphics[angle=-90,width=0.33\textwidth]{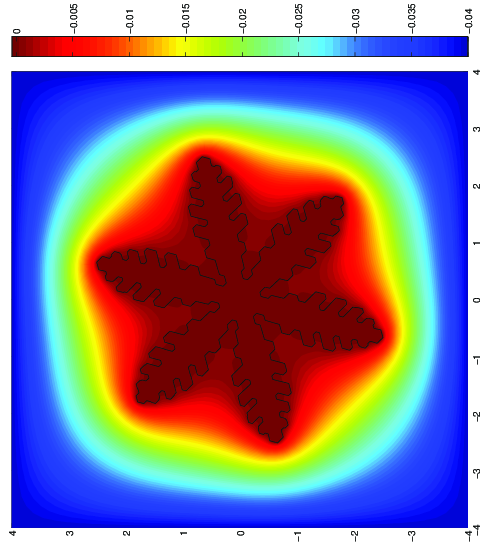}
 \includegraphics[angle=-90,width=0.33\textwidth]{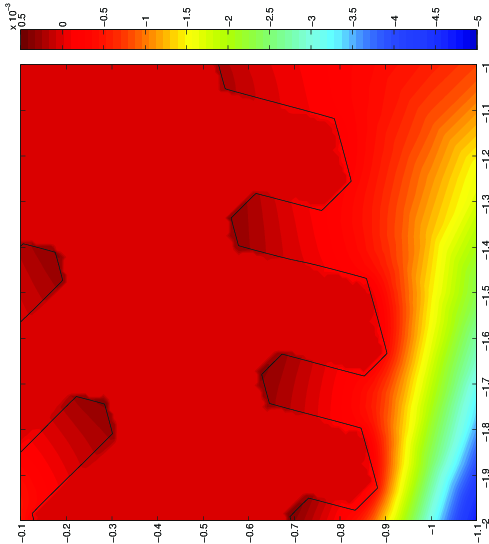}
}
\caption{(One-sided Mullins--Sekerka problem)
$\vec{X}(t)$ for for $t=0,\,5,\ldots,40$ (left).
$\vec X(t)$ and $U(t)$ for $t=40$ on $[-4,4]^2$ (middle) and on
$[-2,-1]\times[-1.1,-0.1]$ (right).
}
\label{fig:2dhex04_tempX}
\end{figure} 

As a comparison, we repeat the same experiment as in
Figure~\ref{fig:2dhex04_tempX} now for (i) the one-sided Stefan problem,
(ii) the two-sided Mullins--Sekerka problem, and (iii) the two-sided Stefan
problem with $\vartheta=1$ for the Stefan problems.
Note that for (ii) and (iii) we employ the finite-element approximation
from \cite{dendritic}, while for (i) we use (\ref{eq:uHGa}--c) with
$\vartheta=1$. The corresponding plots are shown in 
Figures~\ref{fig:2dhex04_ossp}--\ref{fig:2dhex04_tssp}. We observe that
the difference between the one-sided and the two-sided problems is not
very pronounced, but one notices that the sidearms in the two-sided
problems grow more slowly due to the fact that diffusion into the crystal
is possible. 
\begin{figure}[ht]
\center
\mbox{
\hspace{-1.3cm}
 \includegraphics[angle=-90,width=0.45\textwidth]{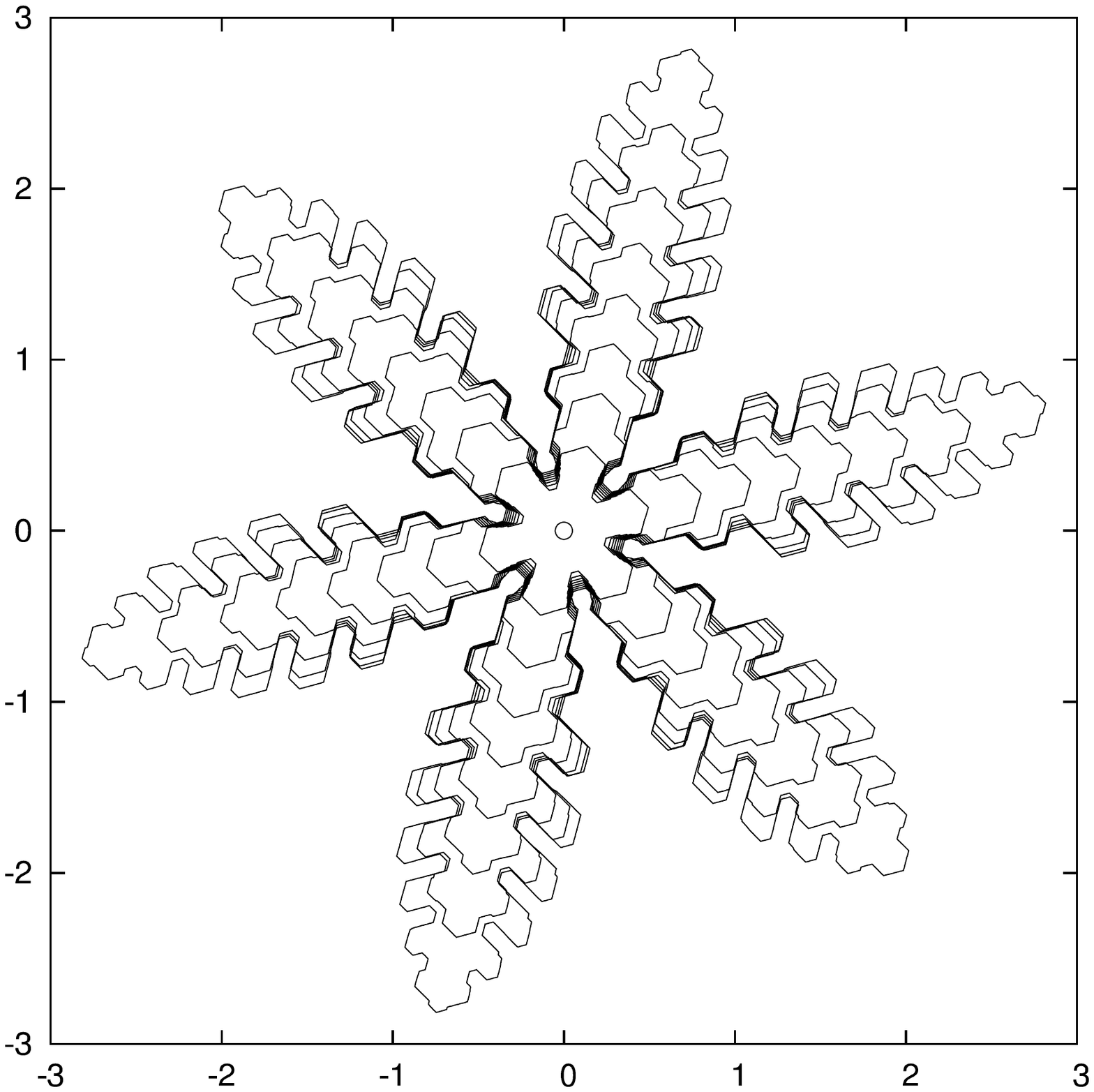}
\hspace{-1.3cm}
 \includegraphics[angle=-90,width=0.33\textwidth]{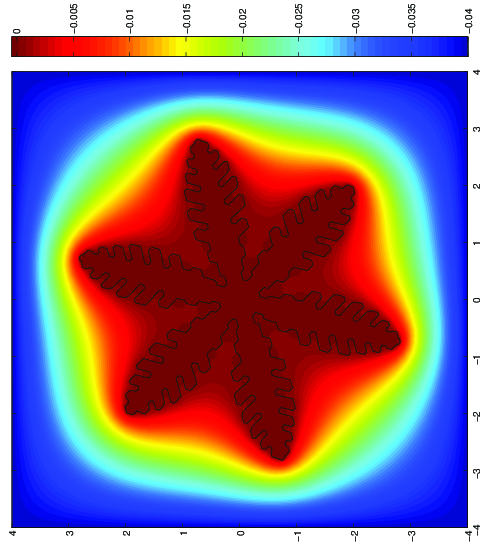}
 \includegraphics[angle=-90,width=0.33\textwidth]{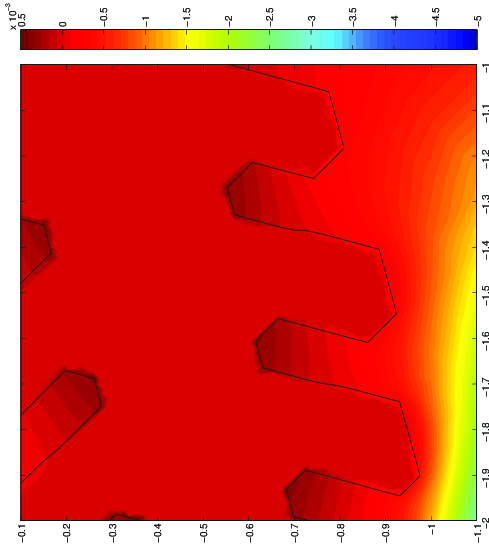}
}
\caption{(One-sided Stefan problem) 
$\vec{X}(t)$ for for $t=0,\,5,\ldots,40$ (left).
$\vec X(t)$ and $U(t)$ for $t=40$ on $[-4,4]^2$ (middle) and on
$[-2,-1]\times[-1.1,-0.1]$ (right).}
\label{fig:2dhex04_ossp}
\end{figure} 
\begin{figure}[ht]
\center
\mbox{
\hspace{-1.3cm}
 \includegraphics[angle=-90,width=0.45\textwidth]{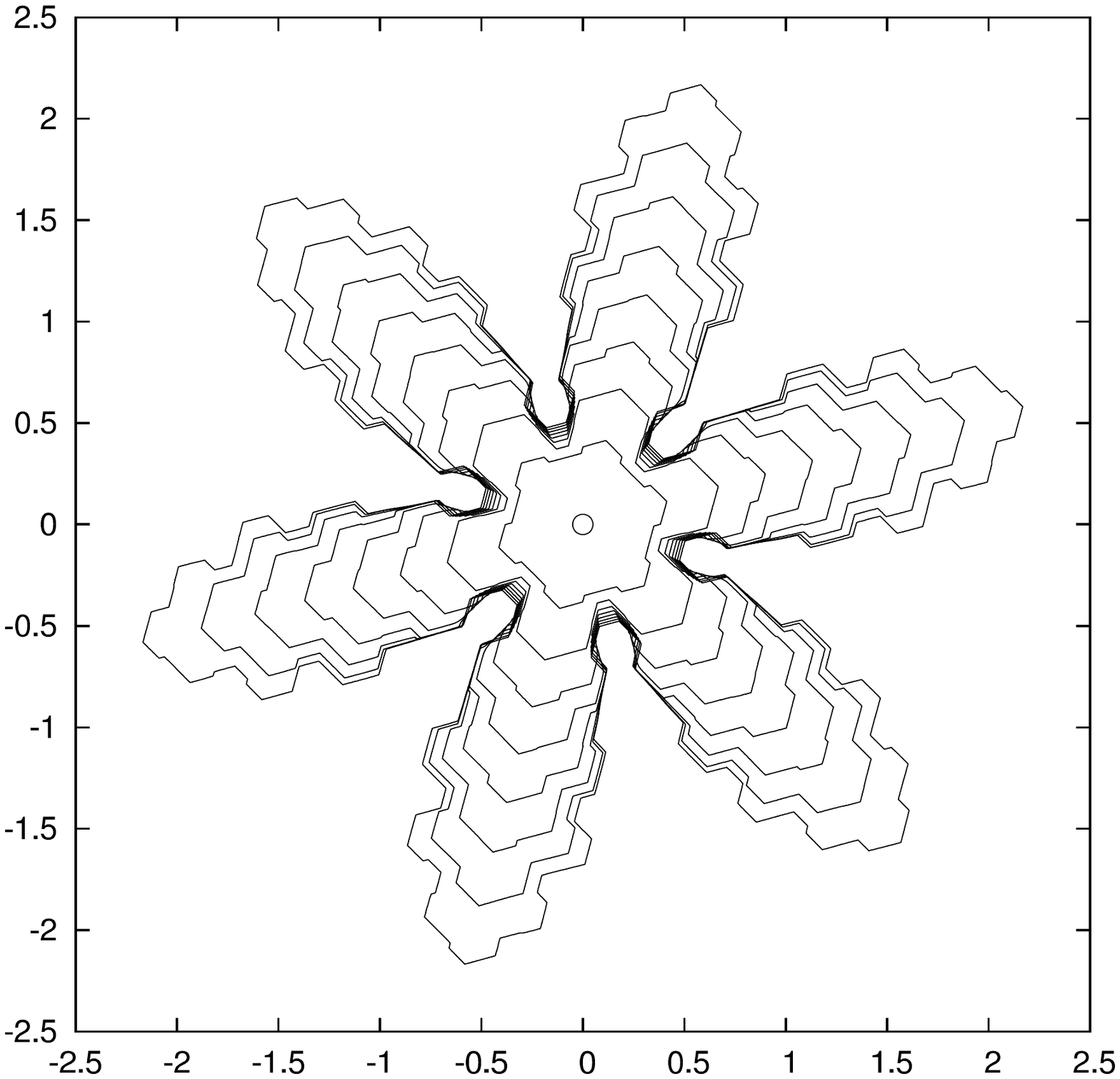}
\hspace{-1.3cm}
 \includegraphics[angle=-90,width=0.33\textwidth]{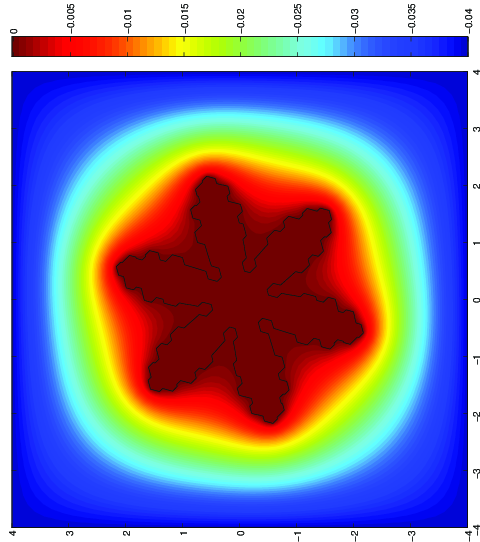}
 \includegraphics[angle=-90,width=0.33\textwidth]{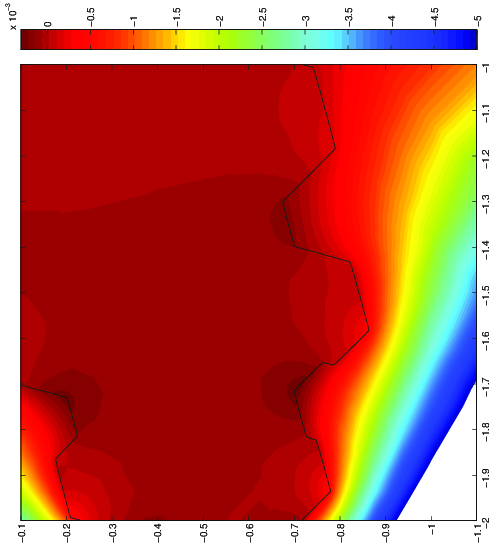}
}
\caption{(Two-sided Mullins--Sekerka problem)
$\vec{X}(t)$ for for $t=0,\,5,\ldots,40$ (left).
$\vec X(t)$ and $U(t)$ for $t=40$ on $[-4,4]^2$ (middle) and on
$[-2,-1]\times[-1.1,-0.1]$ (right).}
\label{fig:2dhex04_tsms}
\end{figure} 
\begin{figure}[ht]
\center
\mbox{
\hspace{-1.3cm}
 \includegraphics[angle=-90,width=0.45\textwidth]{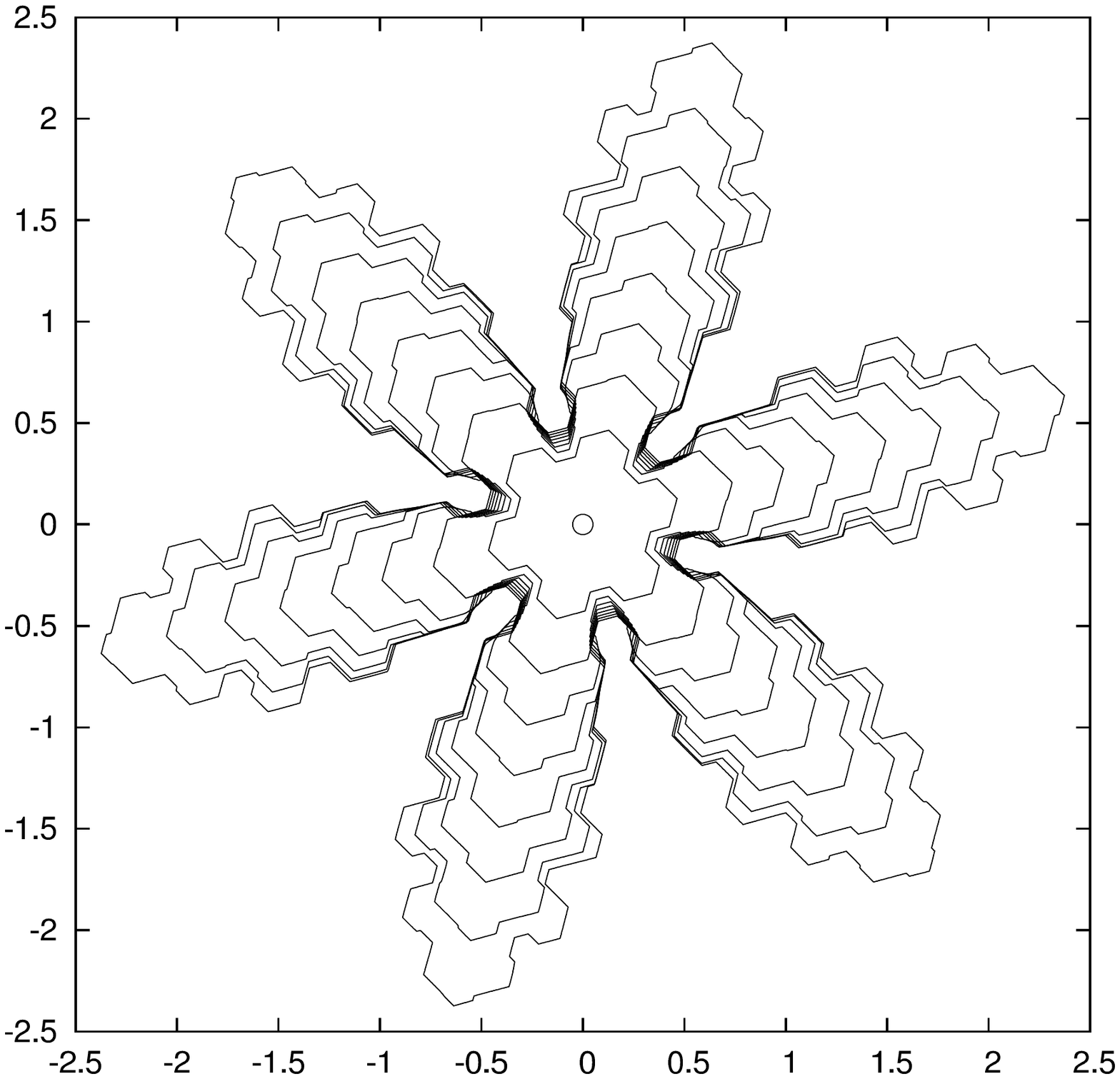}
\hspace{-1.3cm}
 \includegraphics[angle=-90,width=0.33\textwidth]{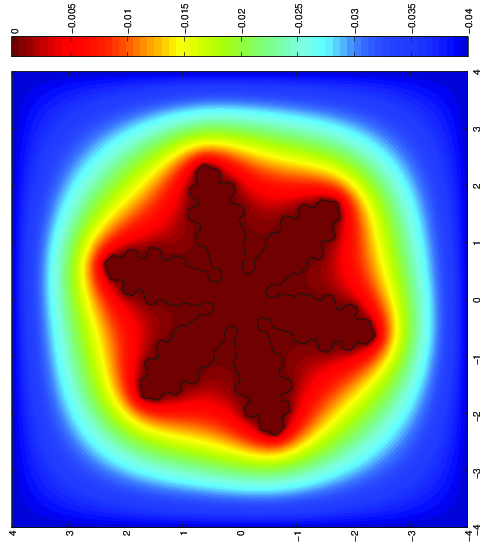}
 \includegraphics[angle=-90,width=0.33\textwidth]{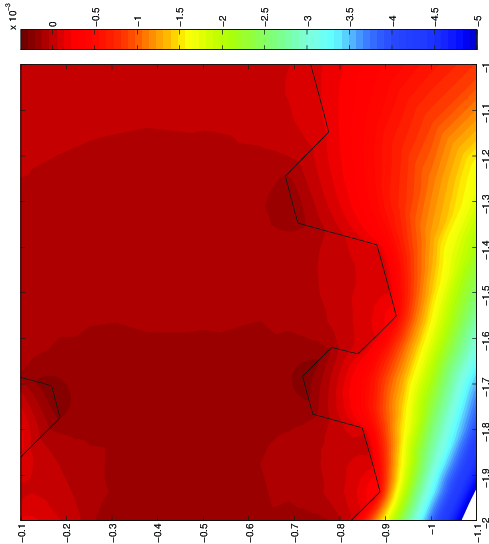}
}
\caption{(Two-sided Stefan problem)
$\vec{X}(t)$ for for $t=0,\,5,\ldots,40$ (left).
$\vec X(t)$ and $U(t)$ for $t=40$ on $[-4,4]^2$ (middle) and on
$[-2,-1]\times[-1.1,-0.1]$ (right).}
\label{fig:2dhex04_tssp}
\end{figure} 

In the final experiments for $d=2$, we return to the one-sided Mullins--Sekerka problem and 
set $u_D = -0.08$ and $u_D=-0.2$. The
results are shown in Figures~\ref{fig:2dhex08} and \ref{fig:2dhex2}, respectively.
\begin{figure}[ht]
\center
\mbox{
 \includegraphics[angle=-90,width=0.5\textwidth]{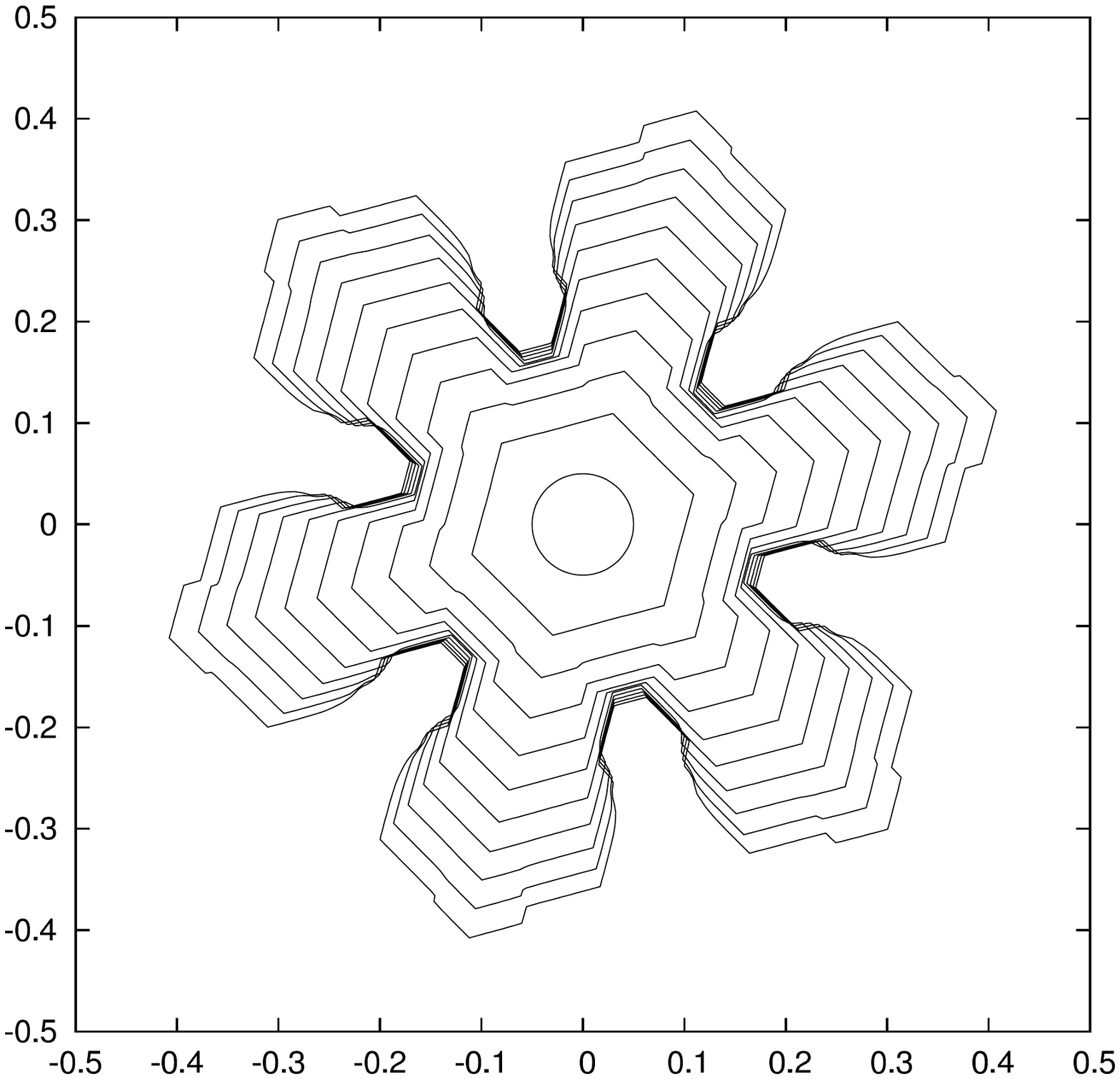}
\quad 
 \includegraphics[angle=-90,width=0.5\textwidth]{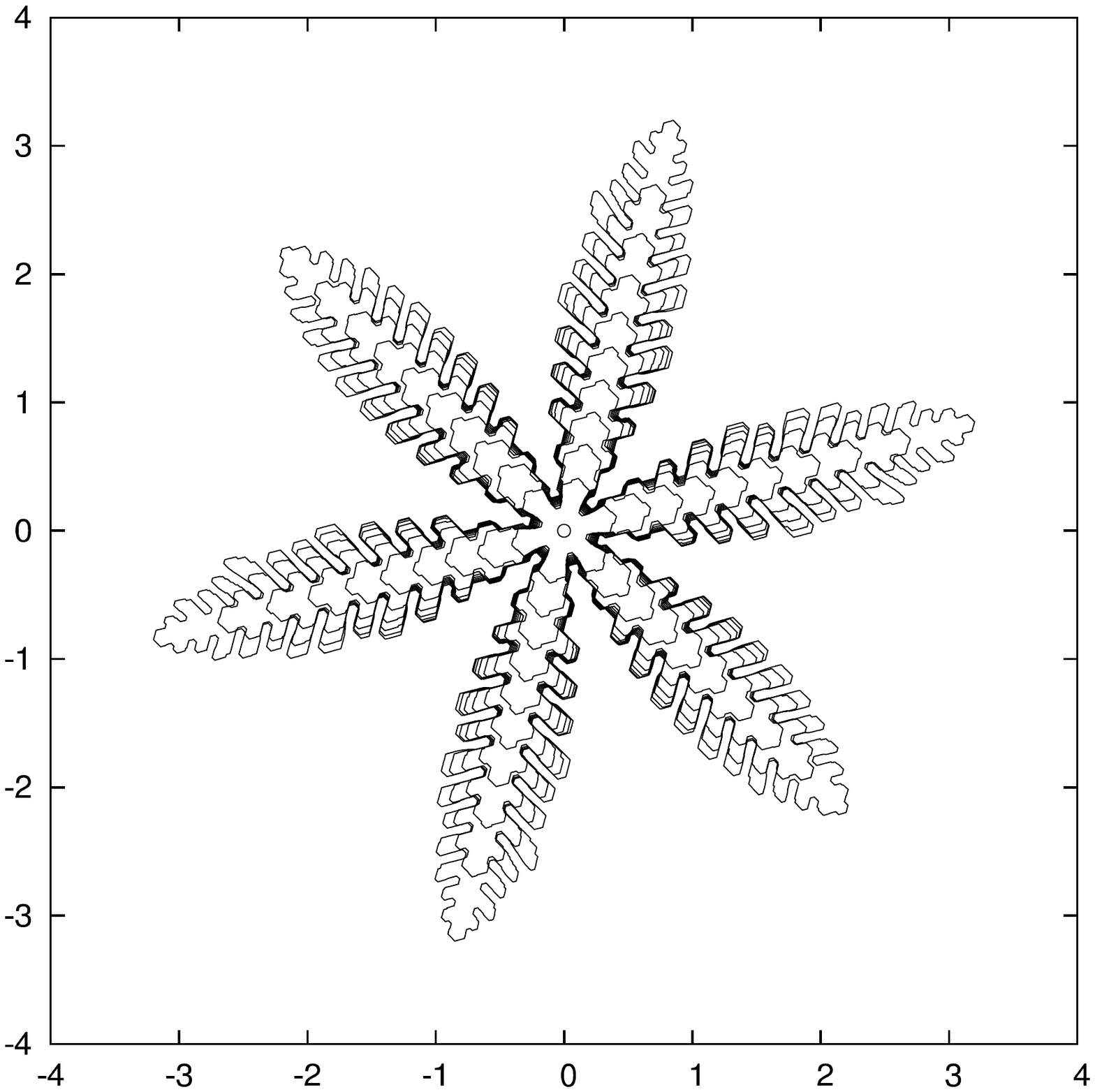}
}
\caption{($\Omega=(-4,4)^2$, $u_D = -0.08$, $\gamma=\beta$)
$\vec{X}(t)$ for $t=0,\,0.2,\ldots,2$ (left), 
and for $t=0,\,2,\ldots,20$ (right).
Parameters are $N_f=1024$, $N_c = K^0_\Gamma = 64$, and $\tau=10^{-3}$.
}
\label{fig:2dhex08}
\end{figure}
\begin{figure}
\center
\mbox{
 \includegraphics[angle=-90,width=0.5\textwidth]{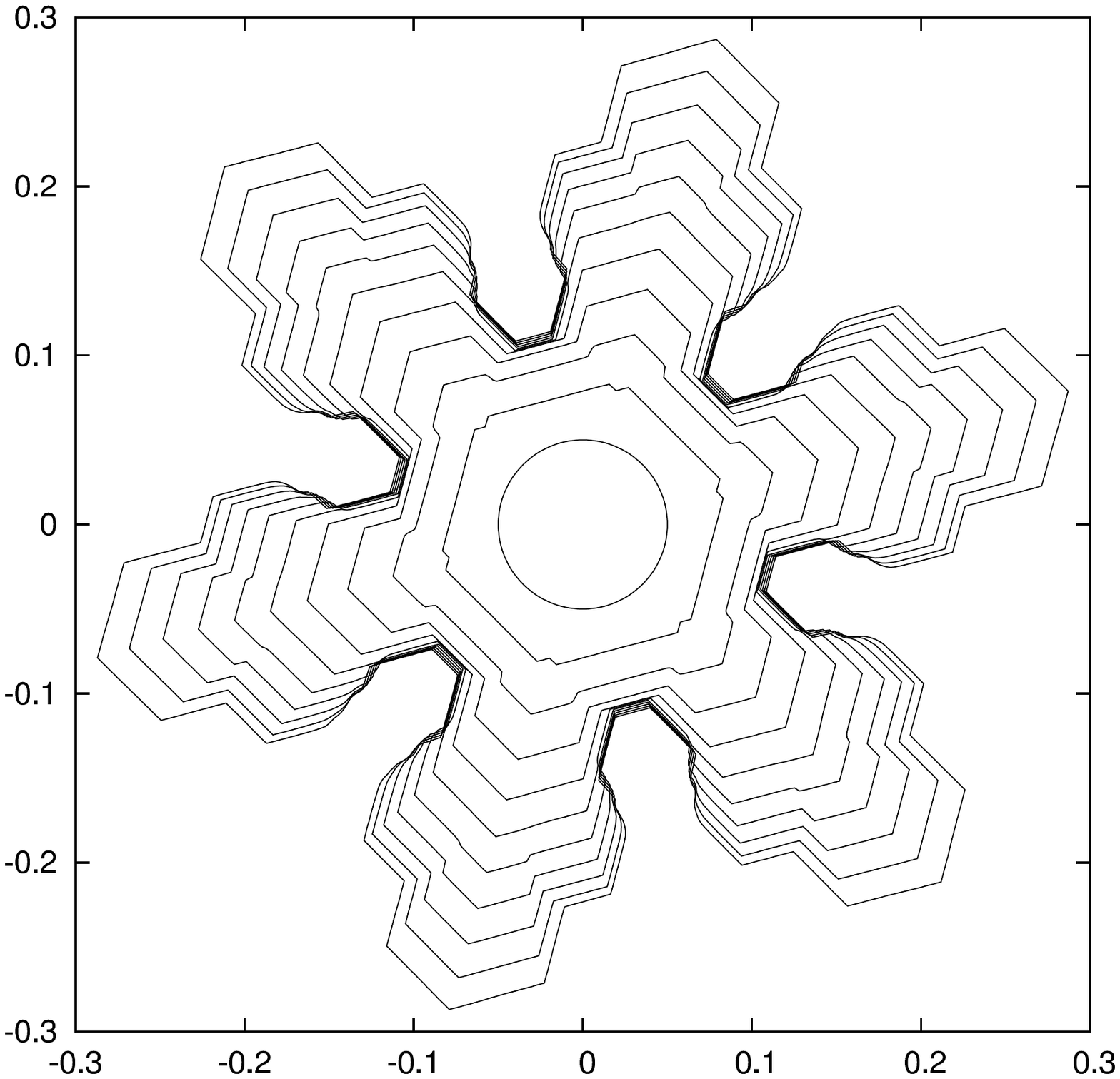}
\quad 
 \includegraphics[angle=-90,width=0.5\textwidth]{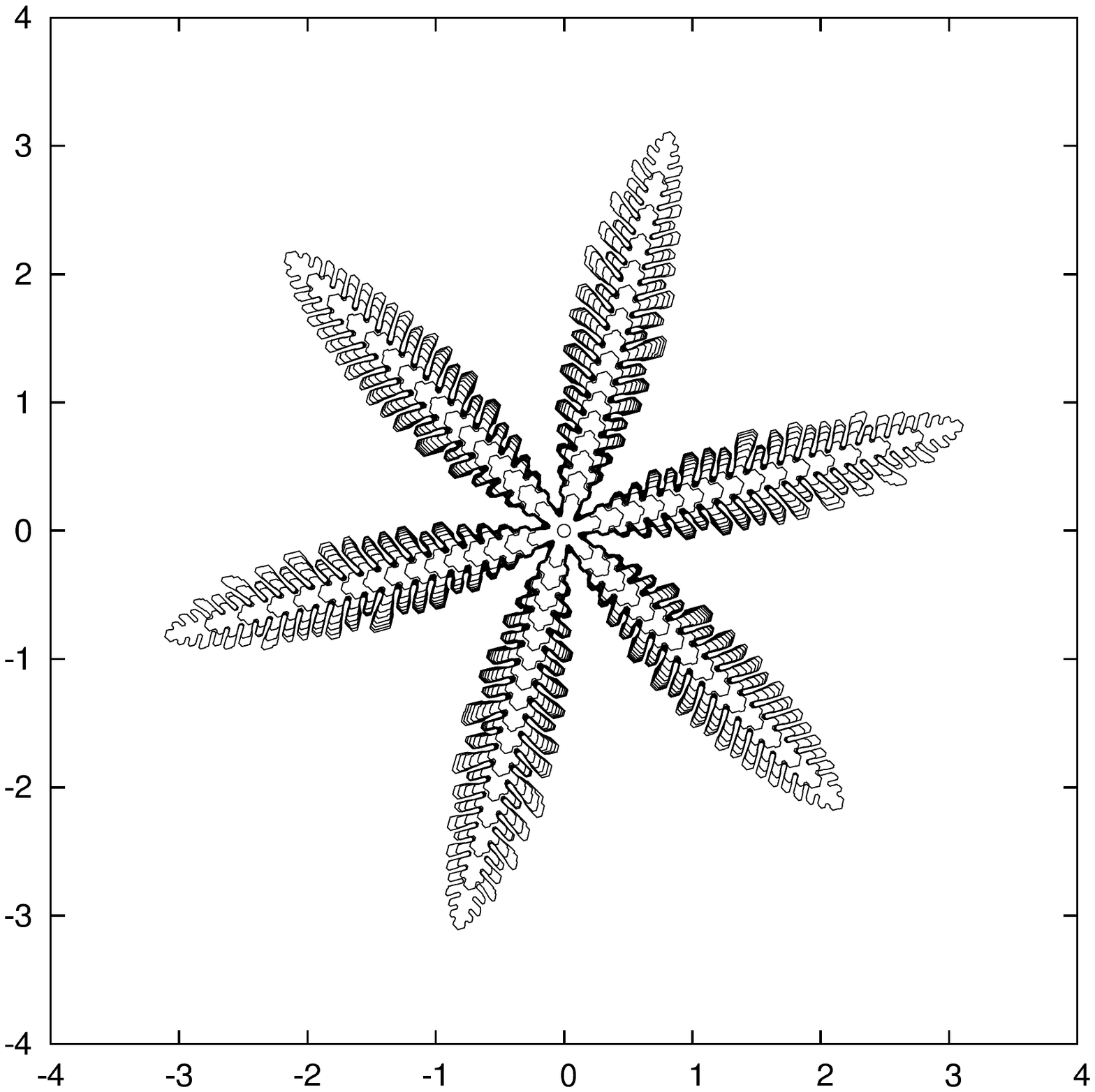}
}
\caption{($\Omega=(-4,4)^2$, $u_D = -0.2$, $\gamma=\beta$)
$\vec{X}(t)$ for $t=0,\,0.04,\ldots,0.4$ (left), 
and for $t=0,\,0.4,\ldots,6.4$ (right).
Parameters are $N_f=2048$, $N_c = K^0_\Gamma = 128$, 
and $\tau=2.5\times10^{-4}$.
}
\label{fig:2dhex2}
\end{figure}

\subsection{Crystal growth simulations for $d=3$}
Throughout this subsection, unless otherwise stated, 
we use the parameters in (\ref{eq:HGparams}) and
$\gamma=\gamma_{hex}$ defined by (\ref{eq:L44}) 
with $\epsilon=0.01$ and $\theta_0=\frac\pi{12}$.
Once again,
we use this rotation of the anisotropy $\gamma$, so that the dominant growth
directions are not exactly aligned with the $x_1$- and $x_2$-directions of the
underlying finite-element meshes
$\mathcal{T}^m$. 
Moreover, the radius of the initial crystal seed $\Gamma(0)$ is always chosen
to be $R_0=0.05$.
For later use, we define the kinetic coefficients
\begin{subequations}
\begin{equation} \label{eq:betaflat}
\beta_{\rm flat}(\vec{p}) = \beta_{\rm flat,\ell}(\vec{p}) 
:= [p_1^2 + p_2^2 + 10^{-2\ell}\,p_3^2]^\frac12
\quad\text{with}\quad  \ell \in \mathbb{N},
\end{equation}
and
\begin{equation} \label{eq:betatall}
\beta_{\rm tall}(\vec{p}) = \beta_{\rm tall,\ell}(\vec{p}) 
:= [10^{-2\ell}\,(p_1^2 + p_2^2) + p_3^2]^\frac12
\quad\text{with}\quad  \ell \in \mathbb{N}.
\end{equation}
\end{subequations}
We note that in practice, similarly to the two-dimensional results in 
Figures~\ref{fig:2dhex004} and \ref{fig:2dhex004beta}, there was hardly any 
difference between the numerical results for a kinetic coefficient $\beta$ that
is isotropic in the $x_1$-$x_2$ plane, such as $\beta_{\rm flat}$ and 
$\beta_{\rm tall}$, and one that is anisotropically aligned to the surface
energy density, such as e.g.\ $\beta = \beta_{\rm flat}\,\gamma$. Hence in all
our experiments we always choose coefficients $\beta$ that
are isotropic in the $x_1$-$x_2$ plane, e.g.\ (\ref{eq:betaflat},b). 

In the first experiment, we set $u_D = - 0.004$ and compare the results for the two
coefficients $\beta = 1$   
and $\beta = \beta_{\rm flat,3}$; see 
Figures~\ref{fig:20}   
and \ref{fig:22}. We observe that the kinetic coefficient seems to be
responsible for the fact whether solid prisms or thin plates grow; see
also the Nakaya diagram in Figure~\ref{fig:libbrecht} 
and \cite{Libbrecht05}. More details of
the evolution for the 
simulation in Figure~\ref{fig:22} 
are given in Figure~\ref{fig:23}.
\begin{figure}[ht]
\center
 \includegraphics[angle=-90,totalheight=3cm]{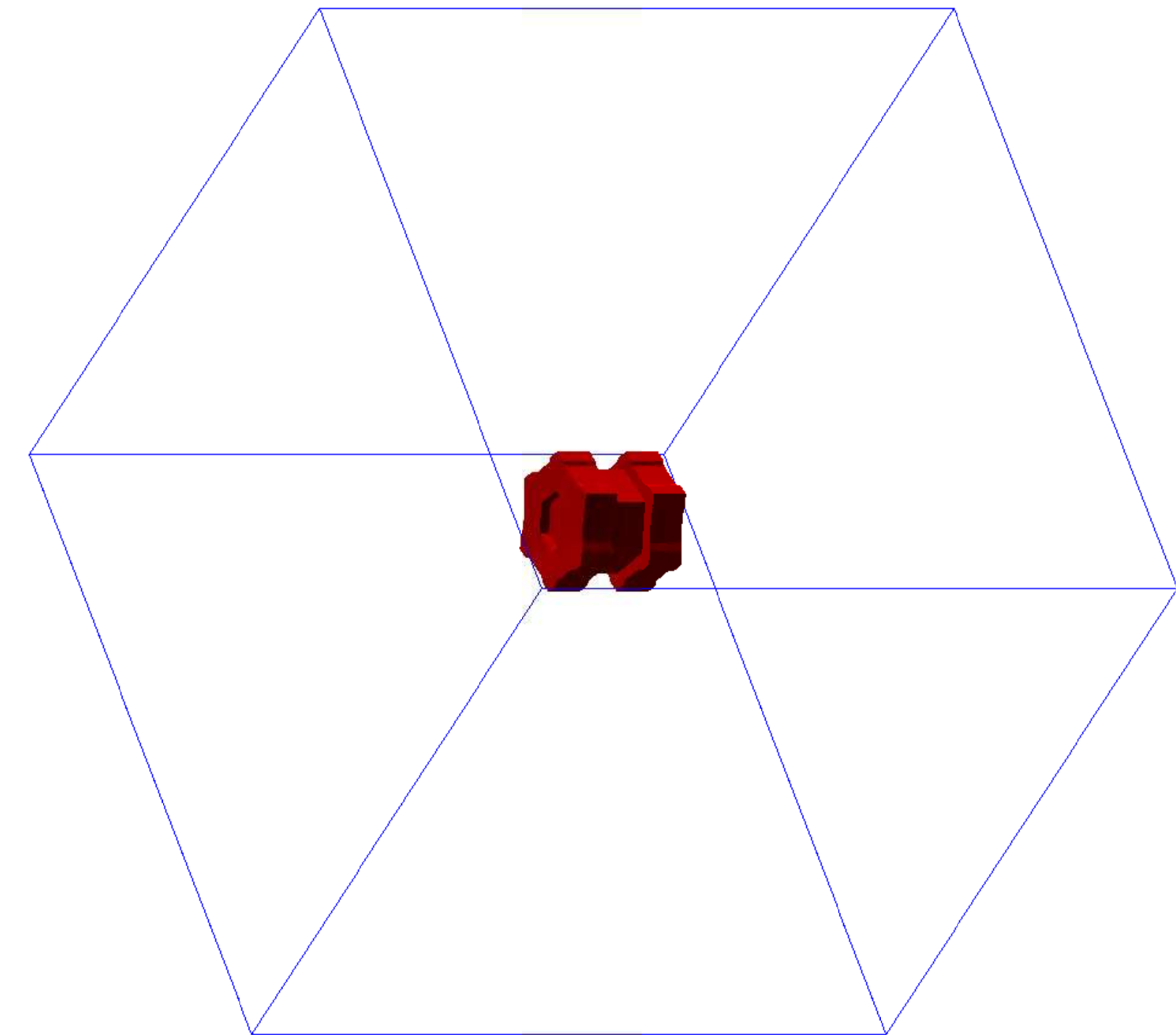}
\qquad
 \includegraphics[angle=-90,totalheight=3cm]{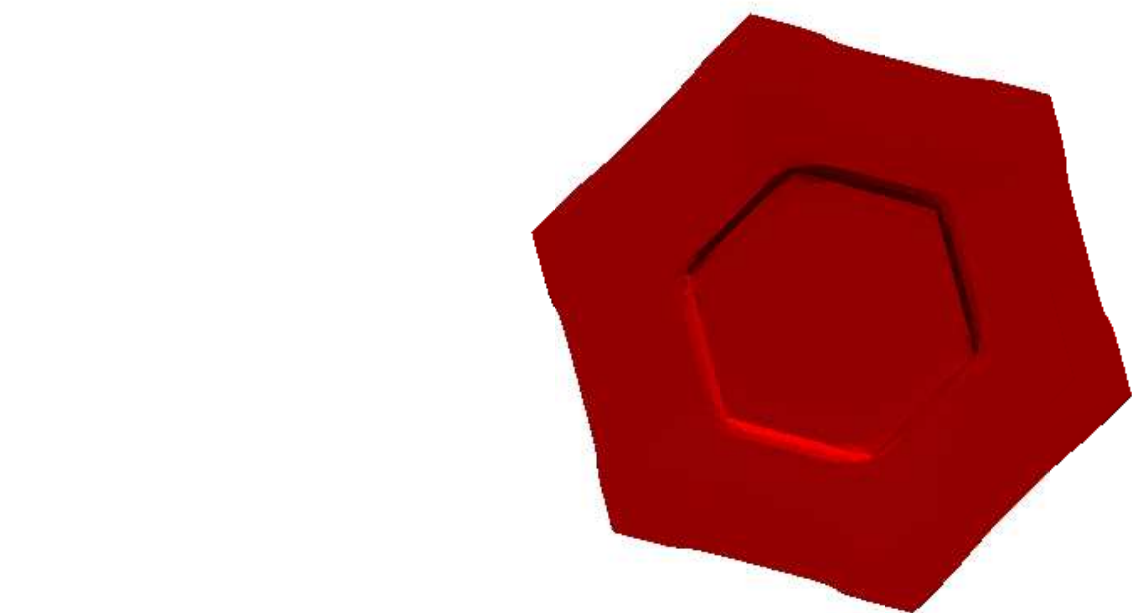}
\qquad
 \includegraphics[angle=-90,totalheight=3cm]{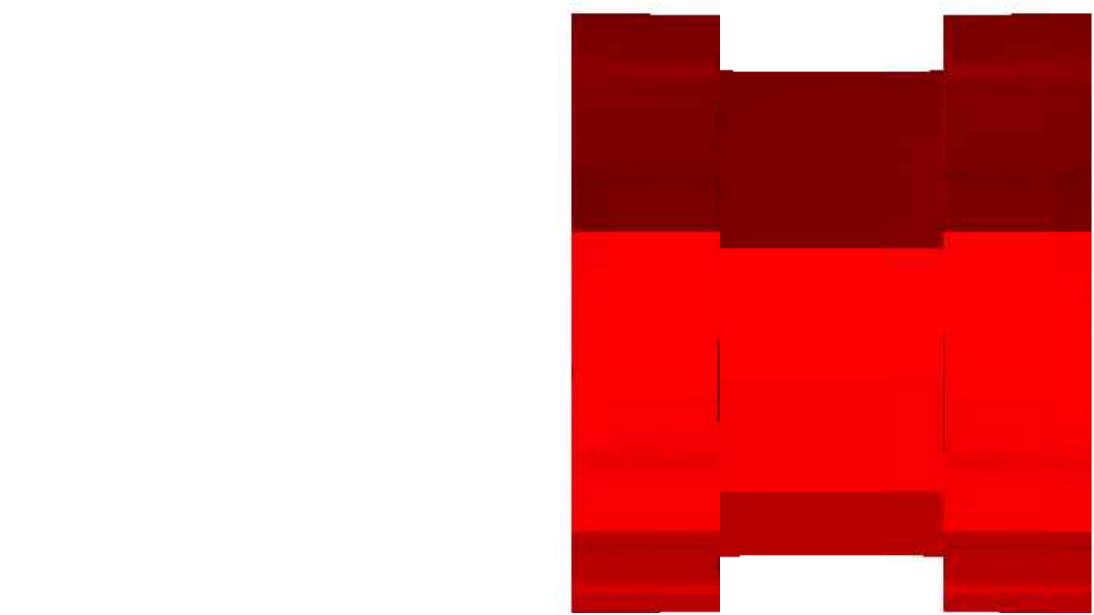}
\caption{($\Omega=(-4,4)^3$, $u_D = -0.004$, $\beta = 1$)
$\vec{X}(T)$ for $T=50$.
Parameters are $N_f=128$, $N_c=16$, $K^0_\Gamma = 98$, 
and $\tau=10^{-1}$.}
\label{fig:20}
\end{figure} 
\begin{figure}[ht]
\center
 \includegraphics[angle=-90,totalheight=2.5cm]{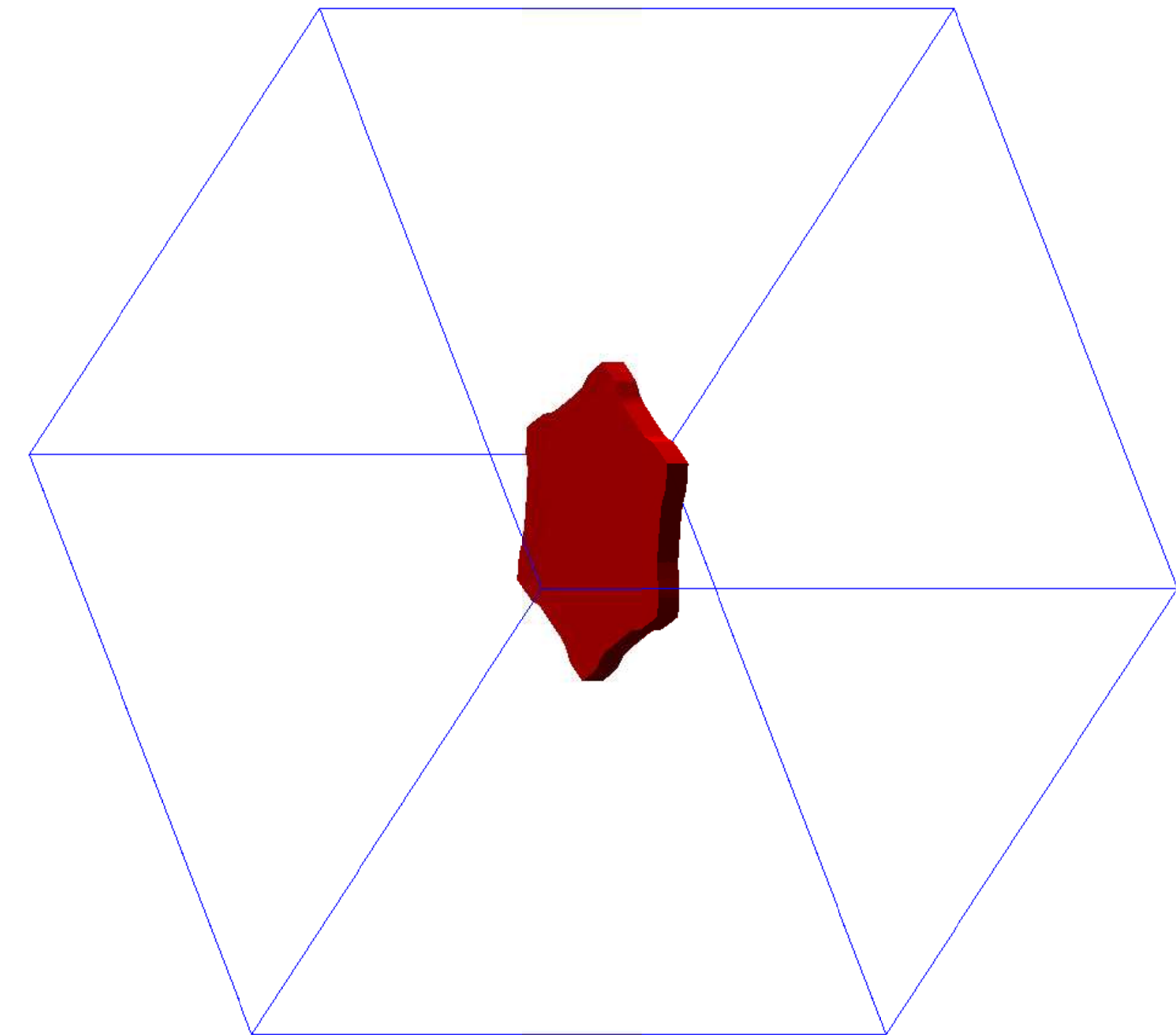}
\qquad
 \includegraphics[angle=-90,totalheight=2.5cm]{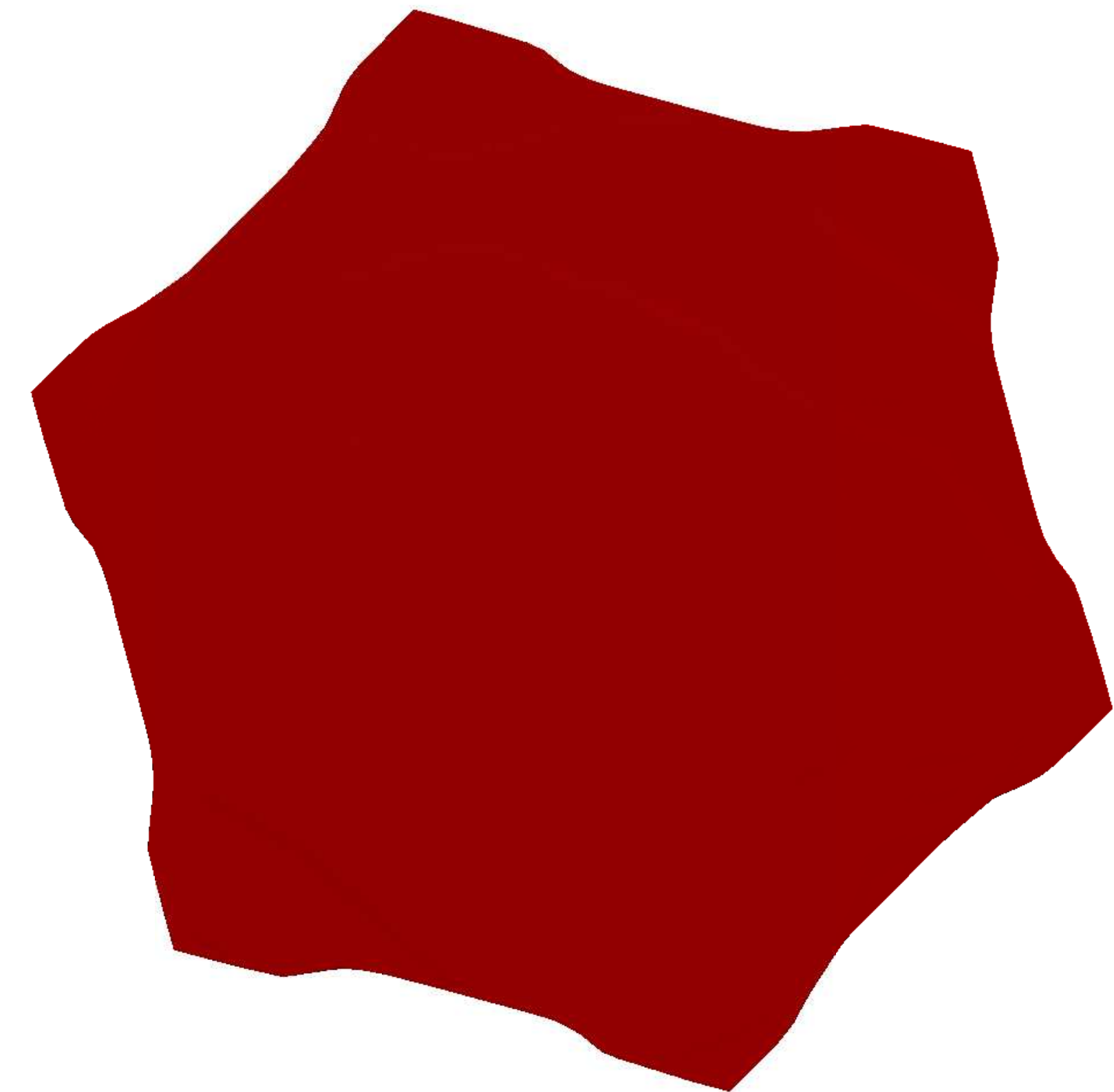}
\qquad
 \includegraphics[angle=-90,totalheight=2.5cm]{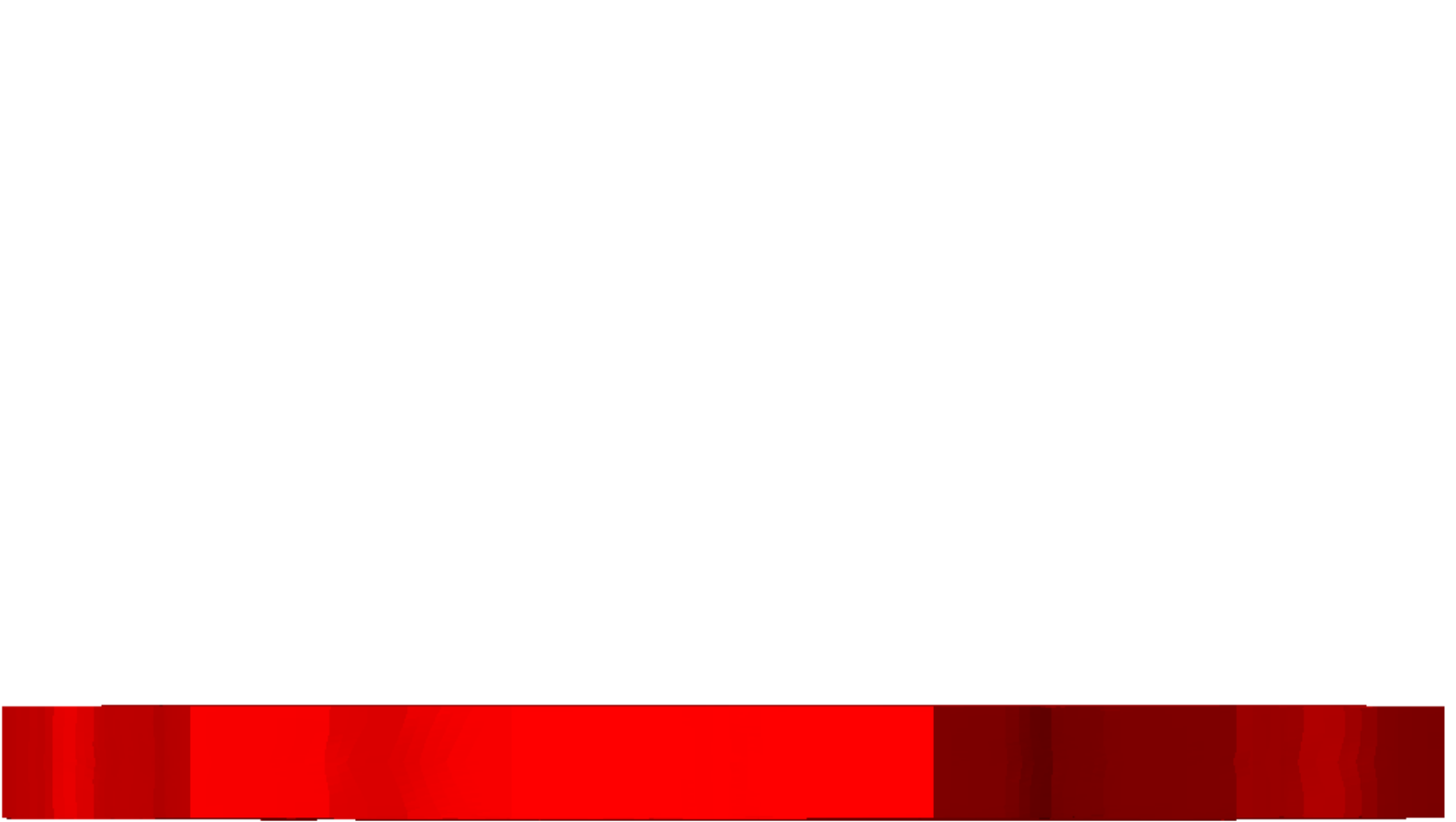}
\caption{($\Omega=(-4,4)^3$, $u_D = -0.004$, $\beta = \beta_{\rm flat,3}$)
$\vec{X}(T)$ for $T=50$.
Parameters are $N_f=128$, $N_c=16$, $K^0_\Gamma = 98$, 
and $\tau=10^{-1}$.}
\label{fig:22}
\end{figure} 
\begin{figure}[ht]
\center
 \includegraphics[angle=-90,totalheight=2.5cm]{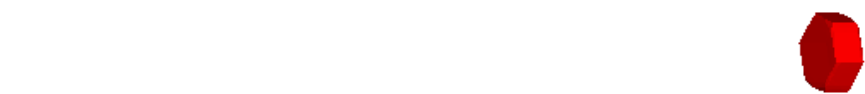}\quad
 \includegraphics[angle=-90,totalheight=2.5cm]{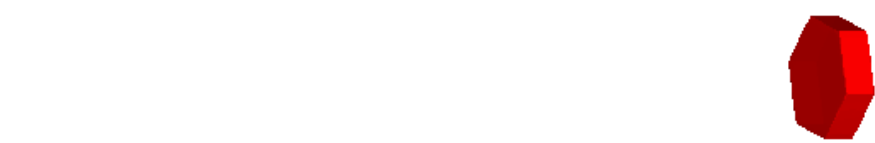}\quad
 \includegraphics[angle=-90,totalheight=2.5cm]{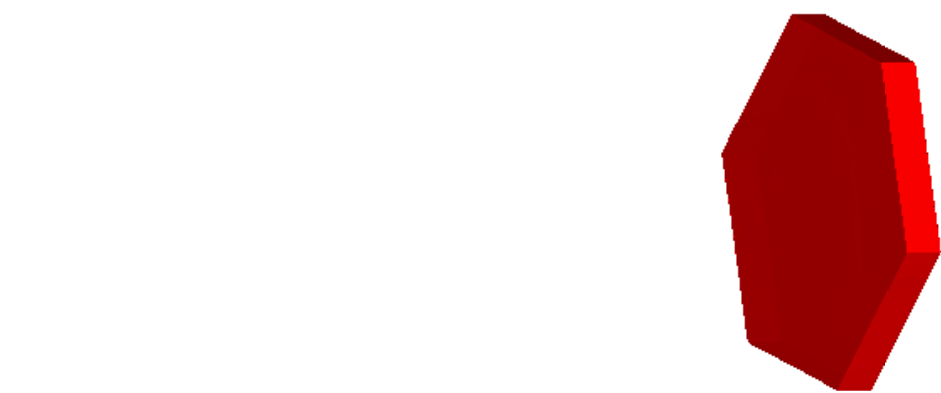}\quad
 \includegraphics[angle=-90,totalheight=2.5cm]{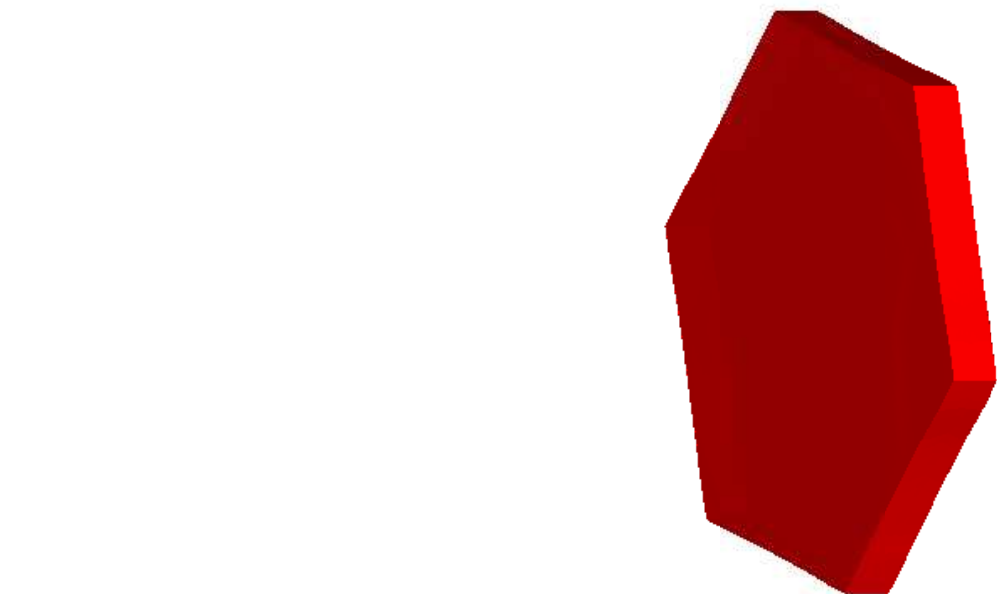}\quad
 \includegraphics[angle=-90,totalheight=2.5cm]{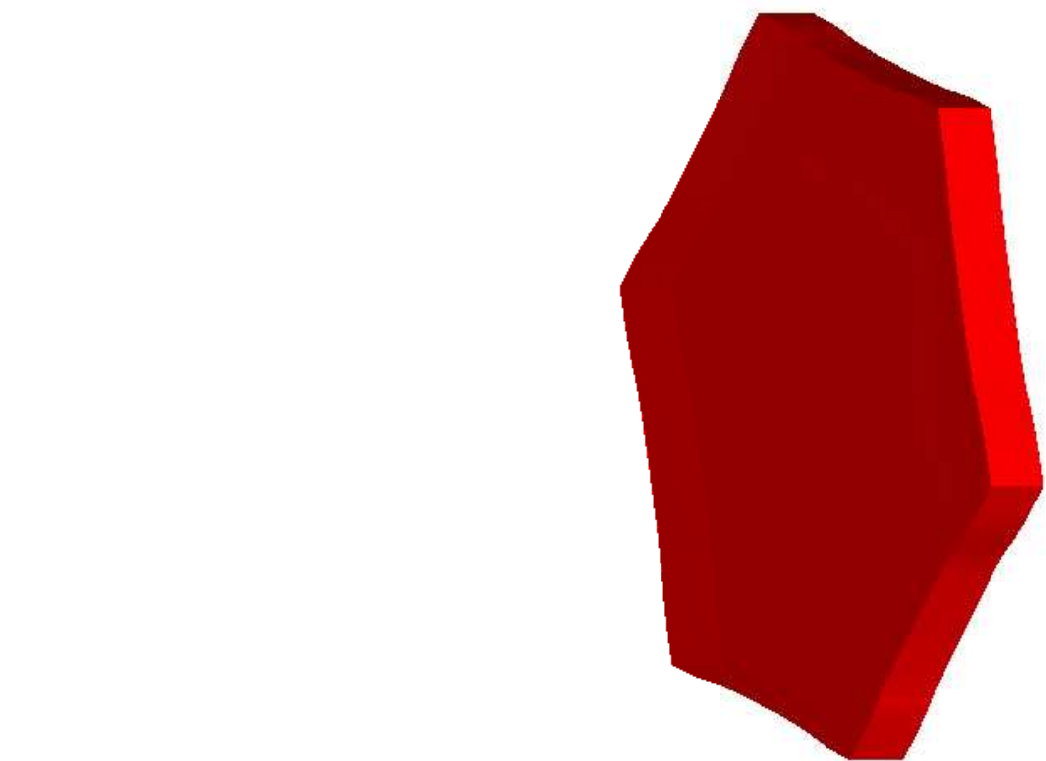}\quad
 \includegraphics[angle=-90,totalheight=2.5cm]{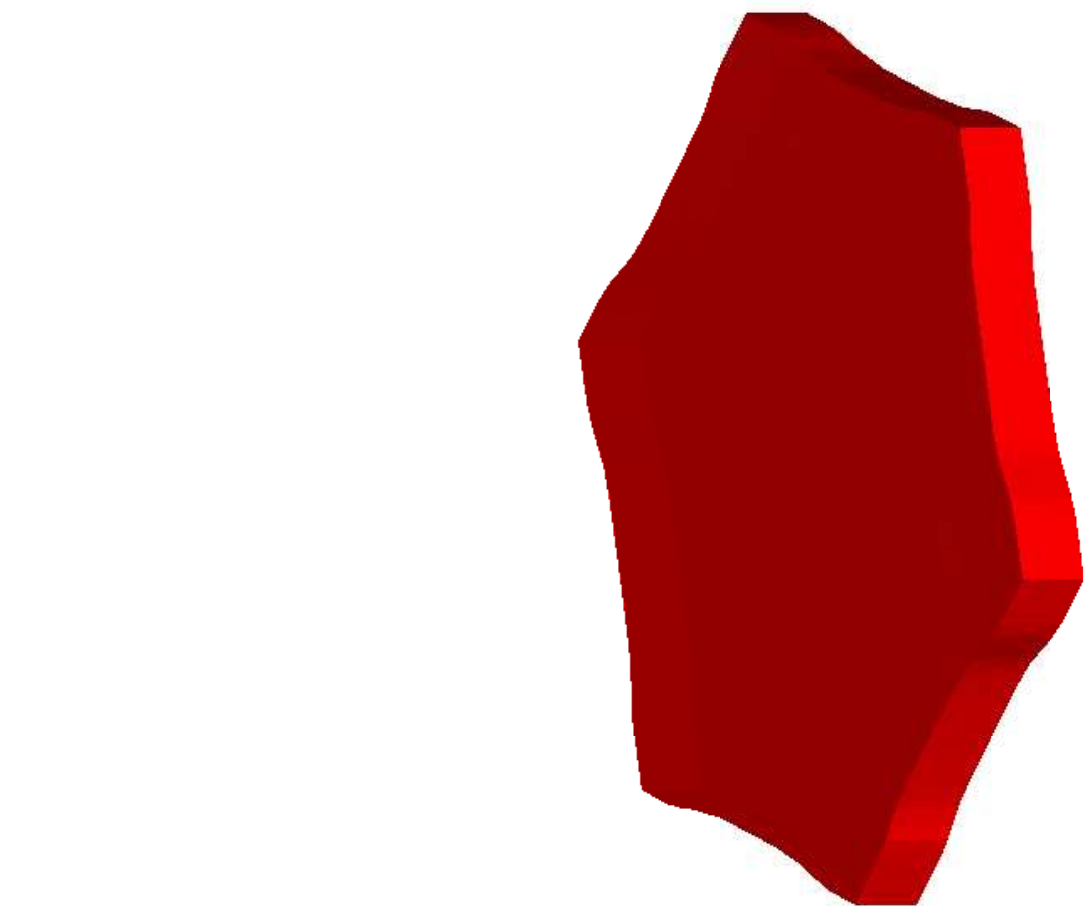}\quad
 \includegraphics[angle=-90,totalheight=2.5cm]{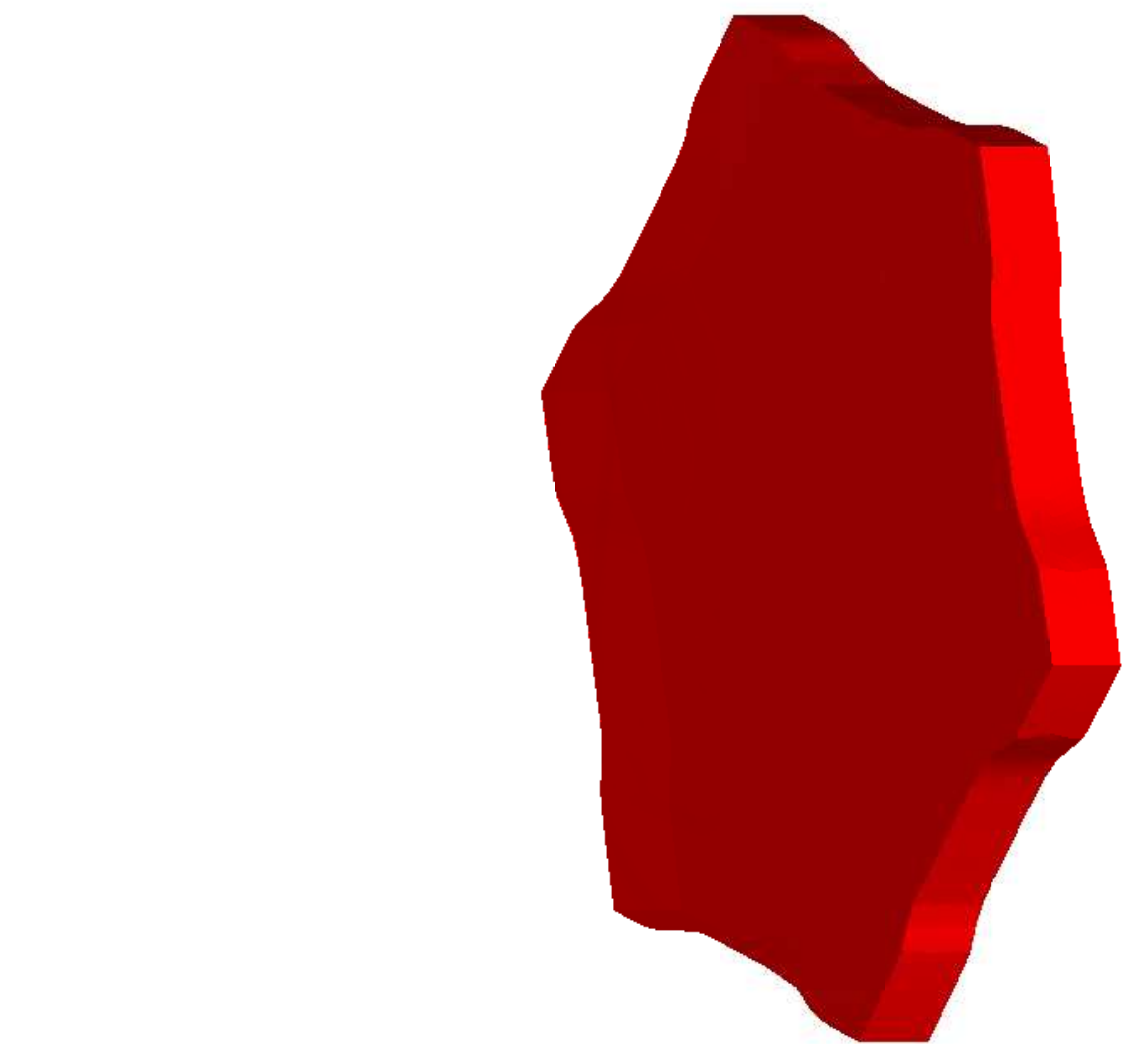}
\qquad
\caption{($\Omega=(-4,4)^3$, $u_D = -0.004$, $\beta = \beta_{\rm flat,3}$)
$\vec{X}(t)$ for $t=1,\,2,\,10,\,20,\,30,\,40,\,50$.
Parameters are $N_f=128$, $N_c=16$, $K^0_\Gamma = 98$, 
and $\tau=10^{-1}$.}
\label{fig:23}
\end{figure} 
A continuation of the evolution shown in Figure~\ref{fig:23}, now on the larger
domain $\Omega = (-8,8)^3$, can be seen in Figure~\ref{fig:23long},
where the onset of dendritic growth can be observed. 
\begin{figure}[ht]
\center
 \includegraphics[angle=-90,totalheight=3.5cm]{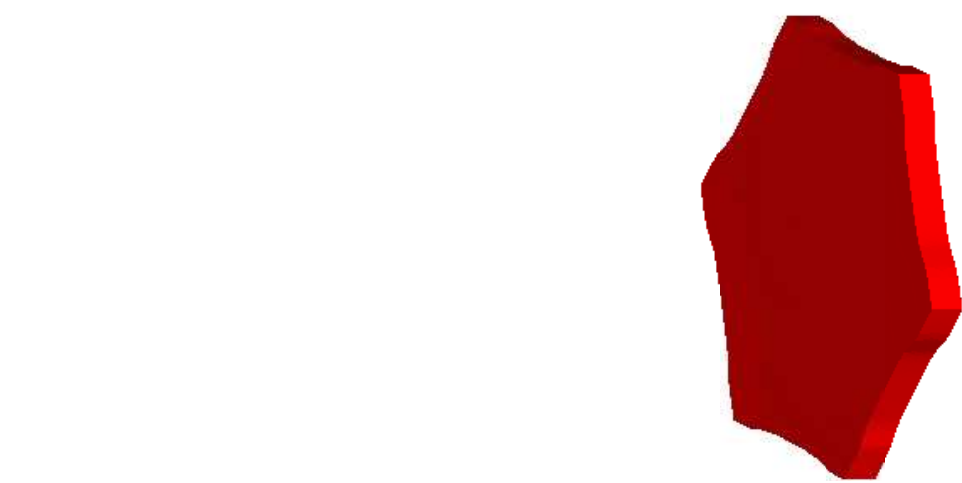}\quad
 \includegraphics[angle=-90,totalheight=3.5cm]{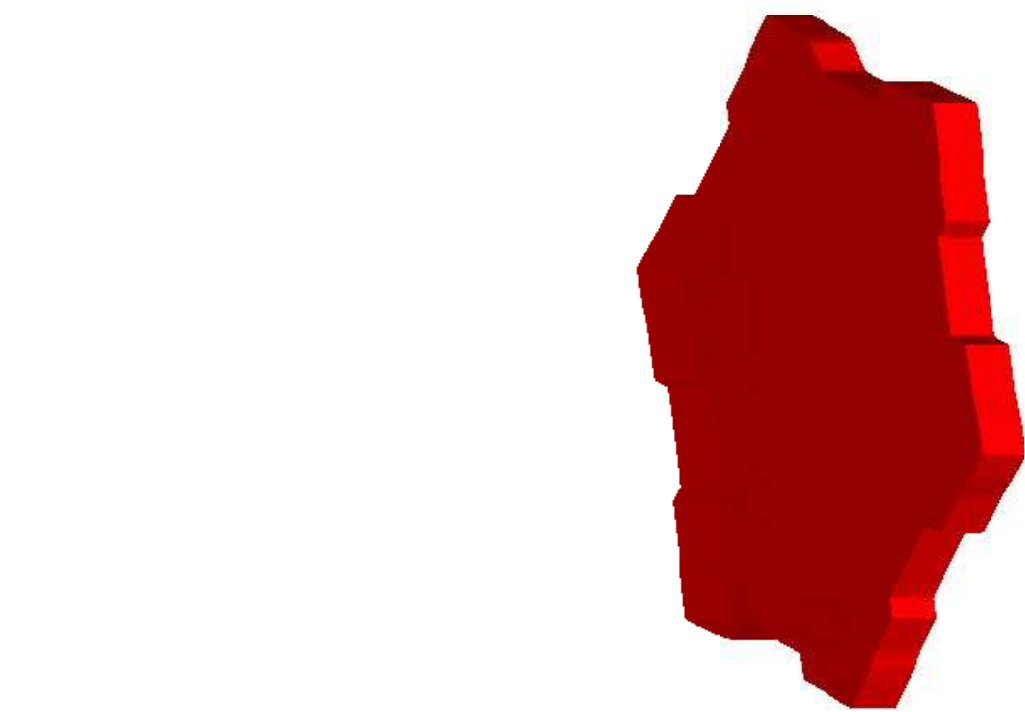}\quad
 \includegraphics[angle=-90,totalheight=3.5cm]{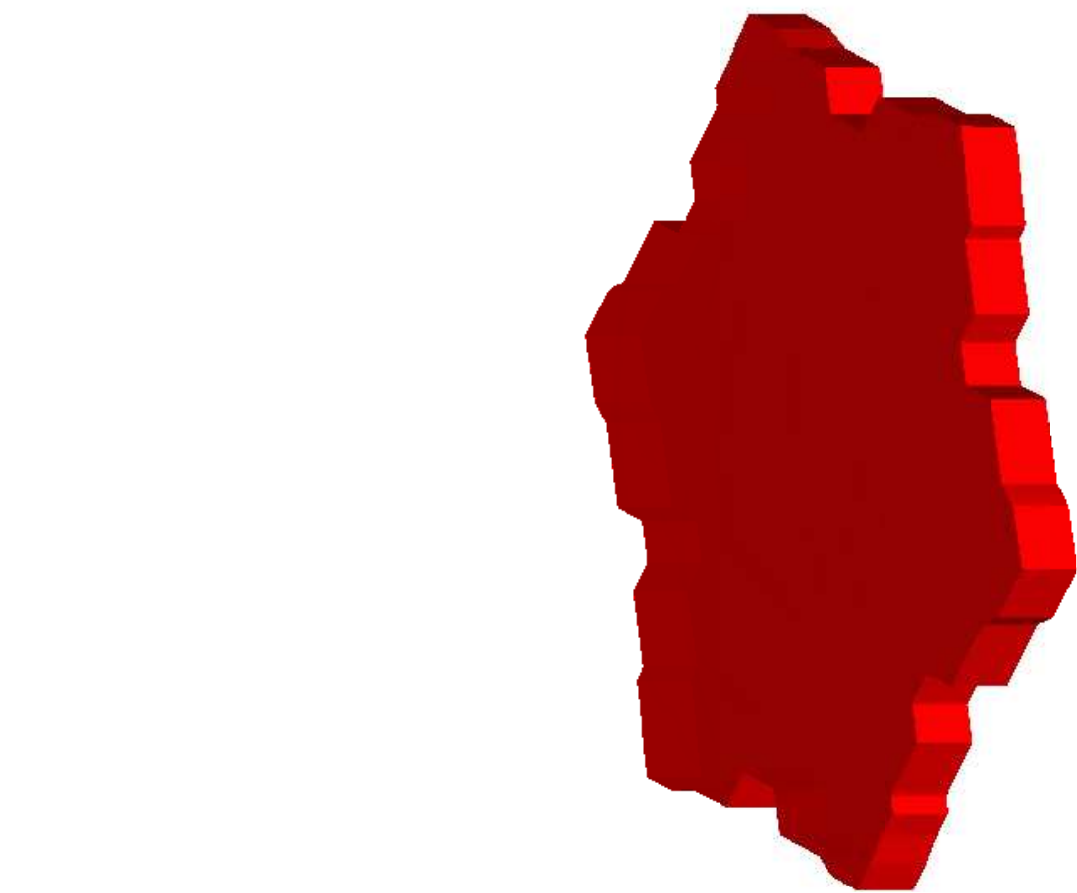}\quad
 \includegraphics[angle=-90,totalheight=3.5cm]{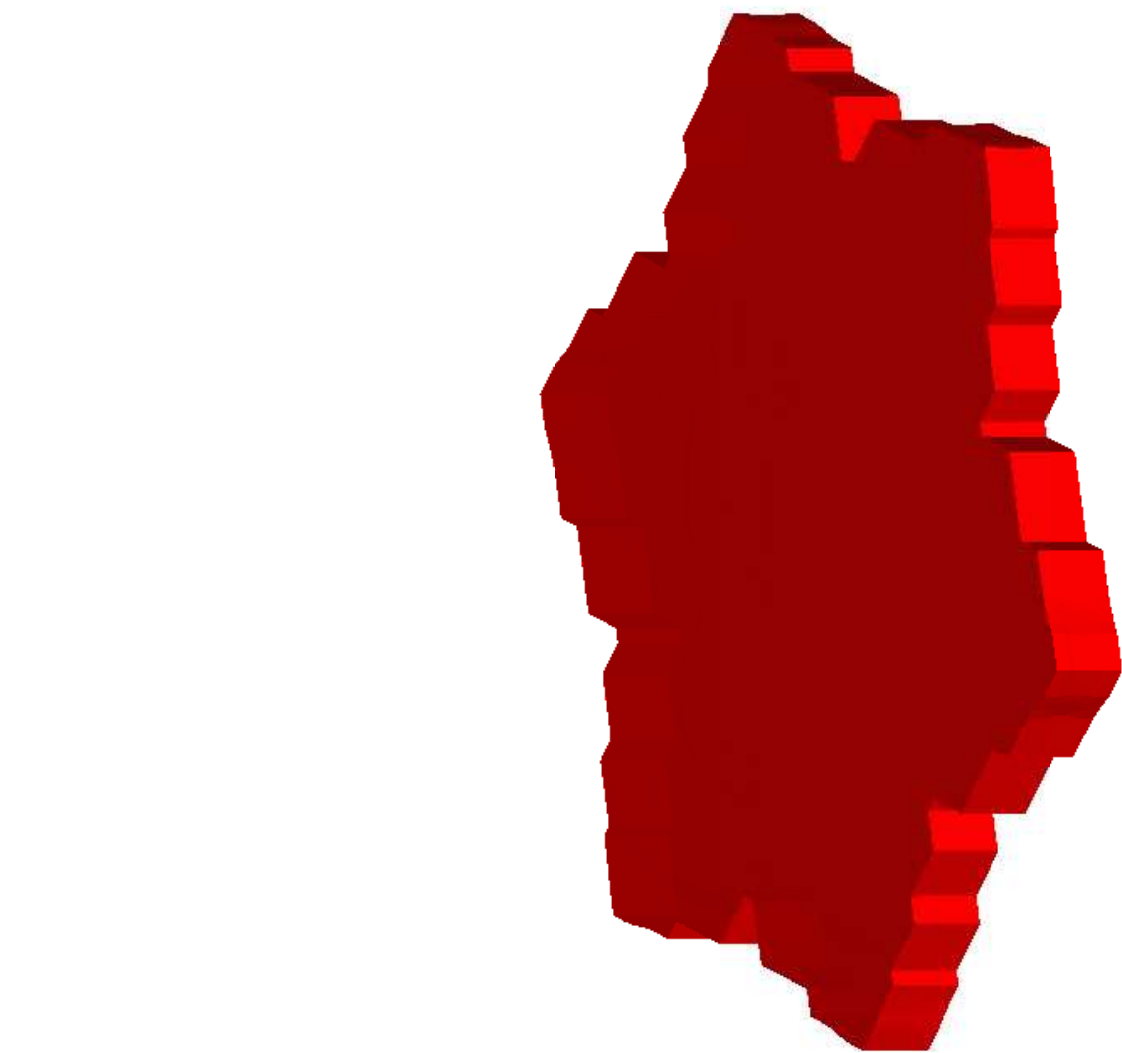}
\caption{($\Omega=(-8,8)^3$, $u_D = -0.004$, $\beta = \beta_{\rm flat,3}$)
$\vec{X}(t)$ for $t=50,\,100,\,150,\,200$.
Parameters are $N_f=256$, $N_c=32$, $K^0_\Gamma = 98$, 
and $\tau=10^{-1}$.}
\label{fig:23long}
\end{figure} 

An experiment for $u_D = -0.002$ and $\beta = \beta_{\rm tall,1}$ can be seen
in Figure~\ref{fig:24_002}, where a solid prism grows.
\begin{figure}[ht]
\center
 \includegraphics[angle=-90,totalheight=3cm]{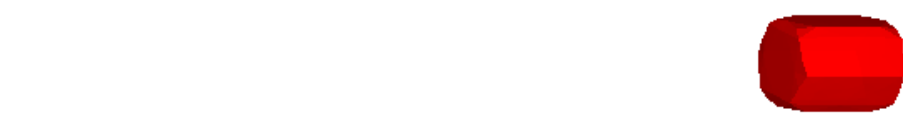} \quad
 \includegraphics[angle=-90,totalheight=3cm]{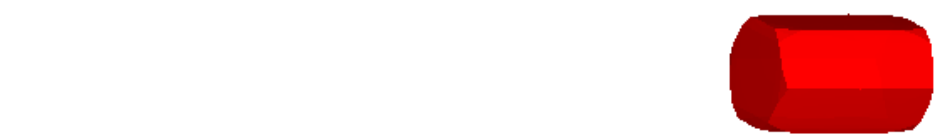} \quad
 \includegraphics[angle=-90,totalheight=3cm]{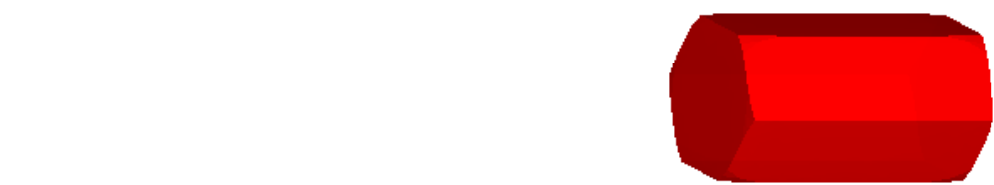} \quad
 \includegraphics[angle=-90,totalheight=3cm]{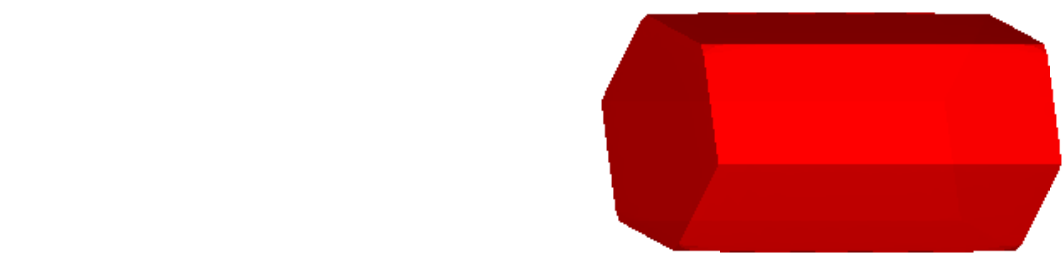} \quad
 \includegraphics[angle=-90,totalheight=3cm]{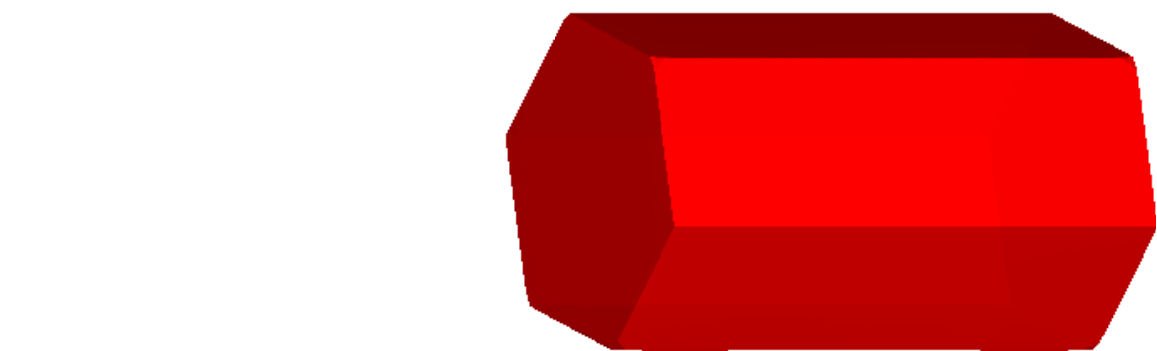} \quad
 \includegraphics[angle=-90,totalheight=3cm]{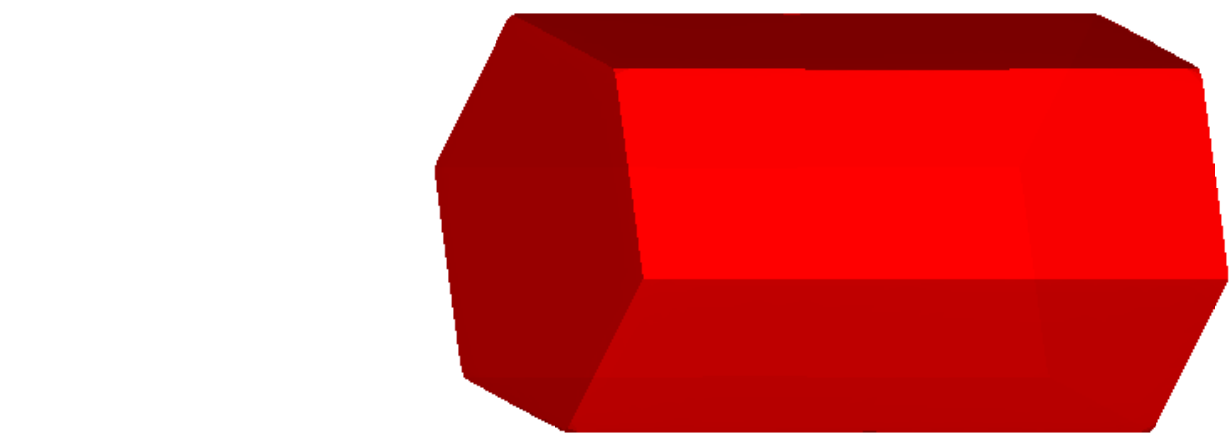} \quad
 \includegraphics[angle=-90,totalheight=3cm]{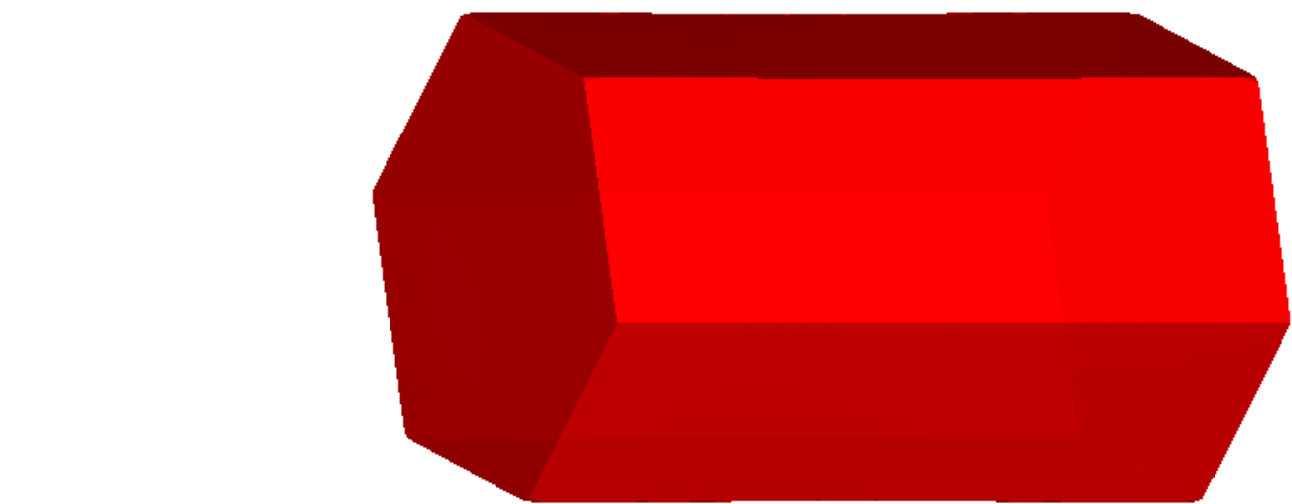} \quad
 \includegraphics[angle=-90,totalheight=3cm]{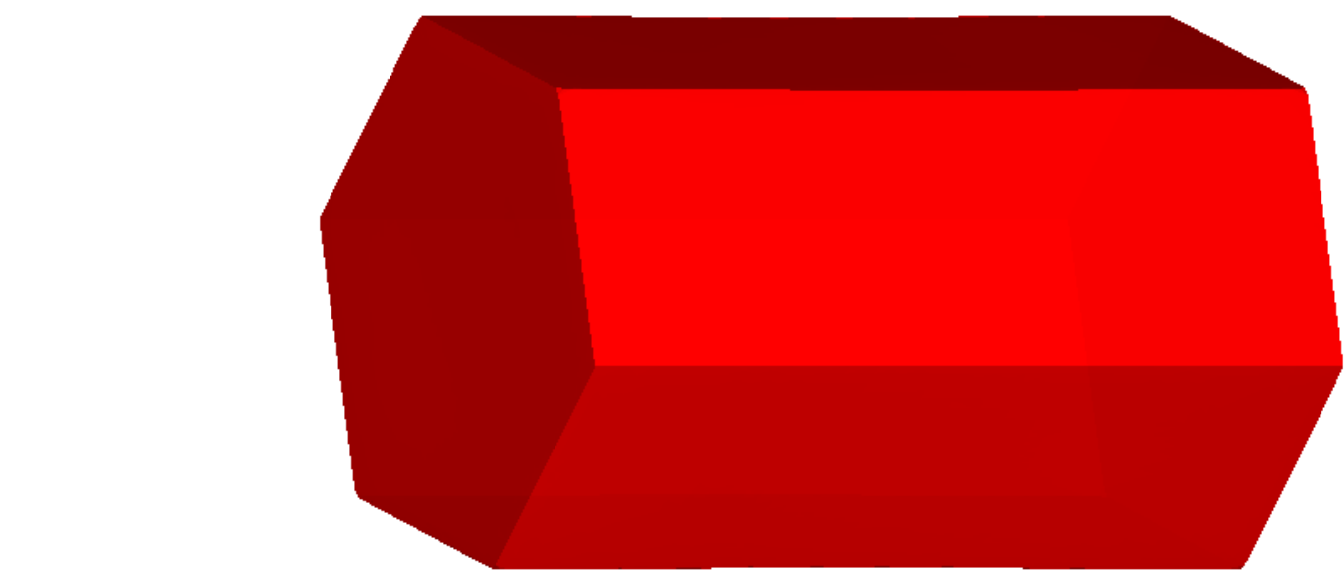} \quad
 \includegraphics[angle=-90,totalheight=3cm]{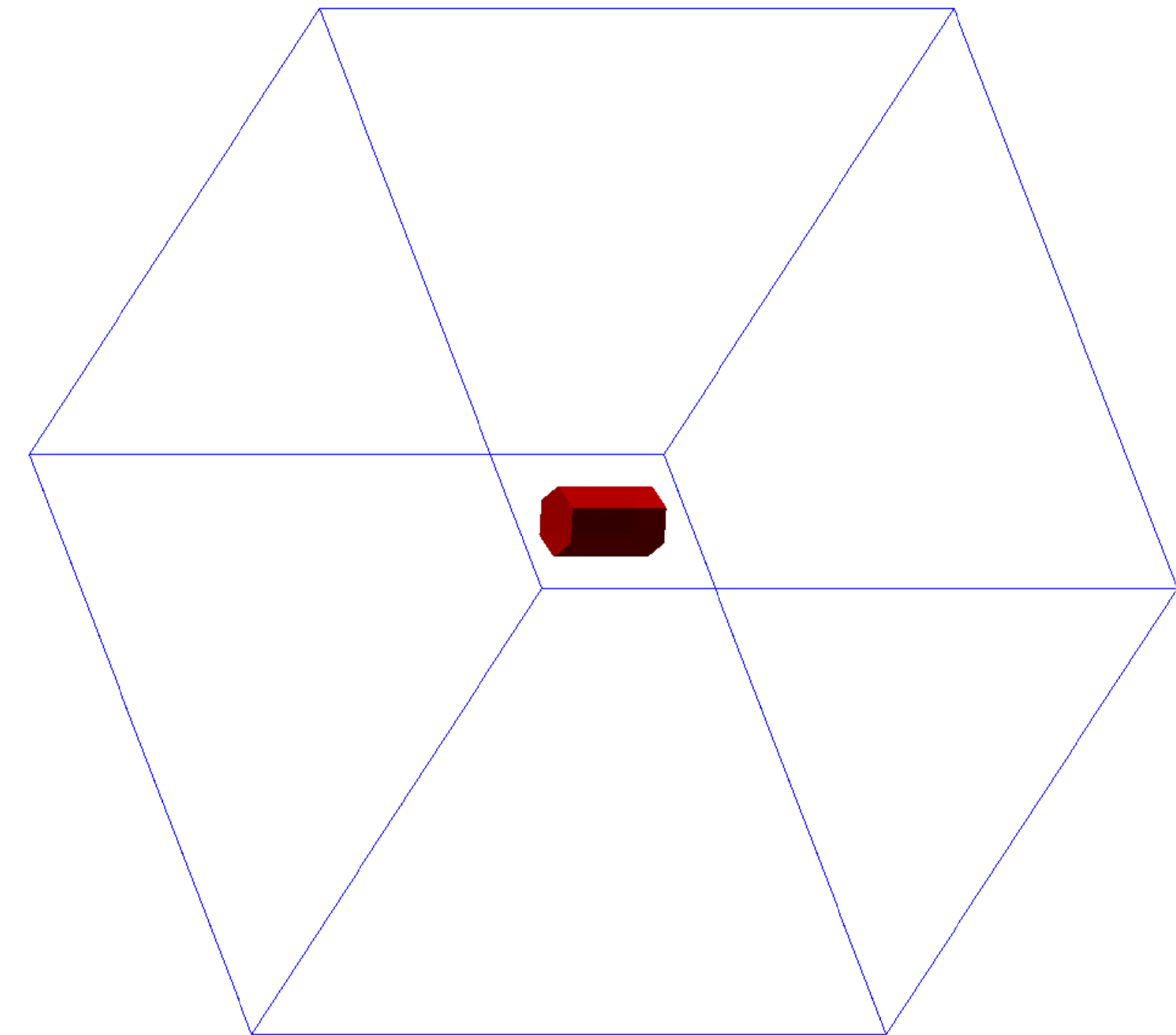} 
\caption{($\Omega=(-4,4)^3$, $u_D = -0.002$, $\beta = \beta_{\rm tall,1}$)
$\vec{X}(t)$ for $t=1,\,2,\,5,\,10,\,20,\,30,\,40,\,50$; 
and $\vec X(50)$ within $\Omega$.
Parameters are $N_f=128$, $N_c=16$, $K^0_\Gamma = 98$, and $\tau=10^{-1}$.}
\label{fig:24_002}
\end{figure} 

An experiment for $u_D = -0.008$ and $\beta = \beta_{\rm tall,2}$ can be seen
in Figure~\ref{fig:25}. In this case the basal facets break,
leading to hollow columns; see Figure~\ref{fig:libbrecht} and 
\cite{GigaR06}. 
\begin{figure}[ht]
\center
 \includegraphics[angle=-90,totalheight=5cm]{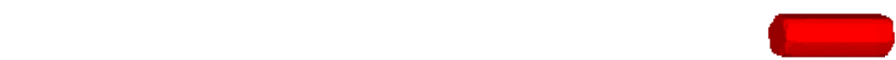} \quad
 \includegraphics[angle=-90,totalheight=5cm]{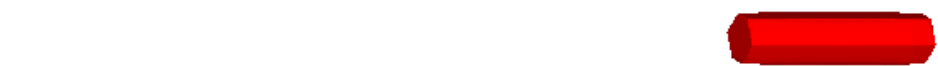} \quad
 \includegraphics[angle=-90,totalheight=5cm]{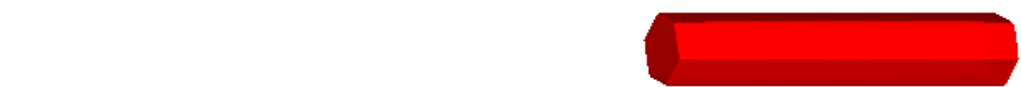} \quad
 \includegraphics[angle=-90,totalheight=5cm]{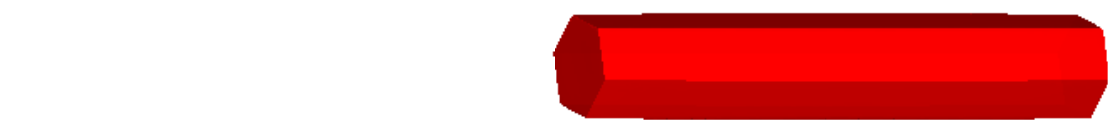} \quad
 \includegraphics[angle=-90,totalheight=5cm]{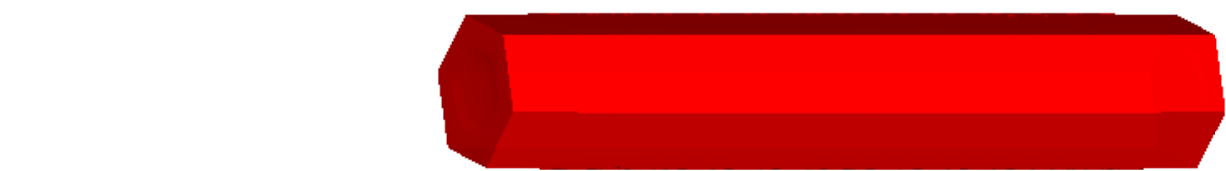} \quad
 \includegraphics[angle=-90,totalheight=5cm]{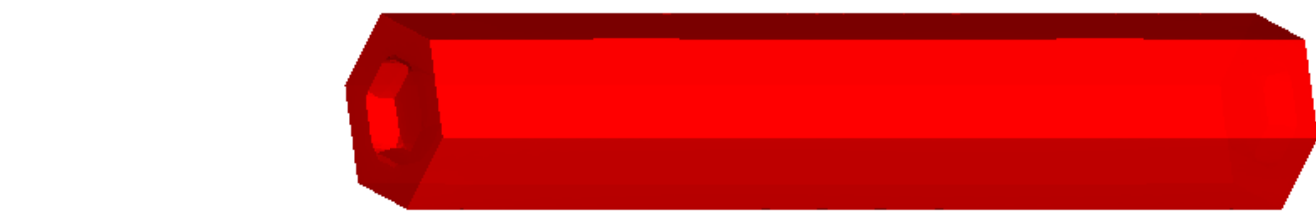} \quad
 \includegraphics[angle=-90,totalheight=5cm]{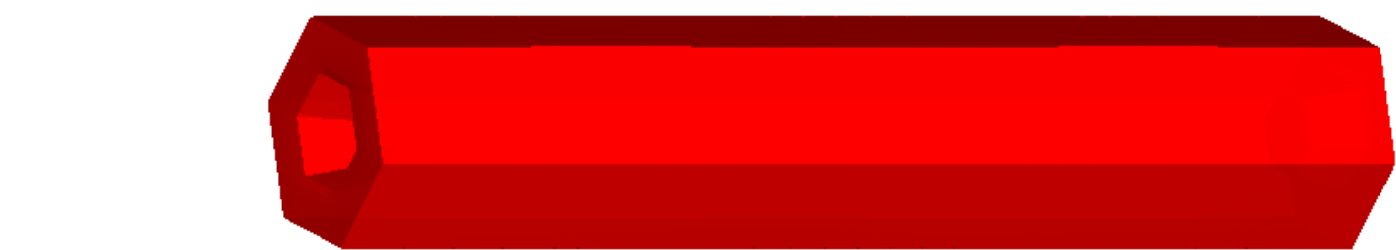} \quad
 \includegraphics[angle=-90,totalheight=5cm]{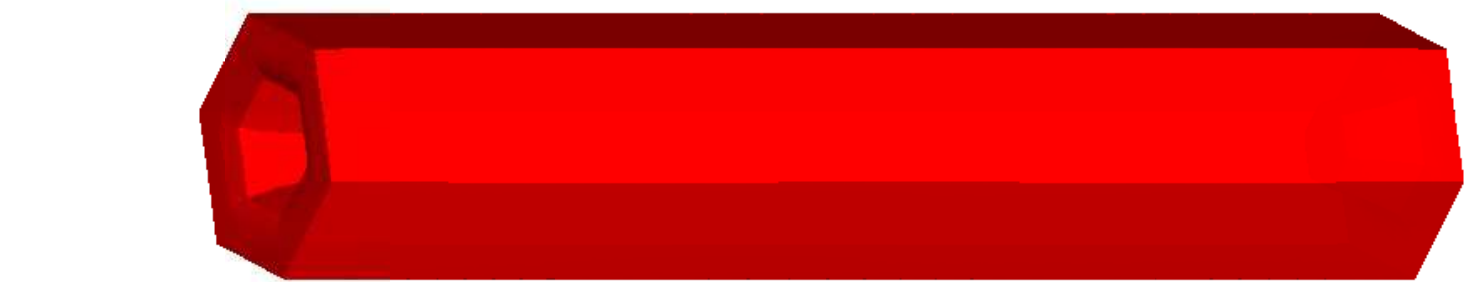} \quad
 \includegraphics[angle=-90,totalheight=5cm]{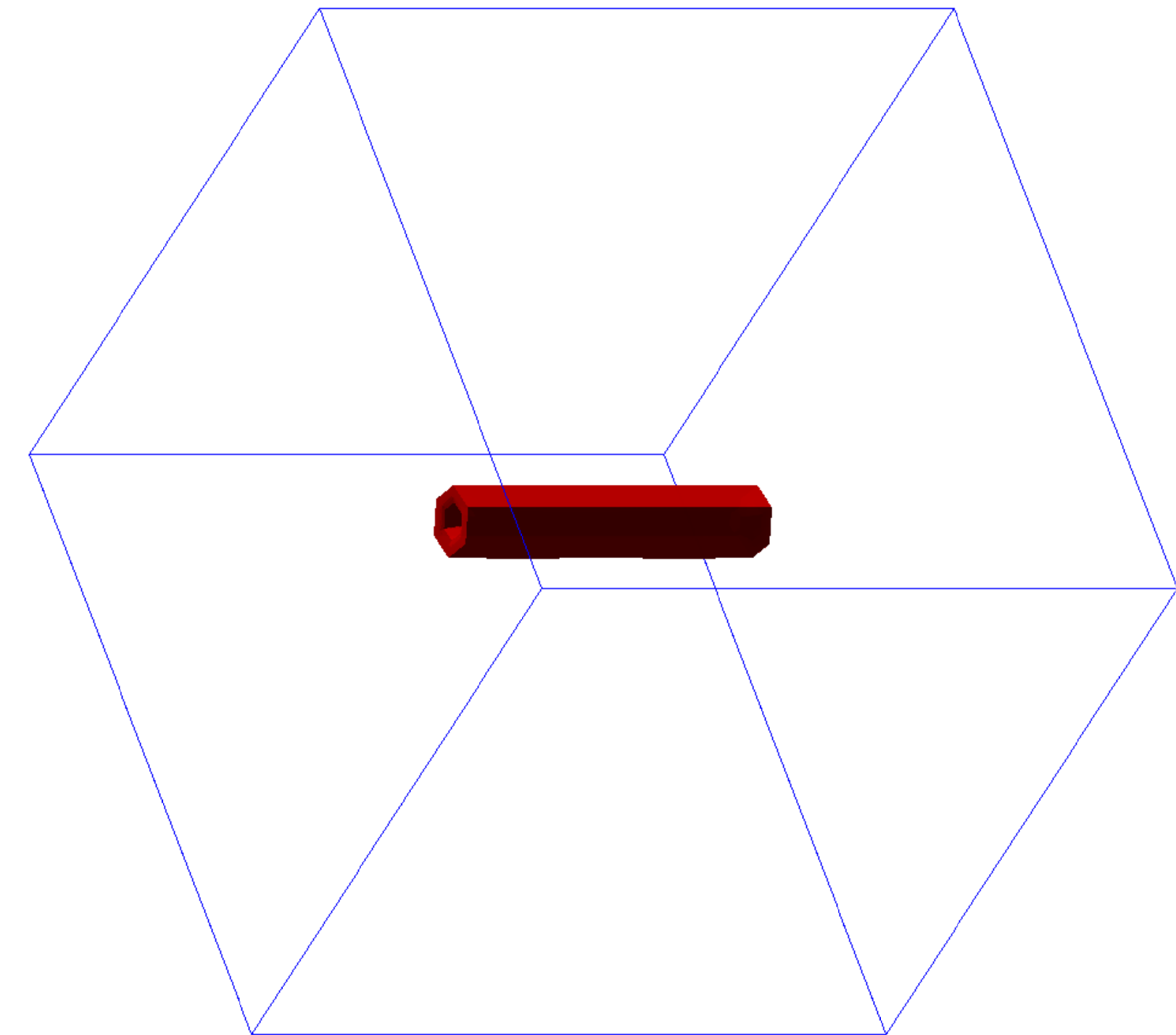}
\caption{($\Omega=(-4,4)^3$, $u_D = -0.008$, $\beta = \beta_{\rm tall,2}$)
$\vec{X}(t)$ for $t=1,\,2,\,5,\,10,\,20,\,30,\,40,\,50$; 
and $\vec X(50)$ within $\Omega$.
Parameters are $N_f=128$, $N_c=16$, $K^0_\Gamma = 98$, and $\tau=10^{-1}$.}
\label{fig:25}
\end{figure} 

An experiment for $u_D = -0.02$ and $\beta = \beta_{\rm flat,3}$ can be seen
in Figure~\ref{fig:c02f6}. In this case the prism facets break,
leading to capped columns which also can be observed in nature; see
\cite{Libbrecht05}. 
\begin{figure}[ht]
\center
 \includegraphics[angle=-90,totalheight=3cm]{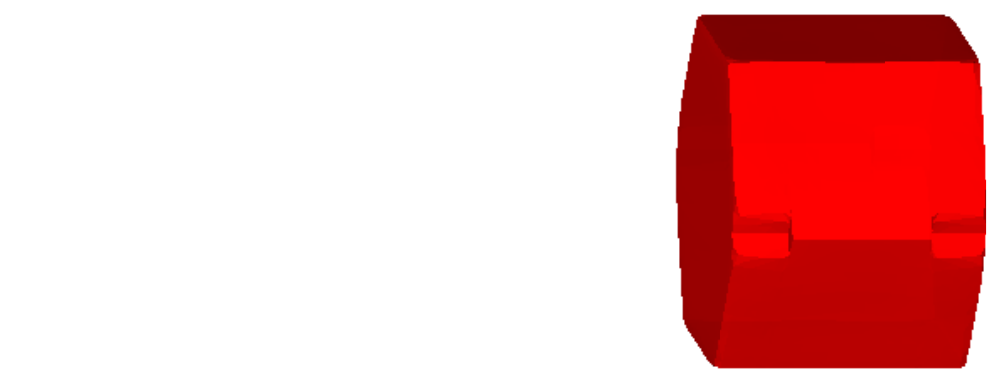} \quad
 \includegraphics[angle=-90,totalheight=3cm]{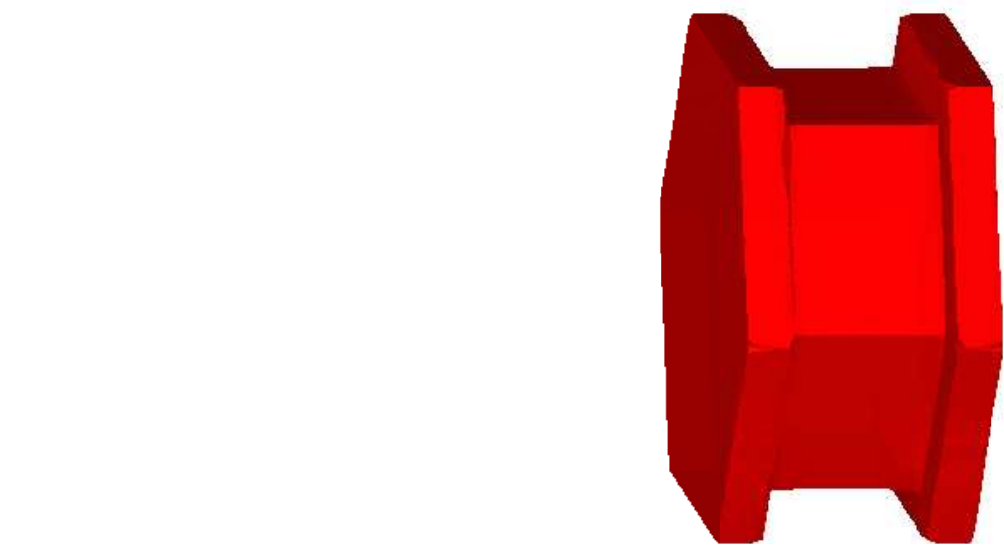} \quad
 \includegraphics[angle=-90,totalheight=3cm]{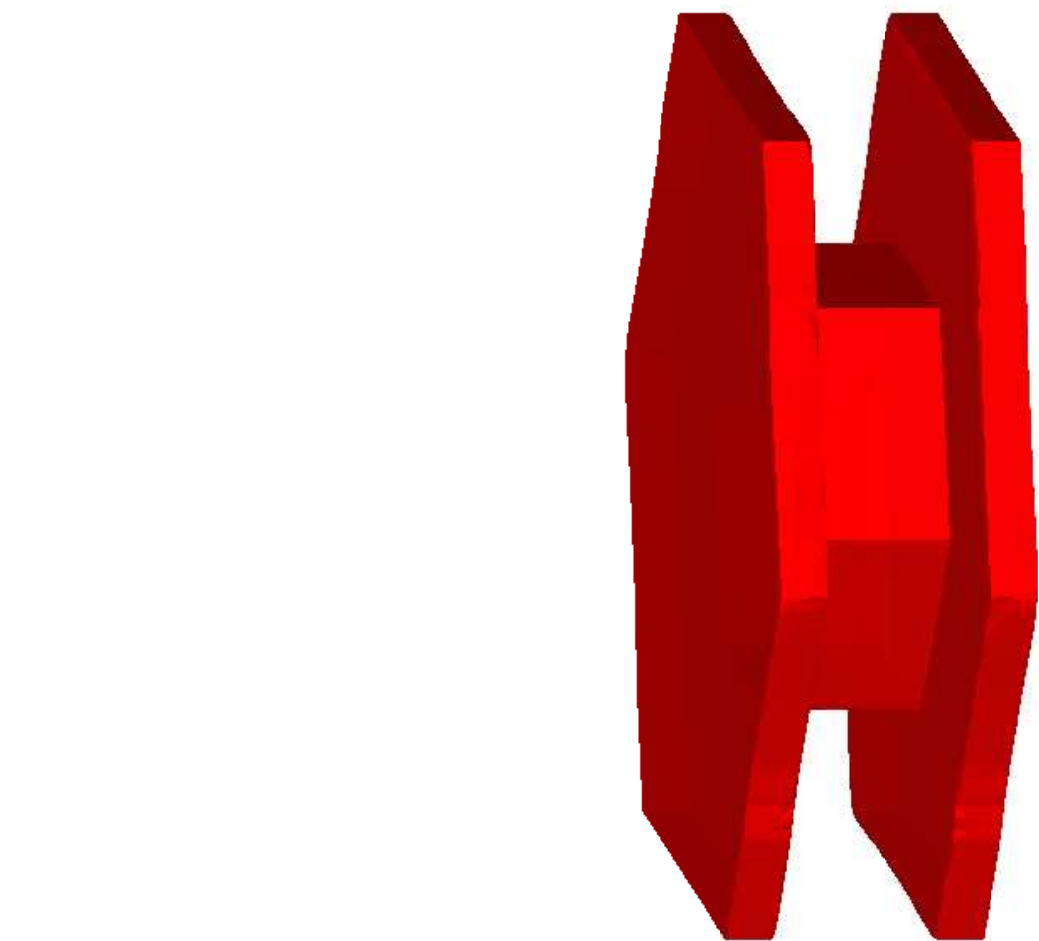} \quad
 \includegraphics[angle=-90,totalheight=3cm]{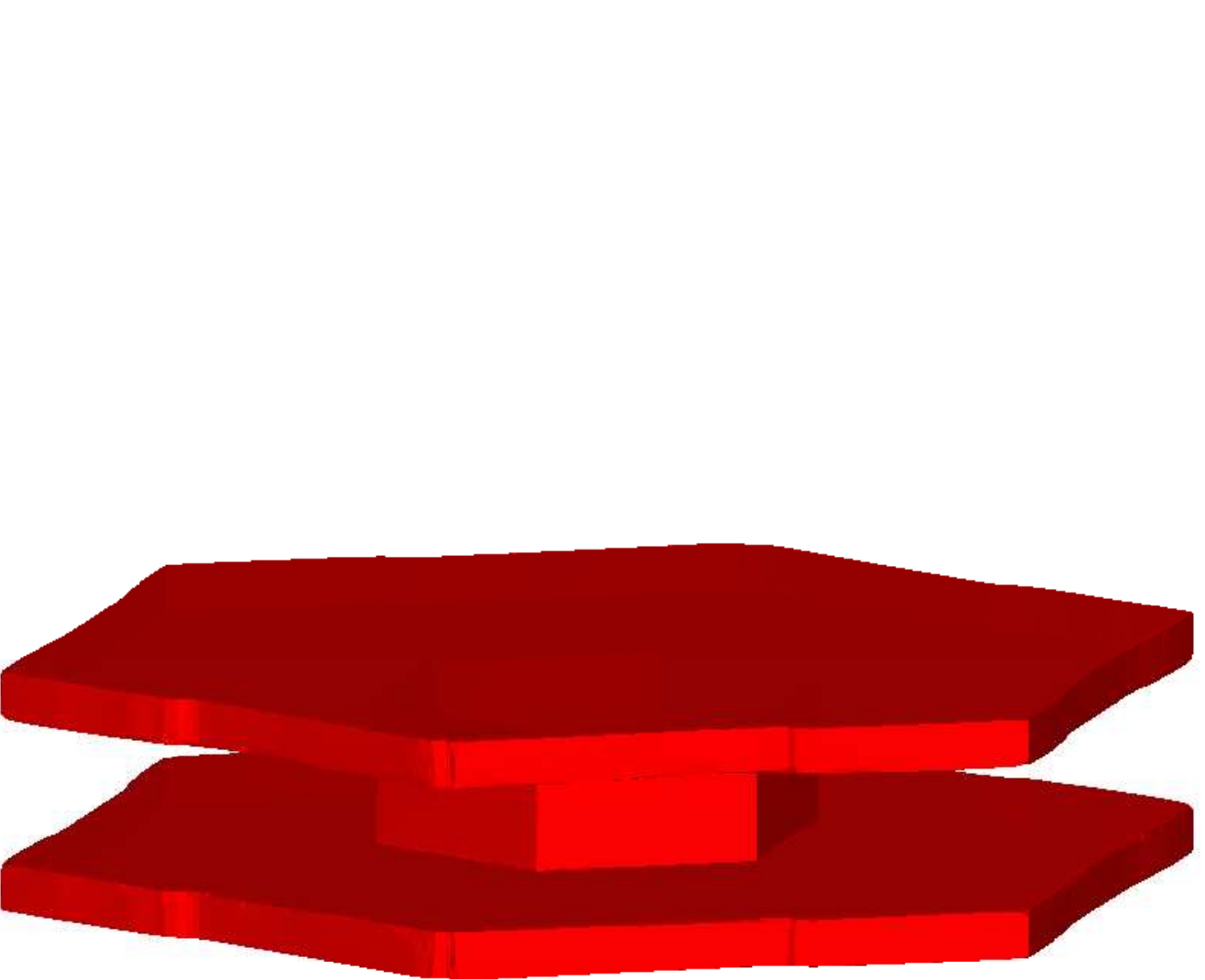} \quad
 \includegraphics[angle=-90,totalheight=3cm]{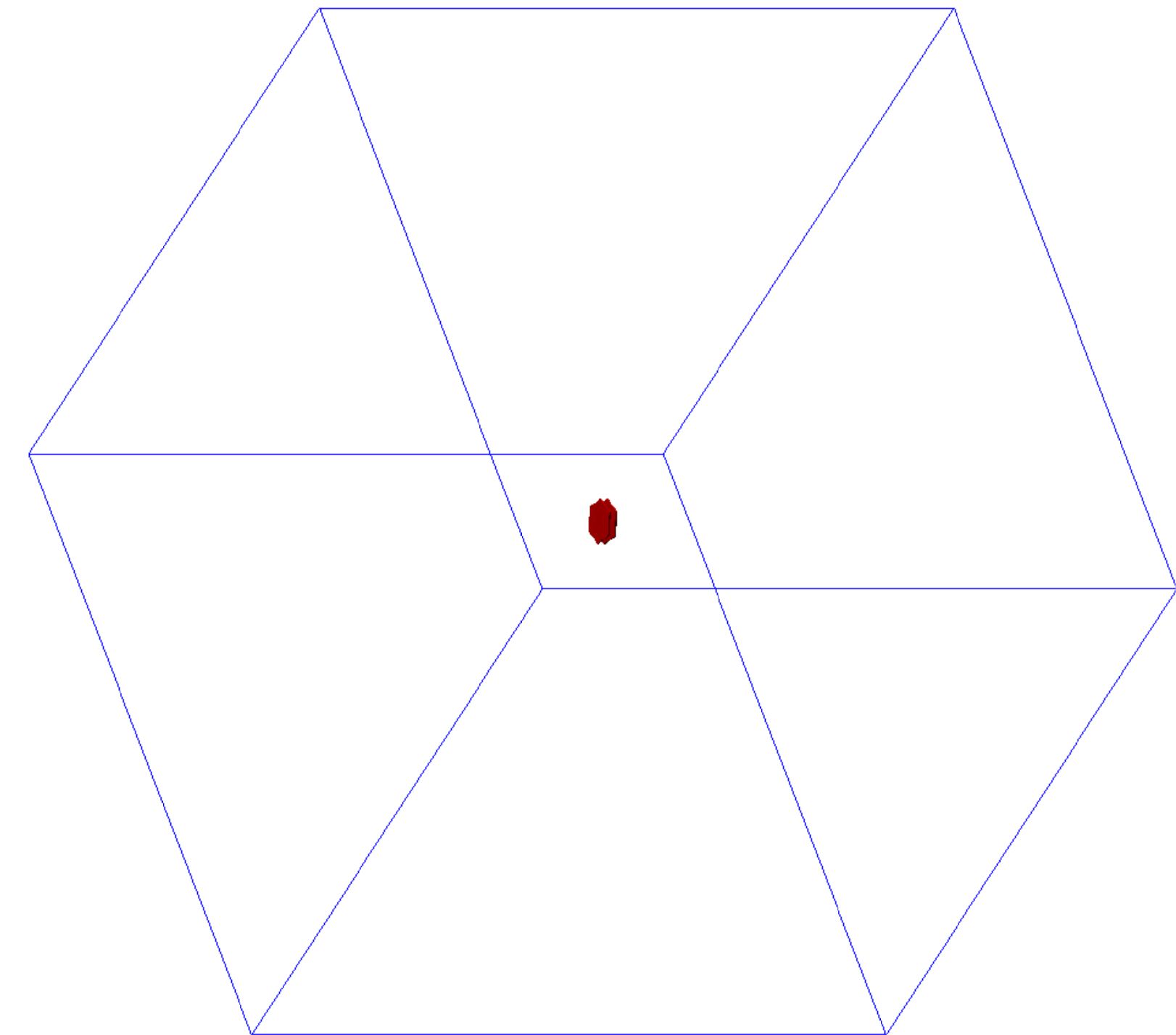}
\caption{($\Omega=(-4,4)^3$, $u_D = -0.02$, $\beta = \beta_{\rm flat,3}$)
$\vec{X}(t)$ for $t=0.05,\,0.1,\,0.2,\,0.3$; 
and $\vec X(0.3)$ within $\Omega$.
Parameters are $N_f=512$, $N_c=32$, $K^0_\Gamma = 1538$ 
and $\tau=5\times10^{-4}$.}
\label{fig:c02f6}
\end{figure} 

An experiment for $u_D = -0.02$ and $\beta = \beta_{\rm flat,3}$,
but for the anisotropy $\gamma$ defined by
\begin{equation} \label{eq:newhexgamma3d}
\gamma(\vec{p}) = 
2\,l_\epsilon(R_2(\tfrac{\pi}2)\,\vec{p}) + \sum_{\ell = 1}^3
l_\epsilon(R_1(\theta_0 + \tfrac{\ell\,\pi}3)\,\vec{p}).
\end{equation}
with $\epsilon=0.01$ and $\theta_0=\frac\pi{12}$
can be seen in Figure~\ref{fig:3dhex02_flat6}. This leads to a
geometrically more complicated breaking of the prismal facets. These can
also be observed in nature, and they are called hollow plates; see
\cite{Libbrecht05}. 
\begin{figure}[ht]
\center
 \includegraphics[angle=-90,totalheight=3.5cm]{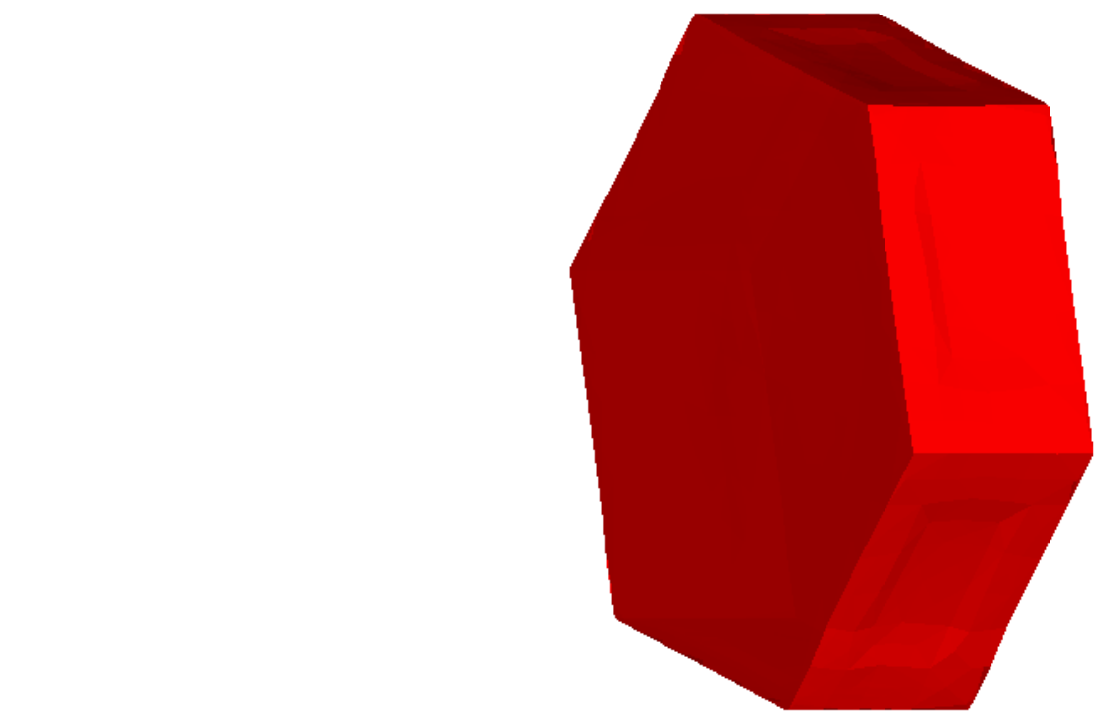}\qquad
 \includegraphics[angle=-90,totalheight=3.5cm]{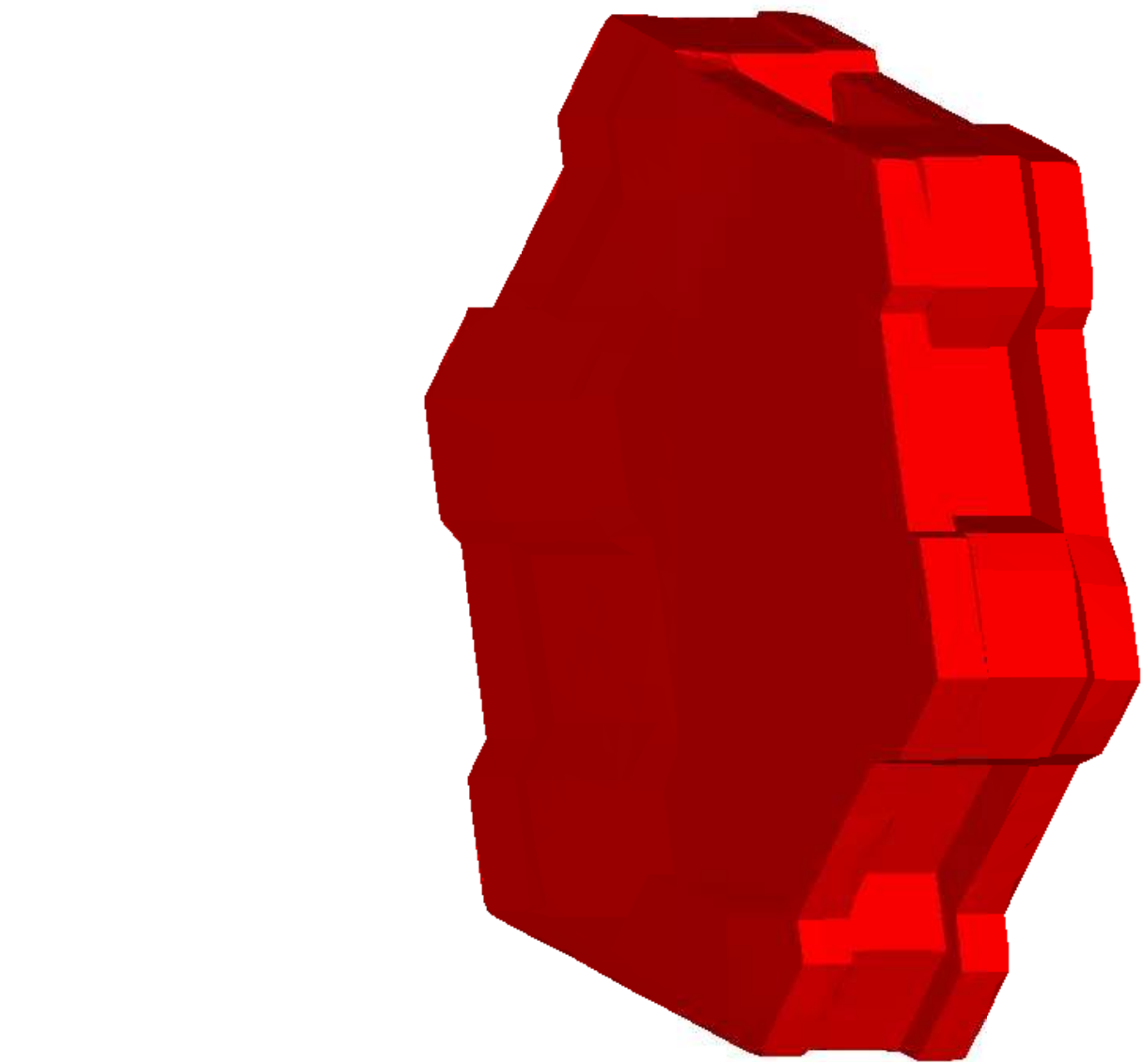}\qquad
 \includegraphics[angle=-90,totalheight=3.5cm]{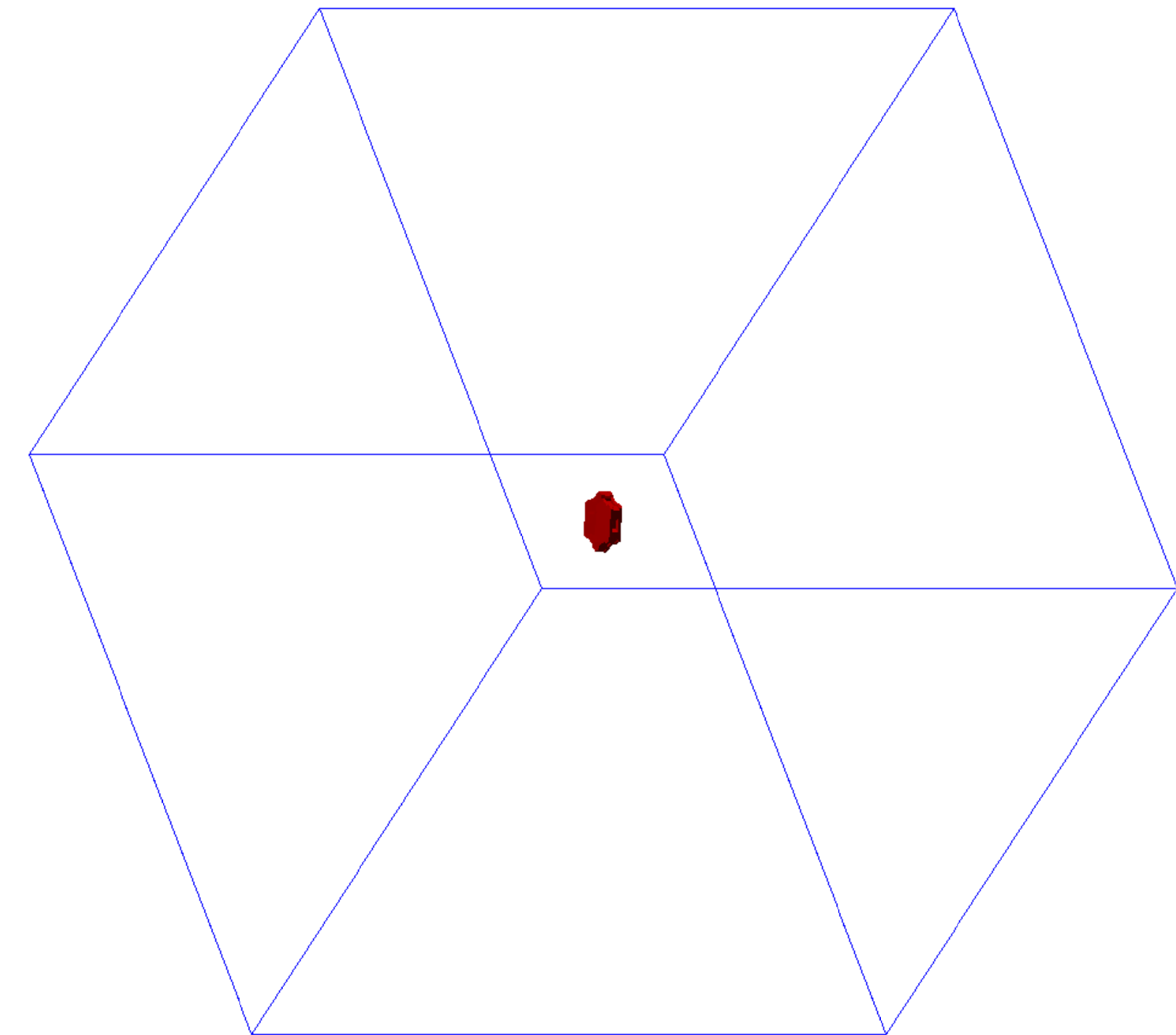}
\qquad
\caption{($\Omega=(-4,4)^3$, $u_D = -0.02$, 
$\gamma$ as in (\ref{eq:newhexgamma3d}),
$\beta = \beta_{\rm flat,3}$)
$\vec{X}(t)$ for $t=0.5,\,1$.
Parameters are $N_f=512$, $N_c=32$, $K^0_\Gamma = 98$, 
and $\tau=10^{-3}$.}
\label{fig:3dhex02_flat6}
\end{figure}

We also performed simulations varying $\beta$ in time. 
This is realistic as a
growing snow crystal falls to the earth through changing weather
conditions, which influence  the governing parameters, e.g.\ via the
temperature. In the first such example, we choose
\begin{subequations}
\begin{equation} \label{eq:beta_sw0}
\beta(\vec p) = \begin{cases}
\beta_{\rm flat,3}(\vec p) & t \in [0,30), \\
\beta_{\rm tall,3}(\vec p) & t \in [30,\infty) .
\end{cases}
\end{equation}
In a second example we choose
\begin{equation} \label{eq:beta_sw1}
\beta(\vec p) = \begin{cases}
\beta_{\rm flat,3}(\vec p) & t \in [0,20), \\
\beta_{\rm flat,1}(\vec p) & t \in [20,\infty) .
\end{cases}
\end{equation}
\end{subequations}
Results for these choices of $\beta$ and for $u_D = -0.004$ can be seen in
Figure~\ref{fig:22_switches}. The shapes in
Figure~\ref{fig:22_switches} can also be observed in nature, and they
are called scrolls on plates. 
\begin{figure}[ht]
\center
 \includegraphics[angle=-90,totalheight=2.7cm]{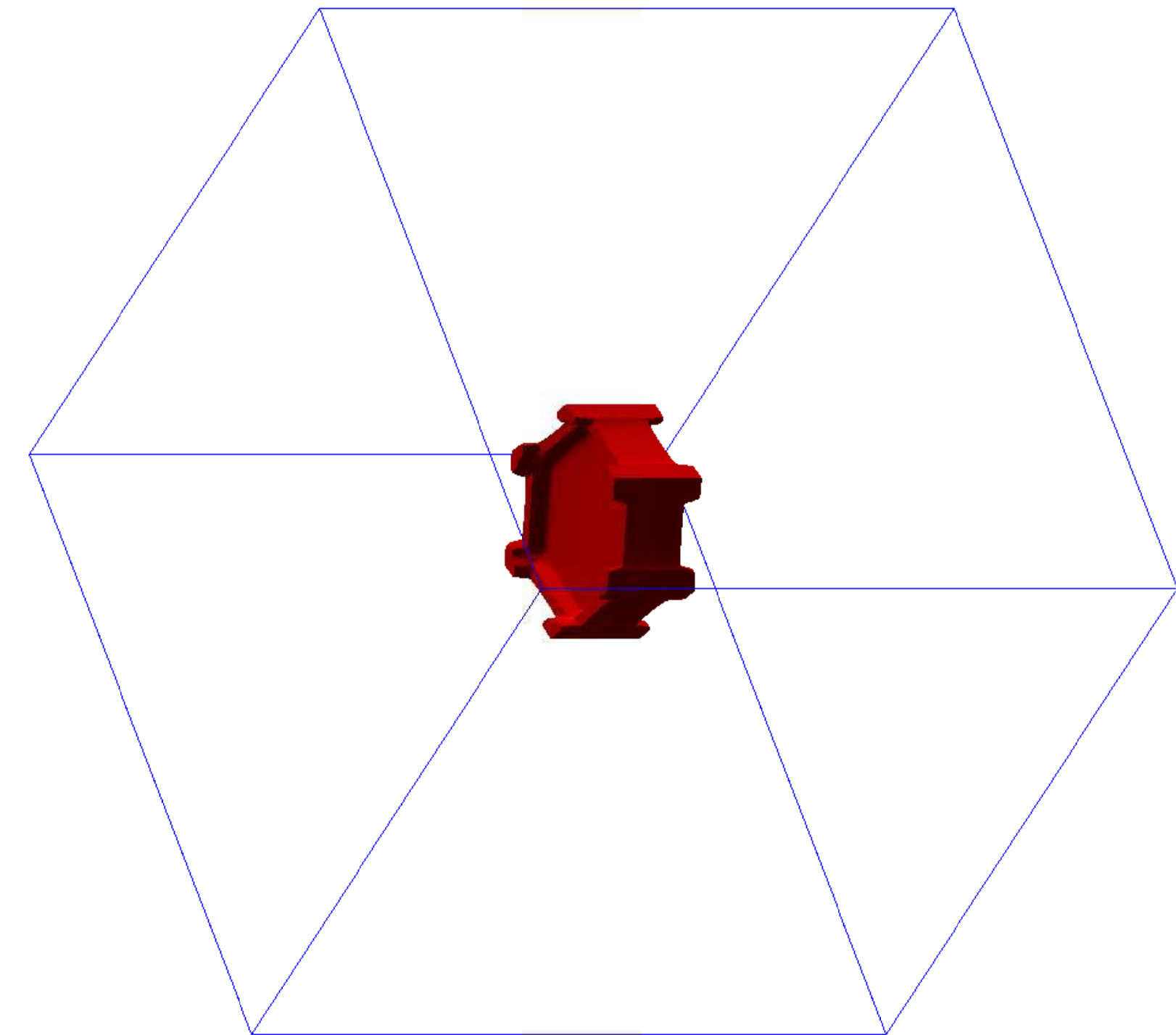}
\quad
 \includegraphics[angle=-90,totalheight=2.7cm]{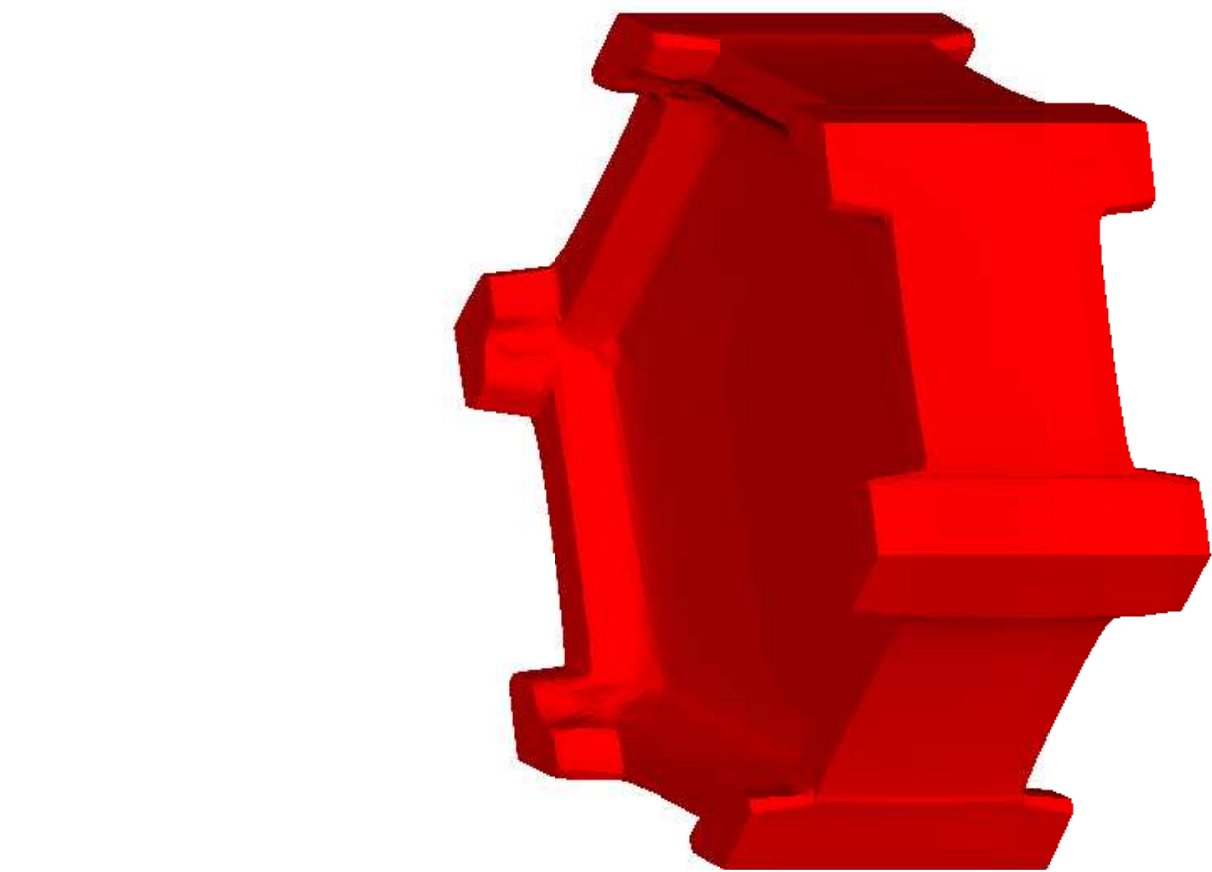}
\qquad\qquad
 \includegraphics[angle=-90,totalheight=2.7cm]{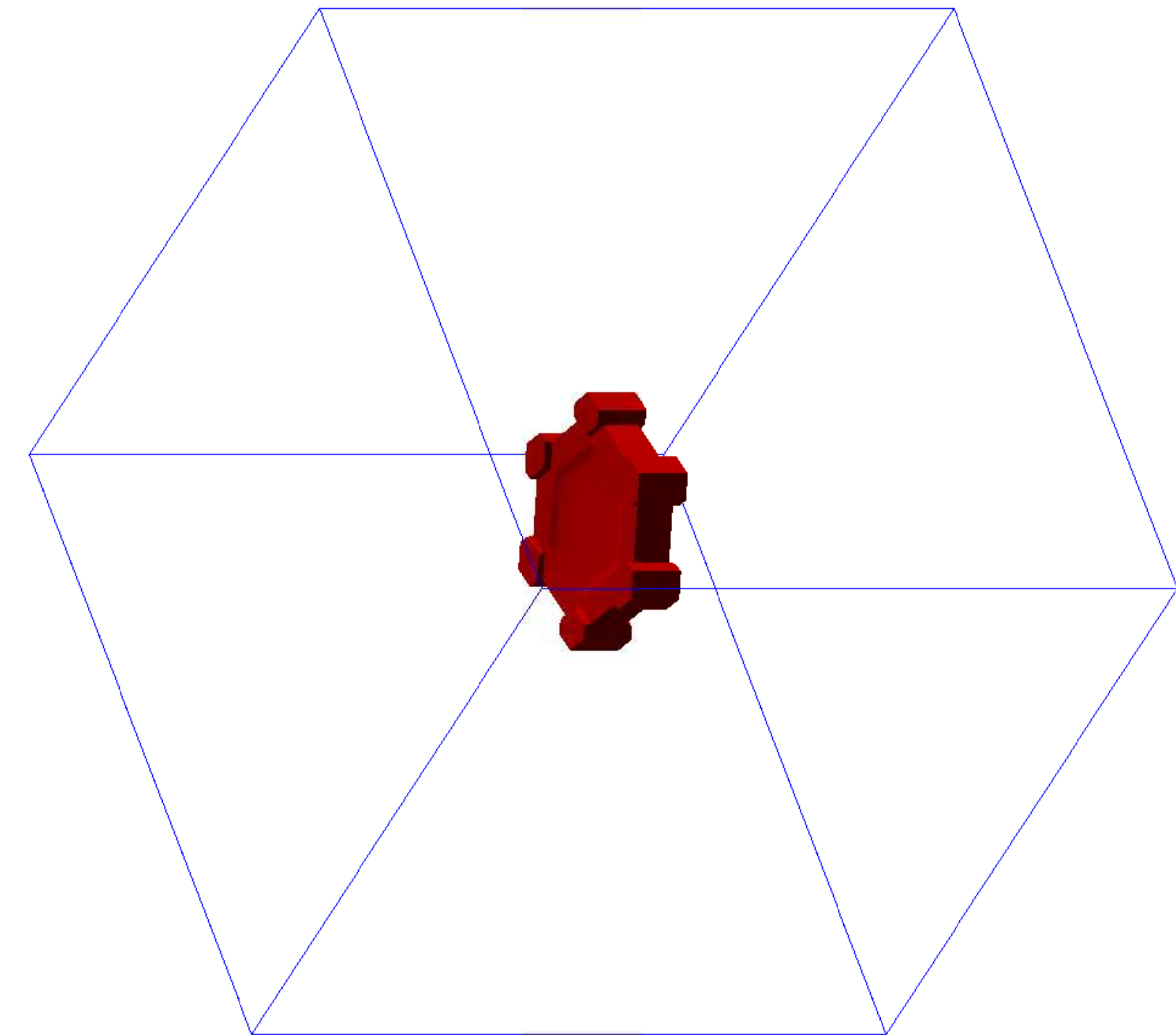}
\quad
 \includegraphics[angle=-90,totalheight=2.7cm]{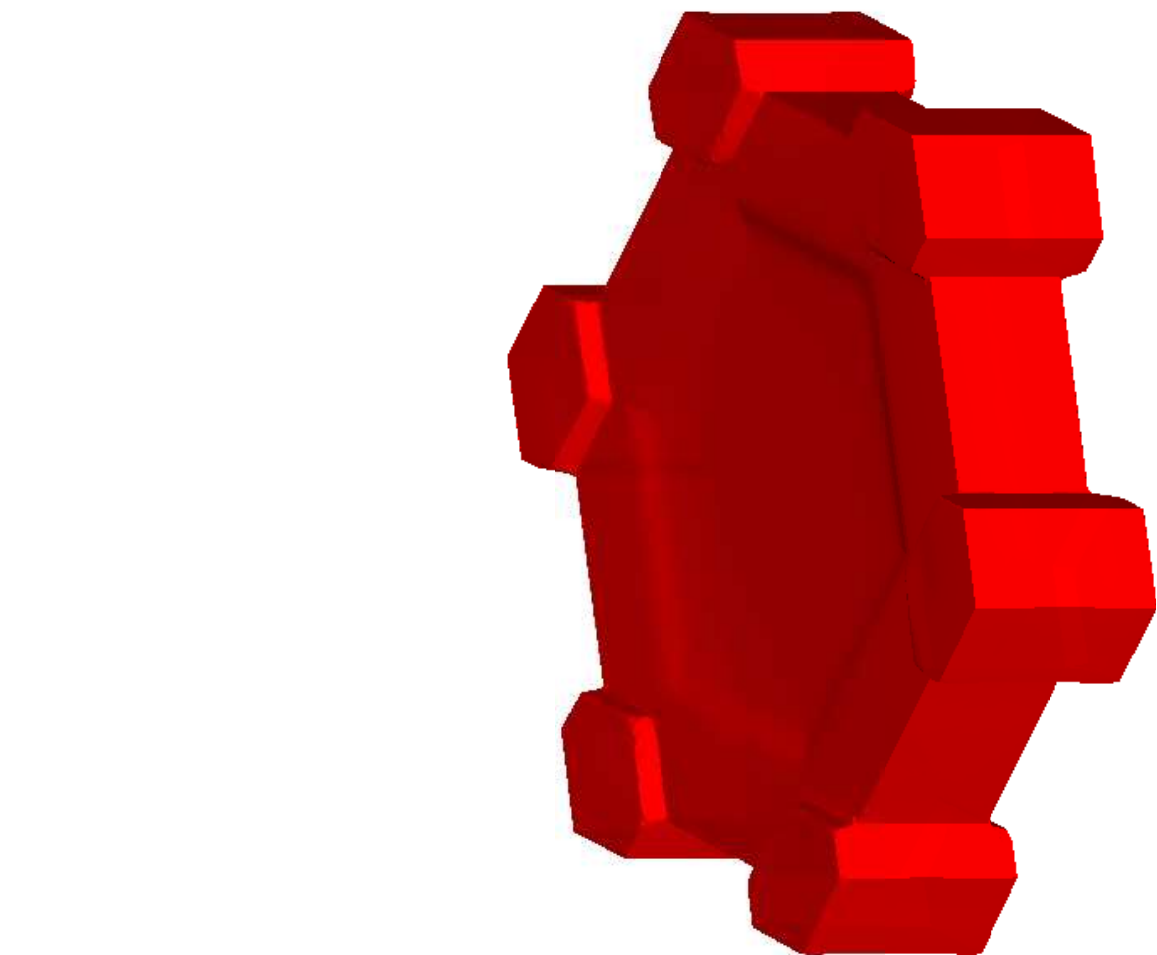}
\caption{\mbox{($\Omega=(-4,4)^3$, $u_D = -0.004$, 
$\beta$ as in (\ref{eq:beta_sw0},b))} 
$\vec{X}(T)$ for $T=50$.
Parameters are $N_f=128$, $N_c=16$, $K^0_\Gamma = 98$, 
and $\tau=10^{-1}$.}
\label{fig:22_switches}
\end{figure}

The remaining numerical experiments are for the cylindrical anisotropy 
(\ref{eq:giga}) with $\epsilon=10^{-2}$; recall Figure \ref{fig:frankgiga}.
The first case is for 
$\gamma_{\rm TB}= 1$, 
$u_D = -0.004$, and $\beta = \beta_{\rm tall,1}$, and the results, 
which show facet breaking both in the basal and
prismal directions, can be seen in Figure~\ref{fig:24r}.
Some plots of the concentration are shown in
Figures~\ref{fig:24rtemppng} and \ref{fig:24rtemppngzoom}, where
\begin{figure}[ht]
\center
 \includegraphics[angle=-90,totalheight=3.2cm]{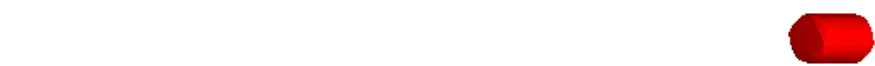} \quad
 \includegraphics[angle=-90,totalheight=3.2cm]{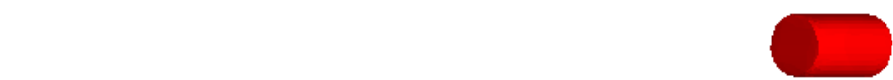} \quad
 \includegraphics[angle=-90,totalheight=3.2cm]{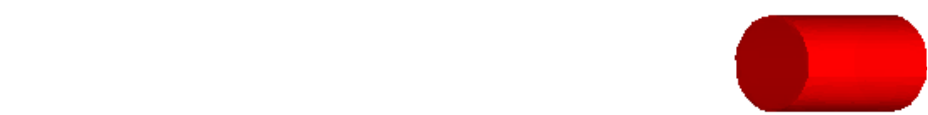} \quad
 \includegraphics[angle=-90,totalheight=3.2cm]{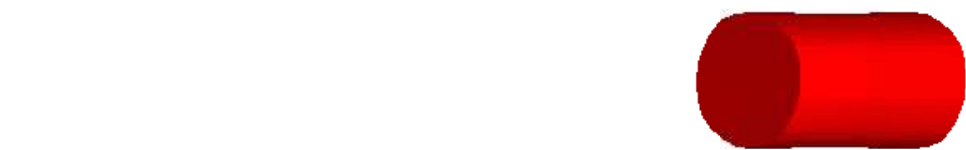} \quad
 \includegraphics[angle=-90,totalheight=3.2cm]{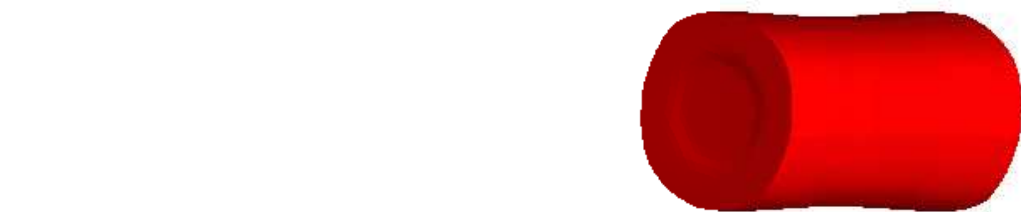} \quad
 \includegraphics[angle=-90,totalheight=3.2cm]{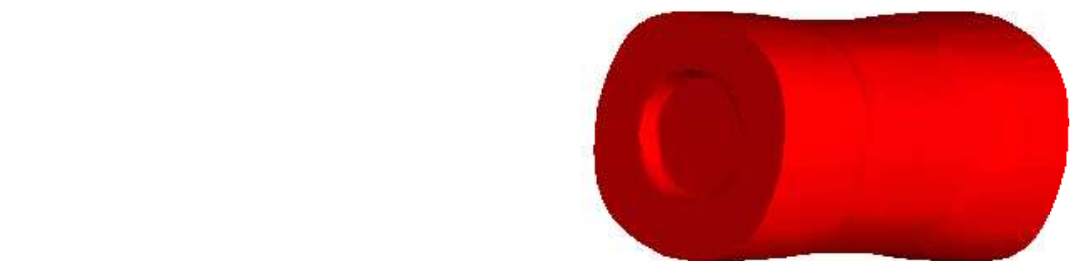} \quad
 \includegraphics[angle=-90,totalheight=3.2cm]{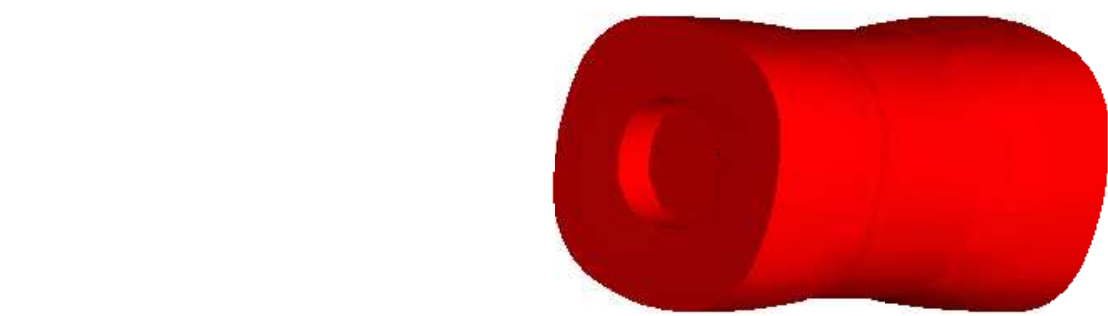} \quad
 \includegraphics[angle=-90,totalheight=3.2cm]{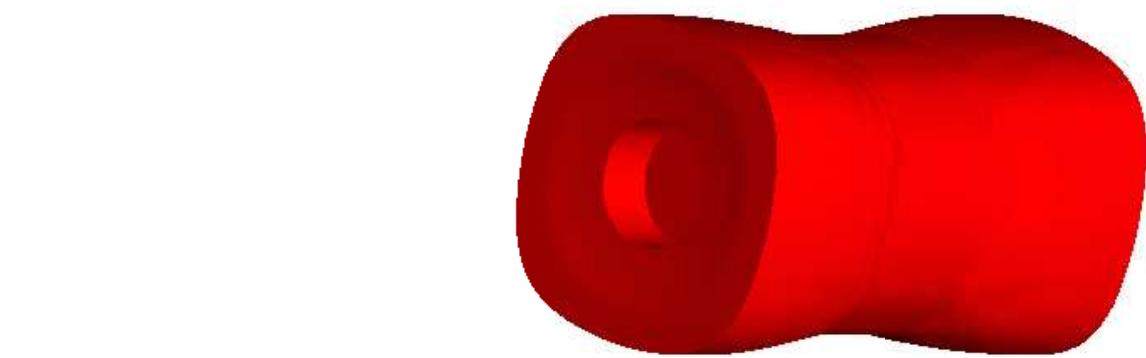} \\
 \includegraphics[angle=-90,totalheight=3.2cm]{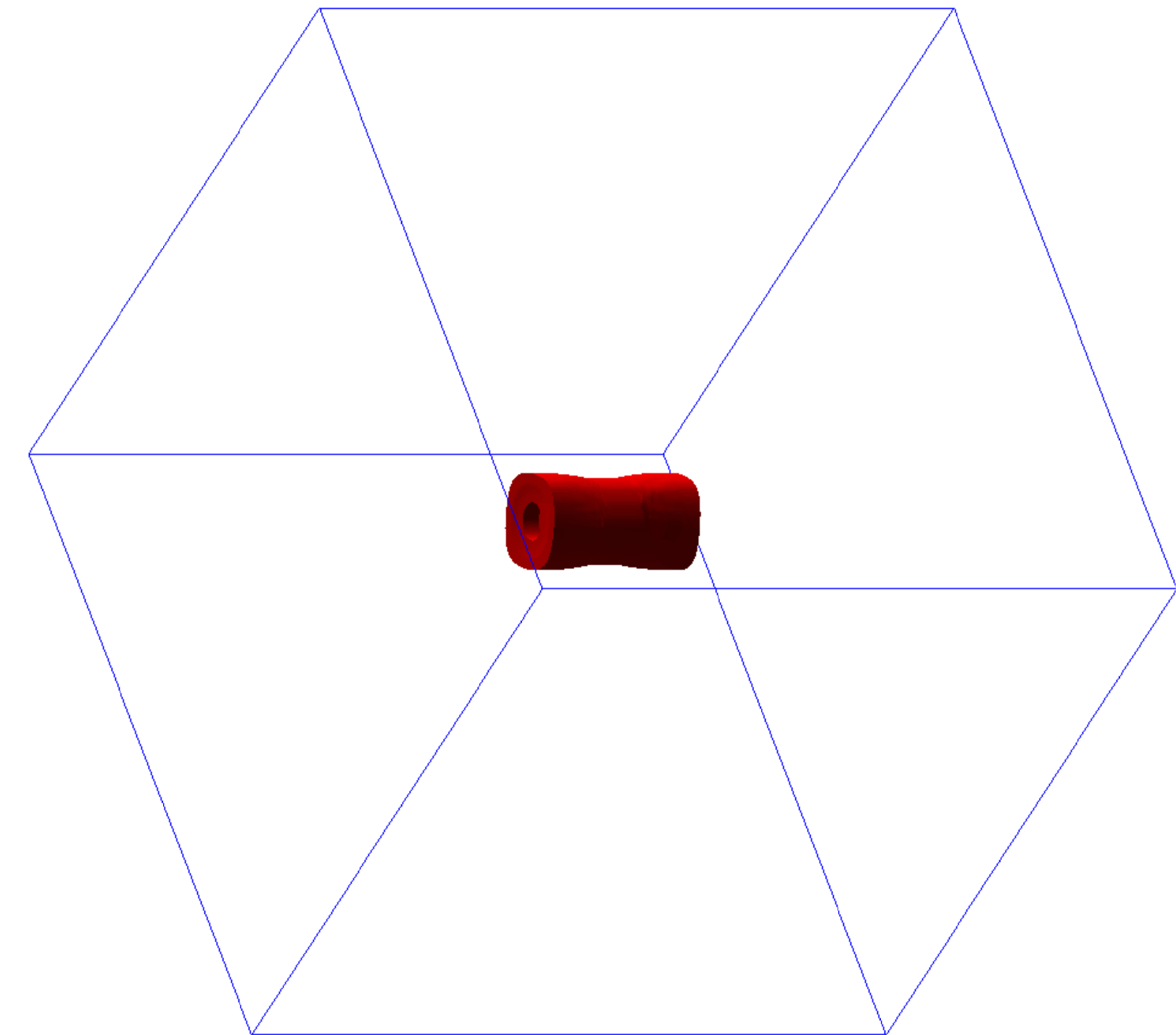}
\caption{($\Omega=(-4,4)^3$, $u_D = -0.004$, $\beta = \beta_{\rm tall,1}$)
$\vec{X}(t)$ for $t=1,\,2,\,5,\,10,\,20,\,30,\,40,\,50$; 
and $\vec X(50)$ within $\Omega$.
Parameters are $N_f=128$, $N_c=16$, $K^0_\Gamma = 98$, and $\tau=10^{-1}$.}
\label{fig:24r}
\end{figure} 
\begin{figure}[ht]
\center
 \includegraphics[angle=-0,totalheight=2.6cm]{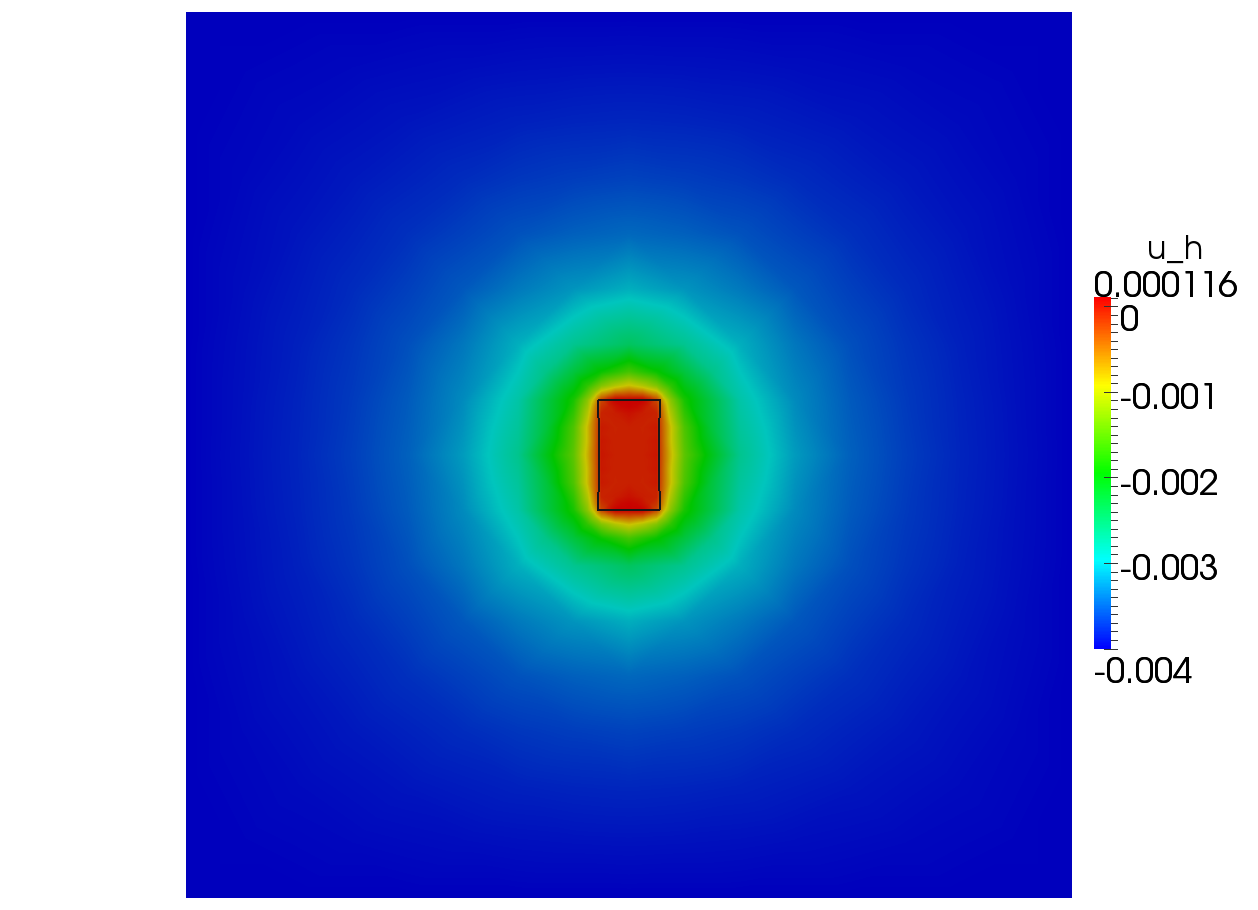}\quad
 \includegraphics[angle=-0,totalheight=2.6cm]{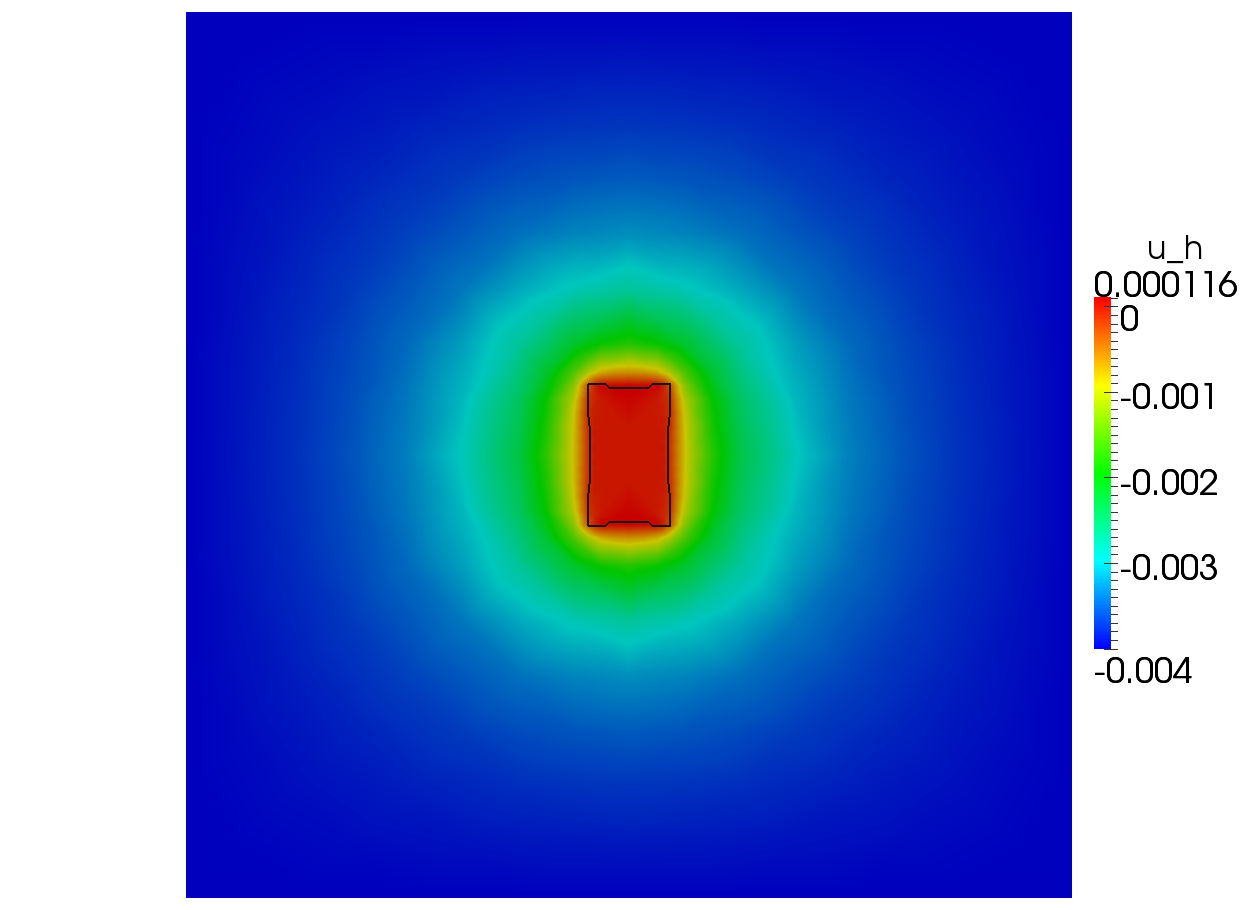}\quad
 \includegraphics[angle=-0,totalheight=2.6cm]{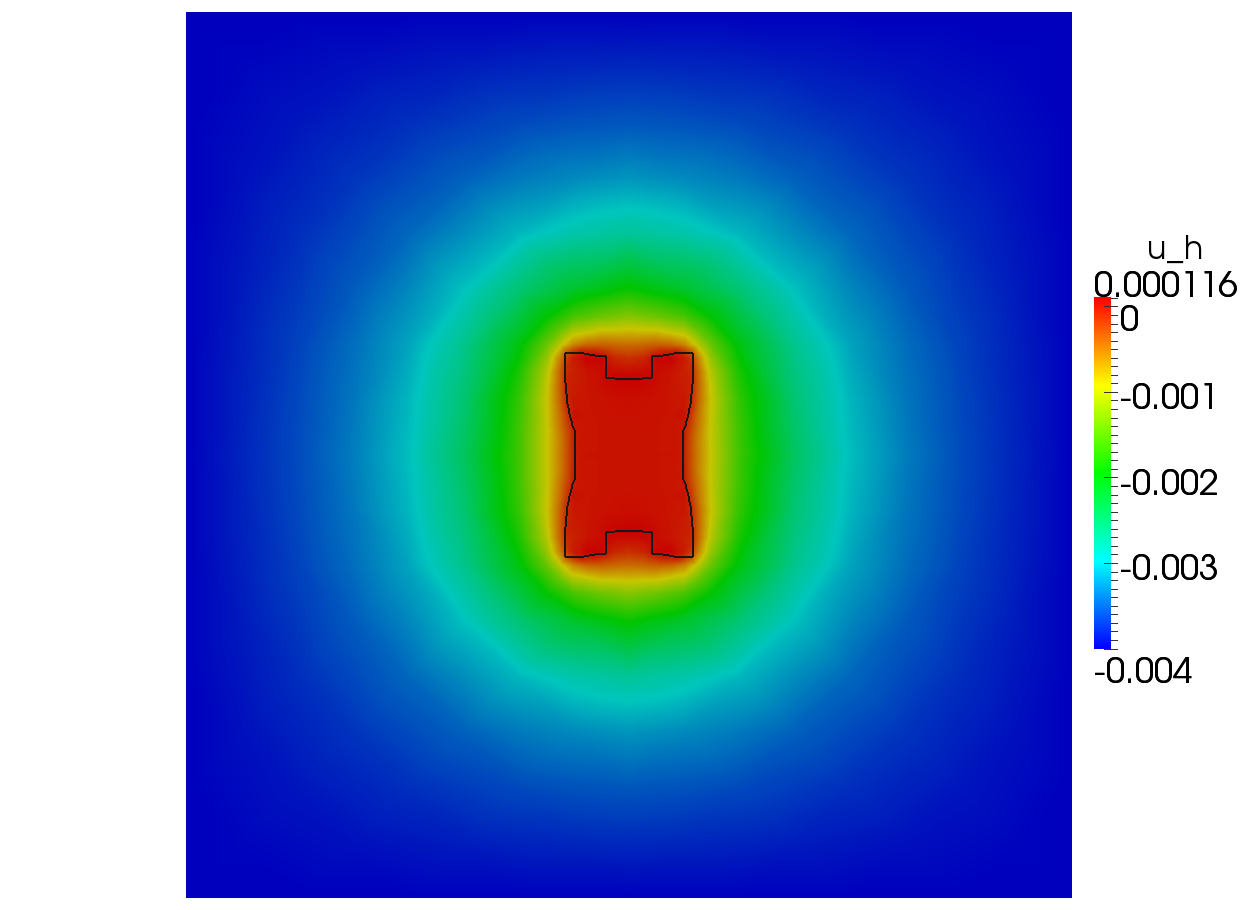} 
\caption{($\Omega=(-4,4)^3$, $u_D = -0.004$, $\beta = \beta_{\rm tall,1}$)
$\vec X(t) \cap \{\vec z : z_1 = 0\}$ 
and $U(t)\!\mid_{z_1=0}$ for $t=15,\, 25,\, 50$.
Parameters are $N_f=128$, $N_c=16$, $K^0_\Gamma = 98$, 
and $\tau=10^{-1}$.}
\label{fig:24rtemppng}
\end{figure} 
\begin{figure}[ht]
\center
 \includegraphics[angle=-0,totalheight=3.5cm]{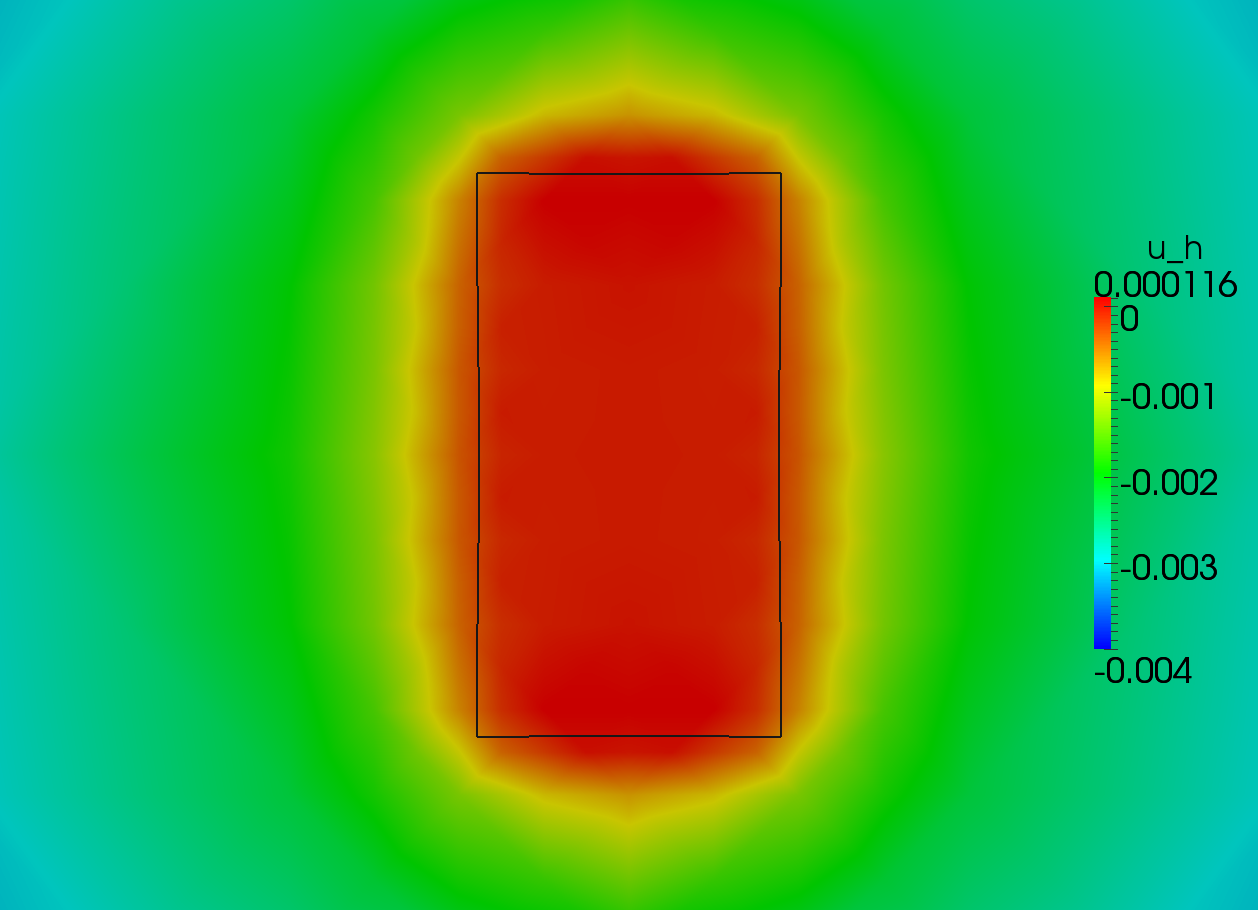}\qquad
 \includegraphics[angle=-0,totalheight=3.5cm]{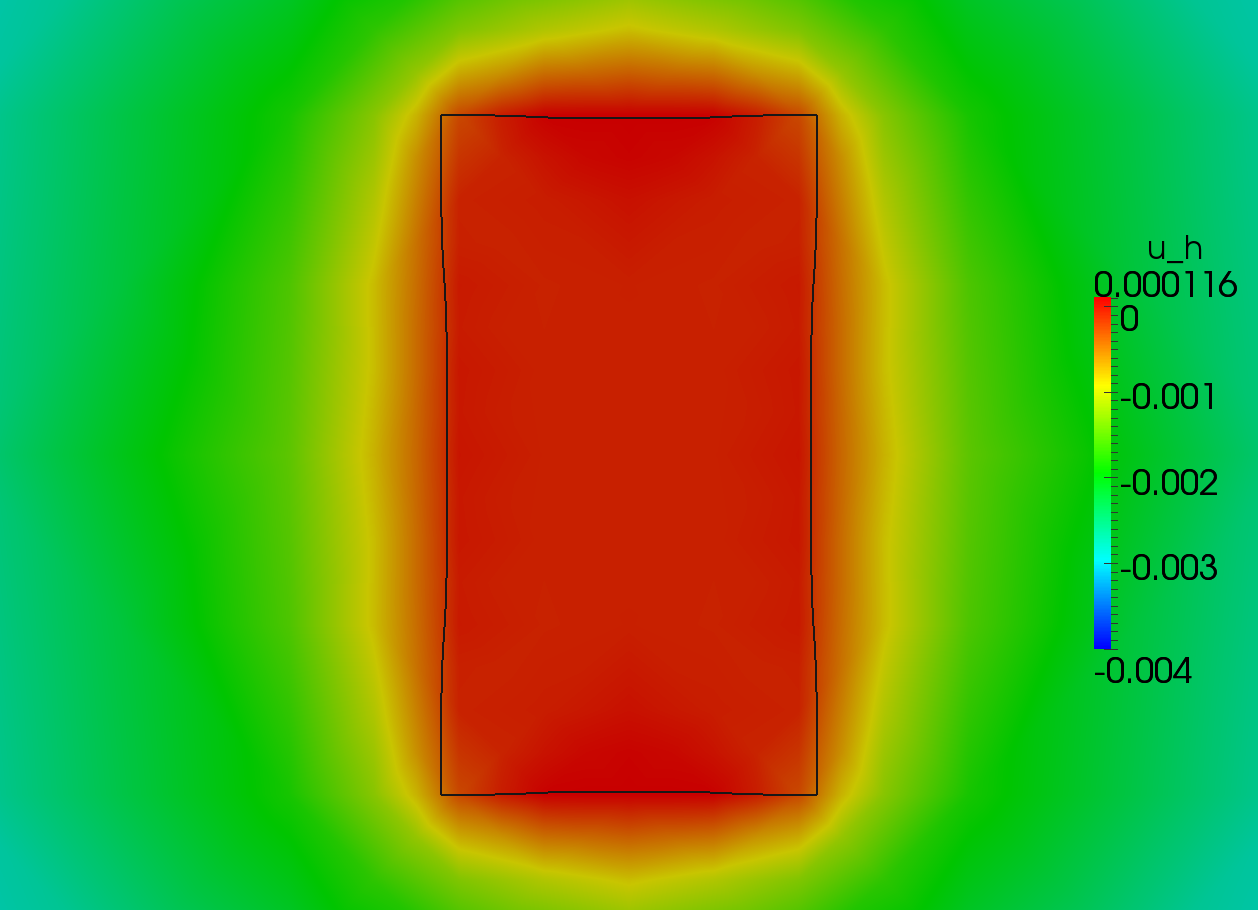}
\qquad
\caption{($\Omega=(-4,4)^3$, $u_D = -0.004$, $\beta = \beta_{\rm tall,1}$)
$\vec X(t) \cap \{\vec z : z_1 = 0\}$ and $U(t)\!\mid_{z_1=0}$ for 
$t=10,\,15$.
Parameters are $N_f=128$, $N_c=16$, $K^0_\Gamma = 98$, 
and $\tau=10^{-1}$.}
\label{fig:24rtemppngzoom}
\end{figure} 
Berg's effect  (see e.g.\ \cite{GigaR03})
can clearly be seen; i.e.,  $U$ increases towards
the centre of the basal face before facet breaking occurs.

For the anisotropy (\ref{eq:giga}) it is 
of interest to find for what value of $\gamma_{\rm TB}$ the evolution of
(\ref{eq:1a}--e) with
\begin{equation} \label{eq:rybkaparams}
\vartheta = 0, \quad \mathcal{K} =1, \quad \lambda = 1, \quad \rho = 1, 
\quad \alpha = 1, \quad a = 1,\quad \beta = \gamma, \quad f = 0
\end{equation}
is self-similar. For example, in \cite{GigaR04} it was shown that there exists
a value $\gamma_{\rm TB} > 0$ for which this is the case.
Numerically this can be checked by starting this flow with a
scaled Wulff shape (or a shape close to that), and then to observe whether the
height-to-basal-diameter ratio of the evolving approximate cylinder converges
to $\gamma_{\rm TB}$.

In practice we choose $\Gamma(0)$ to be a cylinder with basal radius $R_0 =
0.1$ and a height/basal diameter ratio of $\gamma_{\rm TB}$. 
In order to obtain the
desired sign for $\mathcal{V}$, i.e.,  for an expanding evolution, we set
$u_D = -21$ in (\ref{eq:1d}). For the domain $\Omega$ we choose
$\Omega = (-8, 8)^3$.

In practice we appear to obtain a value for self-similarity for 
some $\gamma_{\rm TB} \in [0.92,0.93]$, although the precise value seems to
depend on the resolution of the bulk mesh. In Figure~\ref{fig:rybka0925} 
we plot some results for an experiment with $\gamma_{\rm TB} = 0.925$, while in
Figure~\ref{fig:rybka_rho} we show the evolution of the ratio of interest for
two experiments with $\gamma_{\rm TB} = 0.92$ and $\gamma_{\rm TB} = 0.925$,
respectively. These results seem to indicate that there exists a value 
$\gamma_{\rm TB}$ close to $\gamma_{\rm TB} = 0.92$ for which the evolution of
(\ref{eq:1a}--e) with (\ref{eq:rybkaparams}) and (\ref{eq:giga})
is self-similar.

\begin{figure}[ht]
\center
 \includegraphics[angle=-90,totalheight=3.5cm]{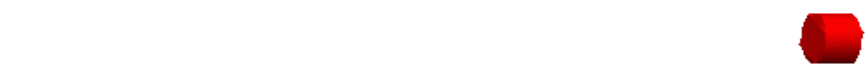} \qquad
 \includegraphics[angle=-90,totalheight=3.5cm]{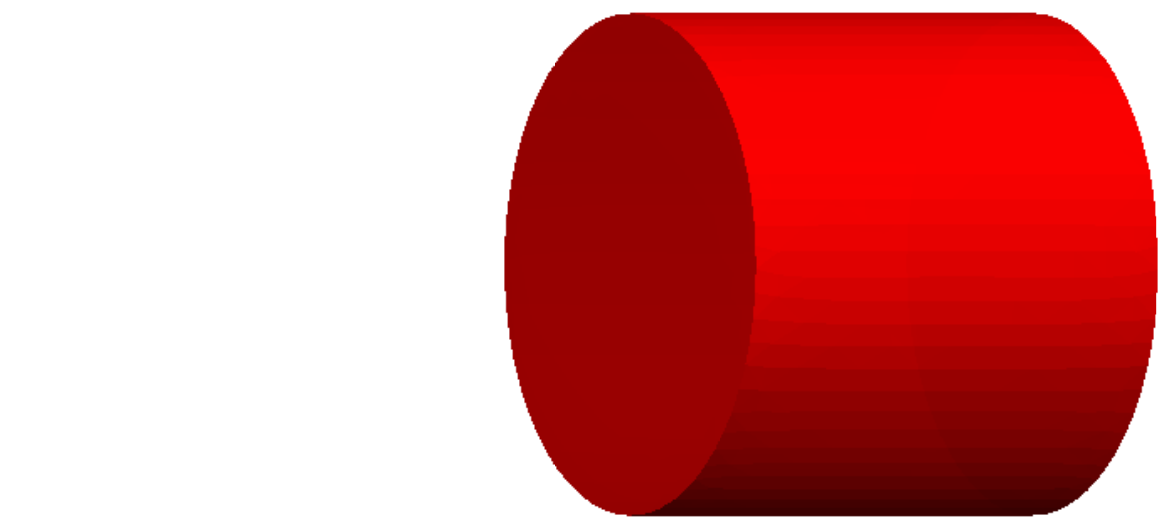} \qquad
 \includegraphics[angle=-90,totalheight=3.5cm]{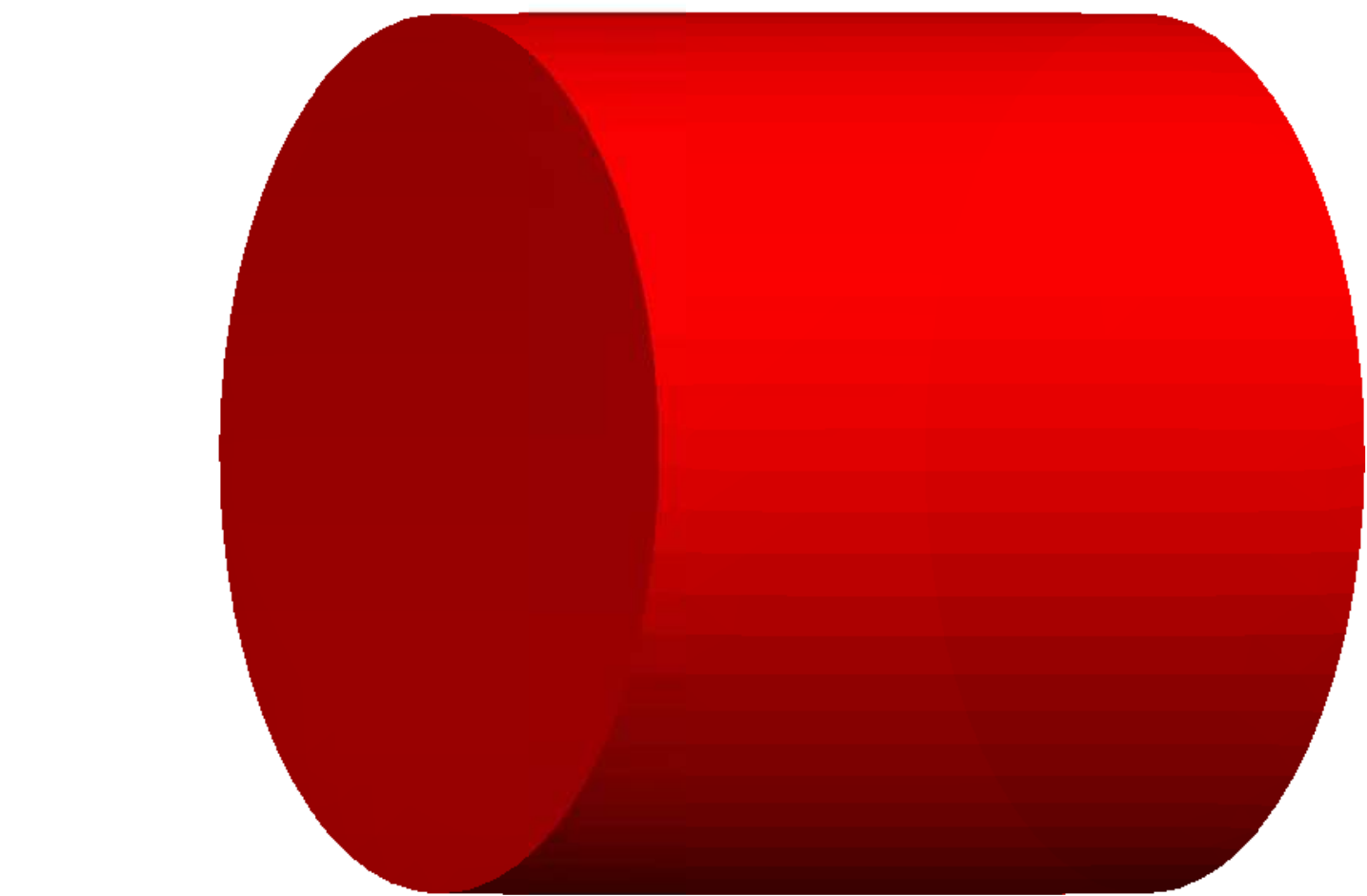} \qquad
 \includegraphics[angle=-90,totalheight=3.5cm]{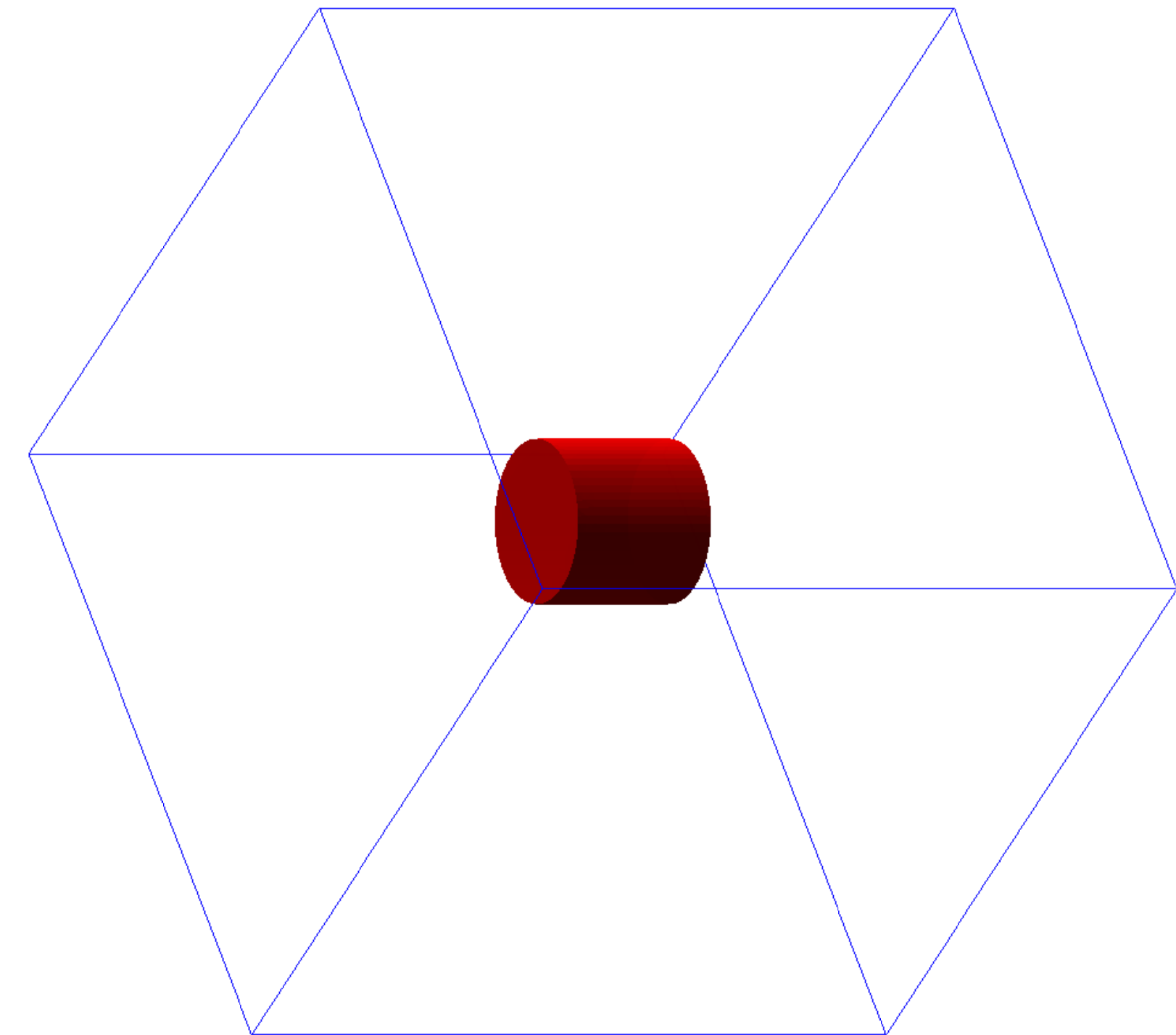} 
\caption{($\Omega=(-8,8)^3$, $\gamma_{\rm TB} = 0.925$)
$\vec{X}(t)$ for $t=0,\,0.1,\,0.2$; and $\vec X(0.2)$ within $\Omega$.
Parameters are $N_f=512$, $N_c=32$, $K^m_\Gamma \equiv 1538$, and 
$\tau=10^{-4}$.}
\label{fig:rybka0925}
\end{figure} 
\begin{figure}[ht]
\center
 \includegraphics[angle=-0,width=0.7\textwidth]{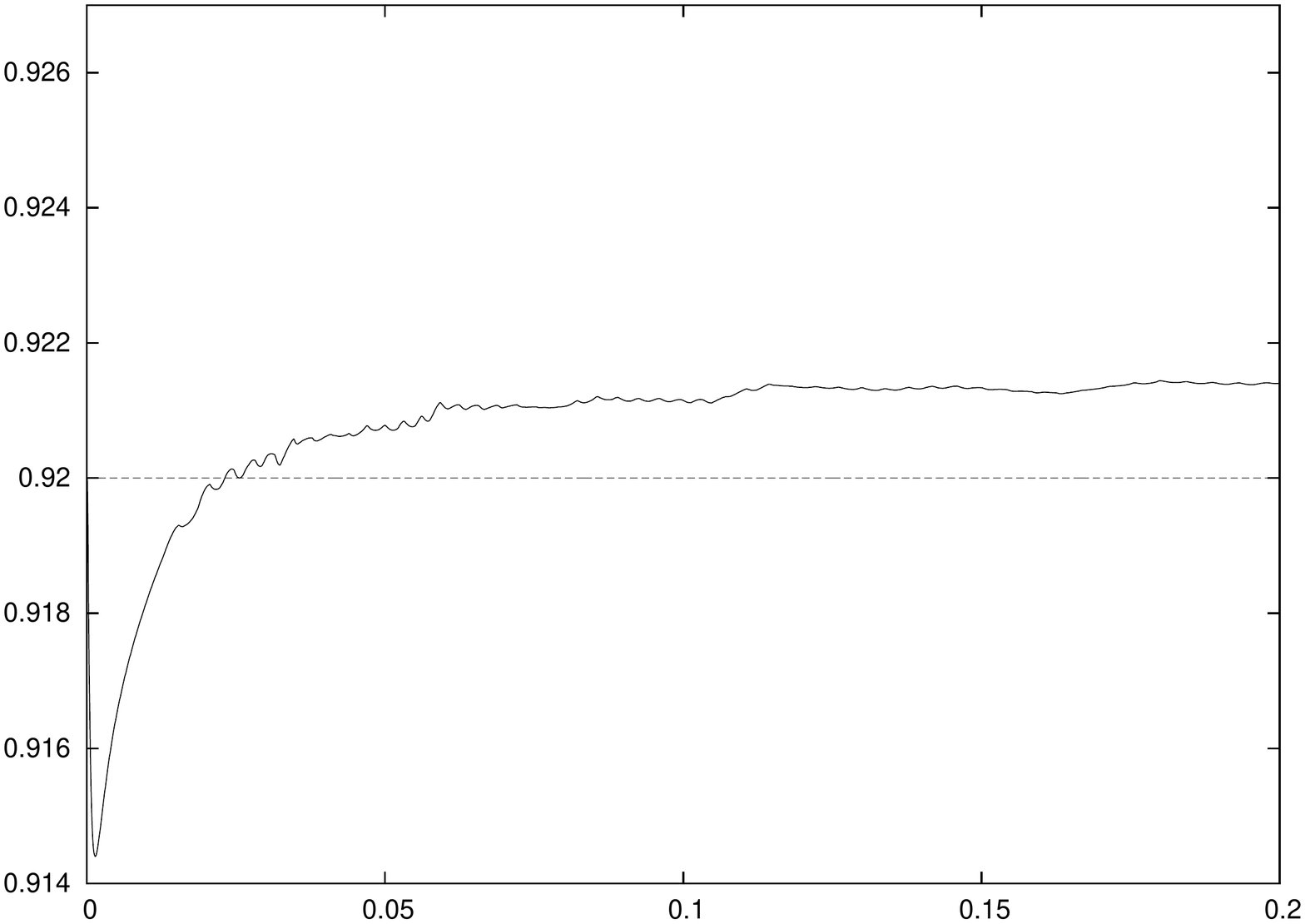}
 \includegraphics[angle=-0,width=0.7\textwidth]{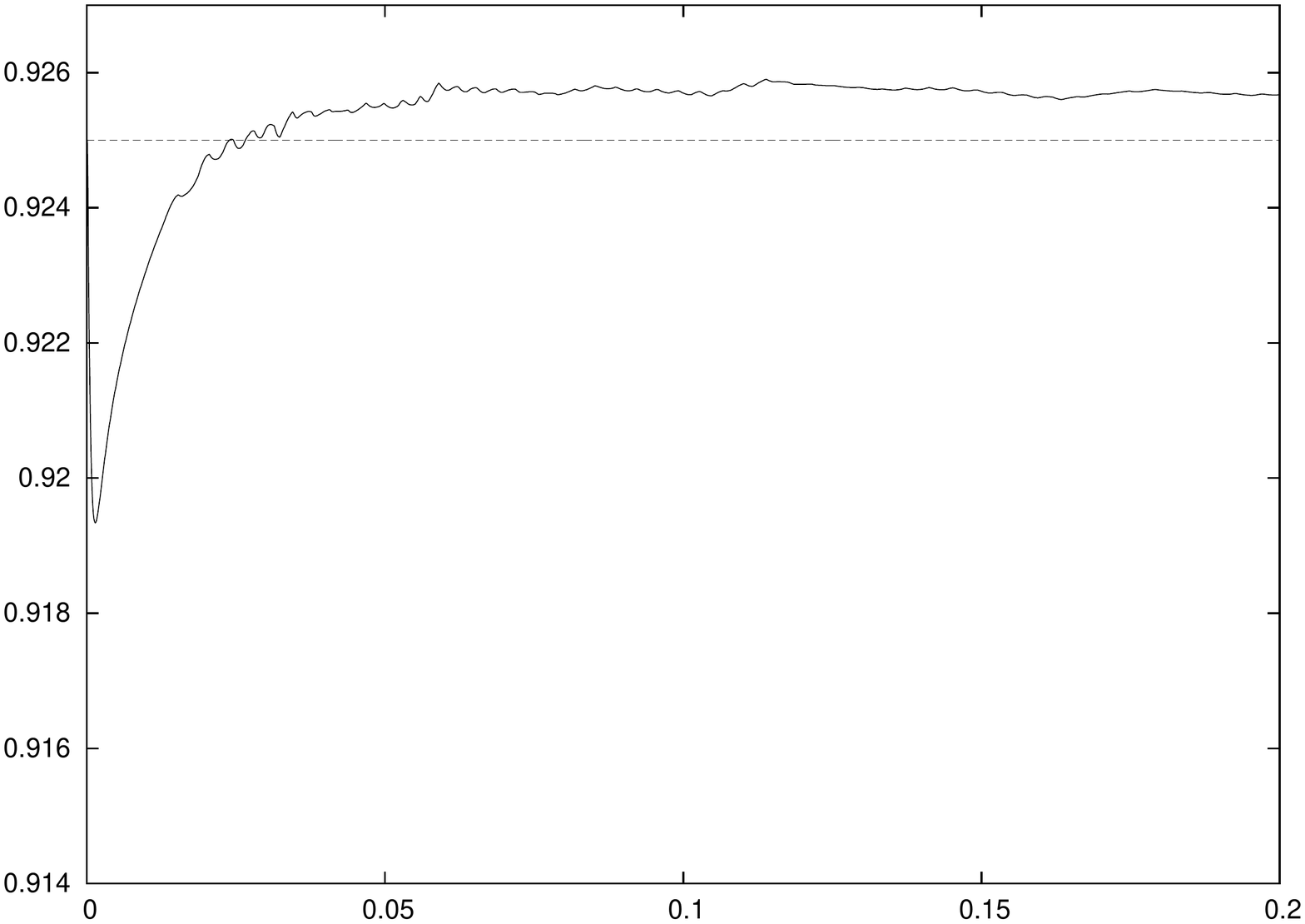}
\caption{Plots of the height/basal diameter ratio for the two runs with
$\gamma_{\rm TB} = 0.92$ (top) and 
$\gamma_{\rm TB} = 0.925$ (bottom). The dashed lines show the value of 
$\gamma_{\rm TB}$.
}
\label{fig:rybka_rho}
\end{figure} 

\section*{Conclusions}

We have presented a fully practical finite-element approximation for one-sided
Mullins--Sekerka and Stefan problems with
anisotropic Gibbs--Thomson law and kinetic undercooling. 
In particular, the method allows the
approximation of a continuum model for snow crystal growth, which is based on
rigorous thermodynamical principles and balance laws. To our knowledge, the
numerical results presented in this paper are the first simulations of snow
crystal growth that are based on such a rigorous, physically motivated model.

In our numerical simulations of snow crystal growth in three space dimensions,
we were able to produce a significant number of different types of snow
crystals. In particular (recall
Figure~\ref{fig:libbrecht}), we obtained results that resemble solid plates,
solid prisms, hollow columns,   
dendrites, capped columns, and
scrolls on plates. Also, facet breaking in the moving-boundary problems
computed have been observed in cases with nearly crystalline
anisotropic energies; see also \cite{GigaR06} for theoretical
predictions of facet breaking. 
We therefore believe that
the results presented here may help to understand the different
factors that play a role in the shaping of snow crystals in the real world.

Producing more complicated dendritic shapes
in three space dimensions, with complicated substructures
such as steps and ridges, as in e.g.\ \cite[Figure~1]{Libbrecht05}, or
as in the beautiful simulations in \cite{GravnerG09},
which were obtained with a cellular automata algorithm, 
would need a much higher computational cost when computed with the help
of a discretized moving-boundary problem for a diffusion 
equation.
The main reason is that the 
highly detailed and irregularly structured surface of snow flakes  (see e.g.\
Figure~1(c) in \cite{Libbrecht05}) would need to be accurately captured with
a triangulated surface $\Gamma^m$, say. On this surface, a second-order partial
differential equation then needs to be solved, which is coupled to a PDE in the
bulk. The necessary resolutions for both meshes, as well as the involved
computational effort to solve the linear systems arising from 
(\ref{eq:uHGa}--c), mean that on currently available computer hardware those
kind of computations cannot be performed.

Nevertheless, it is our belief that the numerical methods presented here, 
combined with suitable randomizations and fluctuations of physical parameters  
together with sophisticated computing equipment, should be able to produce all
the possible variations of realistic snow crystals.
In addition, we believe that the computations presented in this paper are
the most accurate and complex which have been computed so far with the
help of a Stefan or Mullins--Sekerka problem with hexagonal symmetry.

\smallskip

\noindent
{\bf Acknowledgment.} We are grateful to Prof. Libbrecht for allowing us
to use Figure~\ref{fig:libbrecht}.


\begin{thebibliography}{99}
\bibitem{Luckhaus90}
S.~Luckhaus, \textit{Solutions for the two-phase {S}tefan problem with the
  {G}ibbs--{T}homson law for the melting temperature}, European J. Appl. Math.,
  1  (1990), 101--111.
\end{thebibliography}
\end{document}